\documentclass[runningheads]{llncs}
\title{Pattern detection in ordered graphs}

\author{Guillaume Ducoffe\inst{1}\orcidID{0000-0003-2127-5989},\\
Laurent Feuilloley\inst{2}\orcidID{0000-0002-3994-0898}, \\
Michel Habib\inst{3}\orcidID{0000-0002-8564-2314},   \\
Fran\c{c}ois Pitois\inst{2,4}\orcidID{0000-0002-5524-0138}}

\authorrunning{G. Ducoffe, L. Feuilloley, M. Habib, F. Pitois}

\institute{
University of Bucharest, Romania, National Institute for Research and Development in Informatics, Romania \and
 Univ Lyon, CNRS, INSA Lyon, UCBL, LIRIS, UMR5205, F-69622 Villeurbanne, France \and
IRIF, CNRS \& Universit\'e Paris Cit\'e,  France \and 
LIB, Université Bourgogne Franche-Comté, France
}

\usepackage[utf8]{inputenc} 
\usepackage[english]{babel} 
\usepackage{thm-restate}

\makeatletter

\usepackage{subcaption} 
\usepackage{multicol}
\usepackage{rotating}
\usepackage{amssymb} 
\usepackage{tikz} 
\usepackage{hyperref}
\usepackage{algorithm2e}

\usepackage{enumitem}
\usepackage{amsmath} 

\spnewtheorem{hypothesis}{Hypothesis}{\bfseries}{\itshape}

\newcommand{\CC}{\mathcal{C}}

\newcommand{\edit}{dist_{out}}

\newcommand{\pnograph}{\textsc{No Graph}}
\newcommand{\ppath}{\textsc{Linear Forest}}
\newcommand{\pstar}{\textsc{Star}}
\newcommand{\pinterval}{\textsc{Interval}}
\newcommand{\psplit}{\textsc{Split}}
\newcommand{\pforest}{\textsc{Forest}}
\newcommand{\pbipartite}{\textsc{Bipartite}}
\newcommand{\pchordal}{\textsc{Chordal}}
\newcommand{\pcomparability}{\textsc{Comparability}}
\newcommand{\ptriangle}{\textsc{Triangle}}

\begin{document}

\maketitle{}

\begin{abstract}
A popular way to define or characterize graph classes is via forbidden subgraphs or forbidden minors. 
These characterizations play a key role in graph theory, but they rarely lead to efficient algorithms to recognize these classes.
In contrast, many essential graph classes can be recognized efficiently thanks to characterizations of the following form: there must exist an ordering of the vertices such that some ordered \emph{pattern} does not appear, where a pattern is basically an ordered subgraph. 
These pattern characterizations have been studied for decades, but there have been recent efforts to better understand them systematically.  
In this paper, we focus on a simple problem at the core of this topic: given an ordered graph of size $n$, how fast can we detect whether a fixed pattern of size $k$ is present? 

Following the literature on graph classes recognition, we first look for patterns that can be detected in linear time. 
We prove, among other results, that almost all patterns on three vertices (which capture many interesting classes, such as interval, chordal, split, bipartite, and comparability graphs) fall in this category.
Then, in a finer-grained complexity perspective, we prove conditional lower bounds for this problem. In particular we show that for a large family of patterns on four vertices it is unlikely that subquadratic algorithm exist. 
Finally, we define a parameter for patterns, the \emph{merge-width}, and prove that for patterns of merge-width $t$, one can solve the problem in $O(n^{ct})$ for some constant~$c$. As a corollary, we get that detecting outerplanar patterns and other classes of patterns can be done in time independent of the size of the pattern.
\end{abstract}


\section{Introduction}

\subsection{Motivation}

\paragraph{Forbidden structures characterization and recognition.}
The most popular way to define graph classes is to forbid some substructure, typically (induced) subgraphs or minors. 
This approach has led to a huge amount of results on the structure of these graphs and how to use it \emph{e.g.} for combinatorial optimization. 
It is also a versatile framework, in the sense that many useful classes can be described by forbidden subgraphs or minors. 
For example, \cite{LekkeikerkerB62} established a list of forbidden induced subgraphs characterizing interval graphs, a well-known class, whose original definition does not refer to any forbidden structure.

But when it comes to recognizing graph classes efficiently, forbidden subgraphs become less handy. First, the list of forbidden structures might be long or infinite. For example, the list for interval graphs contains three infinite families.
Second, detecting induced subgraphs efficiently can be challenging.

\paragraph{Efficient recognition and forbidden patterns.}
The classic efficient approach to recognize classes such as interval or chordal graphs is not to use forbidden subgraphs, but instead to use some specific graph traversals to unveil the inner structure of these graphs.
For example, chordal graphs can be recognized in a two-step way: first perform a specific breadth-first search (BFS) called LexBFS, and then check that the reverse order in which the vertices are visited has the property that for any vertex, its neighborhood restricted to its successors is a clique. 
The key principle behind the efficiency of these recognition algorithms is that \emph{many graphs classes are characterized by the existence of an ordering of their vertices with efficiently checkable properties}. 
For example, a graph is an interval graph, if and only if, there exists an ordering of its vertices such that for any three (not necessarily consecutive) vertices $u<v<w$, if $(u,w)$ is an edge, then $(u,v)$ must also be an edge. A compact way to encode this forbidden structure is to use what is called a \emph{pattern}~\cite{FeuilloleyH21}, see Figure~\ref{fig:interval} and its caption.

\begin{figure}[!h]
\vspace{-0.4cm}            
    \centering
    \begin{tikzpicture}
			[scale=1.5,auto=left,every node/.style=		
			{circle,draw,fill=black!5}]
			\node (a) at (0,0) {1};
			\node (b) at (1,0) {2};
			\node (c) at (2,0) {3};
			\draw (a) to[bend left=50] (c);
			\draw[dashed] (a) to (b);
		\end{tikzpicture}
    \caption{This picture illustrates the pattern associated with interval graphs. There is a \emph{realization} of this pattern in an ordered graph if there exists three (not necessarily consecutive) vertices $v_1<v_2<v_3$ such that $(v_1,v_3)\in E$ and $(v_1,v_2)\notin E$. In other words, for the pattern to be realized, the plain edges (called \emph{mandatory edges}) should be present, the dashed  edges (called \emph{forbidden edges}) should be absent, and there is no constraints on the edges that are not drawn (called \emph{undecided pairs}). A graph is an interval graph if and only if there exists a vertex ordering where there is no realization of the pattern. }
    \label{fig:interval}
\end{figure}
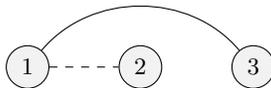
\vspace{-0.5cm}

Characterizations by forbidden patterns have been established for many classes, such as chordal graphs, interval graphs, co-comparability graphs, trivially perfect graphs, etc. 
In~\cite{FeuilloleyH21} the authors studied systematically all the classes that can be characterized by forbidden patterns on three vertices, and proved that they basically all correspond to well-known classes, and can almost all be recognized in linear time.

\paragraph{Pattern detection in ordered graphs.}
The typical recognition algorithms using forbidden pattern characterization usually have two components: they first build a candidate order, and then check that the pattern does not appear. 
At first sight, it seems that the first part should be the challenging one, and should dictate the complexity.
Indeed, for a graph on $n$ vertices, brute-force generation of all vertex orderings takes $O(n!)$, while for a pattern on $k$ vertices, brute-force checking takes at most $O(n^k)$. 
But, as hinted before, for many graph classes, there exists an efficient way to compute a unique candidate ordering via a set of traversals, leading to a polynomial-time algorithm. 
Then, for the recognition problem, the next step is to optimize the exponent of the polynomial, and it is not clear which part, building  the order or verifying the order, is the costliest.

In this paper, we focus on the verification part: we study the fine-grained complexity of detecting patterns in ordered graphs. 
A selection of questions we want to answer is as follows: Given a pattern on $k$ vertices, can we detect it in time much lower than~$n^k$? Independent of~$k$? In linear time? Are there some parameters of the pattern that dictate the complexity? What are the patterns that are hard to detect? 

Although our original motivations are on the side of understanding recognition algorithms and pattern characterizations better, our study is also close to the area that studies the efficient detection of subgraphs in (unordered) graphs, \emph{i.e.} subgraph isomorphism. 
The fine-grained complexity of this problem has been studied extensively recently (see \emph{e.g.} \cite{DaW22} and the references therein), and we hope that the new questions and techniques we introduce here can also inform this more established area. The other direction is already fruitful: we use some of their techniques and insights.  

Finally, this work also continues a line of work called \emph{certifying algorithms}~\cite{McConnellMNS11}. 
In this type of algorithms, not only does the algorithm compute the answer, but also a \emph{certificate} that allows to check quickly that the solution is correct. 
The motivation being that the code or the execution of the algorithm could be faulty. 
Typical examples are certifying bipartiteness and planarity, by a 2-coloring and a planar embedding, respectively. 
In the case of pattern characterization, the ordering avoiding the pattern is a certificate, and in this paper we are interested in the complexity of its \emph{authentication}.

\subsection{Organization of the paper and overview of the results}

Because of the page limit and to improve readability, the paper is mainly an overview of our results and technique. Most of the precise formal statements and proofs are deferred to the appendix. The overview is divided in three sections, that correspond to the three types of results obtained. 
In Section~\ref{sec:o-linear}, following the tradition of the graph class recognition literature, we first establish that several interesting types of patterns can be detected in time linear in the number of edges of the ordered graph: most patterns on three vertices, patterns that are basically forests, some patterns arising from geometry, etc. 
Then, in Section~\ref{sec:o-lower-bounds}, we prove that the problem is at most as hard as detecting cliques in ordered graphs, and then using fine-grained complexity techniques, we prove lower bounds; in particular, we establish that for a large family of patterns on four vertices, it is unlikely that one could do better than quadratic time. 
Finally, in Section~\ref{sec:o-parametrized}, we identify a parameter that we call \emph{merge-width}, such that the complexity of recognizing a pattern is in $O(n^{ct})$, for mergewidth $t$ and a constant $c$ independent from $t$.
We prove that this parameter is bounded for various types of patterns related to outerplanar graphs.

\section{Model, definitions, and basic properties}
\label{sec:model}

We first define formally patterns and the related concepts. We assume that the reader is familiar with basic graph notions. 


\begin{definition}
An ordered graph $(G,\tau)$ is an undirected, finite, simple, and loopless graph, denoted by a pair $(V(G),E(G))$, where $V(G)$ is the vertex set, totally ordered by $\tau$, and $E(G)$ is the edge set. 
The number of vertices will be denoted by $n$ and the number of edges by $m$.
\end{definition}

The ordered graphs given as input are encoded by adjacency lists. 

\begin{definition}
\label{def:pattern}
A \emph{pattern} $P$ is a $4$-tuple $(V(P),M(P),F(P),U(P))$ where $V(P)$ is a totally ordered set of vertices, and every unordered pair of vertices belongs to one of the three disjoint sets $M(P)$, $F(P)$, and $U(P)$ called respectively \emph{mandatory edges}, \emph{forbidden edges} and \emph{undecided pairs}. 
\end{definition}

In the figures (\emph{e.g.} in Figure~\ref{fig:interval}), we use plain edges for mandatory edges, dashed edges for  forbidden edges, and nothing for undecided pairs. Also, in our drawings, the patterns are ordered from left to right.
The \emph{complement} of a pattern is the pattern on the same vertex set where we have exchanged $M(P)$ and $F(P)$. The \emph{mirror} of a pattern is the same pattern with the reverse ordering. 
A pattern is \emph{fully specified} if it has no undecided pair, and it is \emph{positive} if it has no forbidden edge.

Note that ``\emph{pattern}" has multiple definitions in the literature, but in the paper it always refers to Definition~\ref{def:pattern}, following~\cite{FeuilloleyH21}.

\begin{definition}
    An ordered graph $G$ is a \emph{realization} of a pattern $P$ if $V(G)=V(P)$ and $E(G)=M(P)\cup U'$, where $U'\subset U(P)$.
    An ordered graph $G$ \emph{contains} a pattern $P$ if there exists a set of vertices $X$, such that the ordered subgraph induced by $X$ is a realization of $P$. If an ordered graph $G$ does not contain a pattern $P$, then it \emph{avoids} it.
\end{definition}


A simpler formalism with only two types of edges can also be considered, but it appears that the pattern formalism is better when handling classic graph classes. (See also the discussion in~\cite{FeuilloleyH21}.)
Patterns can also be described as acyclic signed graphs with 2 types of arcs (mandatory edges and forbidden edges).


Our focus is on the following algorithmic problem.

\medskip

\noindent Problem: {\sc Detection of pattern $P$} (where $P$ has $k$ vertices).\\
Input: A graph $G$ on $n$ vertices and $m$ edges, described by its adjacency list, and a total ordering $\tau$ of its vertices. \\
Output: Does $(G,\tau)$ contain $P$?

\medskip

For short, we will sometimes refer to the problem as $P$-\textsc{Detection}. Often, we will not refer explicitly to $\tau$, and simply consider that the vertices are named from 1 to $n$ in the ordering of $\tau$.
We will also mention the generalization where we want to detect a set of patterns, in this case we want to decide whether the ordered graph contains \emph{at least one} of the patterns of the set.

Note that detecting a pattern or its mirror are two problems with the same complexity. Also, if we allow quadratic time, then detecting a pattern or its complement are also equivalent, since we can complement the ordered graph.

\section{Linear-time detection (overview)}
\label{sec:o-linear}

Our first motivation for the study of pattern detection originates from graph classes recognition, and in particular linear-time recognition. For this reason, our first focus in on the question: What are the patterns that can be recognized in linear-time? 
Note that here linear time means time $O(m)$ where $m$ is the number of edges. (Throughout the paper, that deals with higher complexities, most complexities will be expressed as a function of the number of vertices $n$.)

\paragraph{Focus on the patterns on three vertices.}
Since one can expect that the larger the pattern, the more complicated the detection problem, we start with patterns on three vertices. 
Paper~\cite{FeuilloleyH21} made a systematic study of the classes of graphs that can be characterized by the existence of a vertex ordering such that a given pattern (or a given family of patterns) on three vertices does not appear. 
It happens that there are around 20 such classes (and their complements), and that they are almost all well-known classes (paths, trees, interval, chordal, split, permutation, comparability, triangle-free, bipartite graphs etc.). In addition, all of them can be recognized in linear time \cite{FeuilloleyH21}, except for the classes defined by the individual patterns \ptriangle{}, \pcomparability{} and their complements. (For patterns on three vertices, we will use the notations from~\cite{FeuilloleyH21}. In particular, the patterns have names in small capital letters, that are related to the graph classes in which they are forbidden. The details are given in Appendix~\ref{sec:three-vertices}. The pattern we have just mentioned are illustrated in Figure~\ref{fig:comp-triangle-chordal}.)

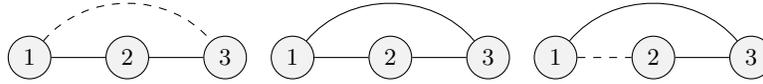
\begin{figure}[!h]
    \centering
    \begin{tabular}{ccccc}
    \begin{tikzpicture}
			[scale=1.3, auto=left,every node/.style=		
			{circle,draw,fill=black!5}]
			\node (a) at (0,0) {1};
			\node (b) at (1,0) {2};
			\node (c) at (2,0) {3};
			\draw[dashed] (a) to[bend left=50] (c);
			\draw (a) to (b);
			\draw (b) to (c);
		\end{tikzpicture} \hspace{1.5cm}&
    \begin{tikzpicture}
			[scale=1.3,auto=left,every node/.style=		
			{circle,draw,fill=black!5}]
			\node (a) at (0,0) {1};
			\node (b) at (1,0) {2};
			\node (c) at (2,0) {3};
			\draw (a) to[bend left=50] (c);
			\draw (a) to (b);
			\draw (b) to (c);
		\end{tikzpicture}
    \hspace{1.5cm}&
    \begin{tikzpicture}
	[scale=1.3,auto=left,every node/.style=		
	{circle,draw,fill=black!5}]
	\node (a) at (0,0) {1};
	\node (b) at (1,0) {2};
	\node (c) at (2,0) {3};
	\draw (a) to[bend left=50] (c);
	\draw[dashed] (a) to (b);
	\draw (b) to (c);
\end{tikzpicture}
    \end{tabular}
    \caption{Illustration of the patterns \pcomparability{}, \ptriangle{} and \pchordal{}}
    \label{fig:comp-triangle-chordal}
\end{figure}
\vspace{-.5cm}

\paragraph{Recognition versus detection.}
It is important to make a clear distinction between two algorithmic problems. On the one hand, for a fixed pattern $P$ (resp. a family of patterns), the \emph{recognition problem} takes as input a graph and asks whether there exists an ordering of the graph that avoids the pattern (resp. the patterns).
On the other hand, for a fixed pattern $P$ (resp. a family of patterns), the \emph{detection problem} takes as input an ordered graph, and asks whether the pattern is present or not in the ordered graph. 

These two problems are incomparable in general. 
The recognition problem might seem harder, in the sense that in this problem we need to decide whether there exists a vertex ordering with some given property, whereas in the detection problem we are given the ordering. 
But in many cases, there are very efficient ways to generate an ordering with very strong properties, such that either this ordering avoids the pattern, or no ordering does. 
In these cases, the detection problem can be harder, because one has to design an algorithm that detects the pattern in \emph{any} input ordering, and not only in one that has a very specific shape. 

\paragraph{Detection of one pattern on three vertices.}
We prove the following theorem.

\begin{restatable}{theorem}{thmOnePatternThreeVertices}
\label{thm:OnePatternThreeVertices}
Let $P$ be a pattern on three vertices.
Detecting $P$ can be done in linear time, if $\cal P$ is not \pcomparability{}, nor \ptriangle{}, nor one of their complements.
\end{restatable}

Note that the hard cases are the same as for the recognition problem. 
Actually, for these patterns, generating a good candidate ordering is somehow easier than checking the ordering. Indeed,  the complexity is the same for all patterns except for the \pcomparability{} pattern, for which generating a good candidate can be done in linear time~\cite{McConnellS99} while it is believed that the detection cannot be done in subquadratic time.

\paragraph{Easy cases using neighborhood manipulations}

All the cases of the proof of Theorem~\ref{thm:OnePatternThreeVertices} can be handled with similarly, except for \pchordal{} (see Figure~\ref{fig:comp-triangle-chordal}), that we will consider later. 
First, for every vertex $i$, we build $N^-(i)$ and $N^+(i)$, that are the ordered adjacency lists of the predecessors and successors of $i$, respectively. 
These (ordered) lists can be computed in time $O(m)$, from the unsorted adjacency list, since the numbers are between 1 and $n$.
We illustrate the technique for \pinterval{}. Here, it is straightforward to see from the pattern, that it is \emph{not} present, if and only if, for all $i$, either $N^+(i)$ is empty or it forms an interval of vertices: $[i+1,j]$ for some~$j$. Since for checking this property, the neighborhood of each vertex is scanned only $O(1)$ times, the complexity is $O(m)$. 

\paragraph{The case of \pchordal{}.}

A general rule of thumb for pattern detection is that the more decided edges, the more difficult the detection. 
On three vertices, \ptriangle{} and \pcomparability{} (and their complement) are a priori the most difficult to recognize, 
but the last fully-specified pattern, \pchordal{}, is different: we can detect it in linear time.

It has been proved in~\cite{GalinierHP95,TarjanY84} that we can detect \pchordal{} in linear time when the ordering has been produced by a graph traversal called LexBFS. 
But here we want an algorithm that does not exploit the nice structure given by LexBFS. 
Note that the natural algorithm that consists in checking that the predecessors of a vertex form a clique is quadratic.
Instead, in our algorithm we only check for each vertex $i$ whether the pattern appears on some positions $(k,j,i)$, where $j$ is the right-most predecessor of $i$. 
We prove that, thanks to the structure of the pattern, such a partial verification is enough to ensure correctness, and it can be done in linear time. 

\paragraph{Extension to patterns characterizing other classes.}
We prove the following extension, that considers all classes characterizable by a set of patterns on three vertices.

\begin{restatable}{theorem}{thmThreeVerticesClass}
\label{thm:three-vertices-class}
Consider a graph class $\mathcal{C}$ characterized by a set of forbidden patterns on three vertices, that is neither the class of (co-)comparability nor \mbox{(co-)} triangle-free graphs. 
There exists a set of forbidden patterns characterizing $\mathcal{C}$ that can be detected in linear time.
\end{restatable}

As a consequence, since for (co-)comparability we can build a candidate ordering in linear time~\cite{McConnellS99}, we get that for every graph class characterized by a set of patterns on 3 vertices, \emph{it easier to find the ordering than to verify it}. This is not true for larger patterns: the path pattern on $k$ vertices corresponds to $k$-colorable graphs and therefore finding a good order is NP-hard.

\paragraph{Other linear-time cases}

We also prove that other types of patterns can be recognized in linear time. The formal theorems and proofs are deferred to the appendix, but roughly our results are as follows. 
First, all the positive patterns whose mandatory edges form a forest with no crossings, can be recognized in linear time (Appendix~\ref{sec:forests}).
 Then, all the positive patterns on four vertices whose mandatory edges form a path can be recognized in linear time (Appendix~\ref{sec:positive-P4}). 
Finally, among the few patterns on more than three vertices studied in the literature, two studied in intersection graph theory can be recognized in linear time (Appendix~\ref{sec:geometry}).
The techniques used are generalizations and variations of the techniques used for the patterns on three vertices and of the techniques we will describe in Section~\ref{sec:o-parametrized}. 


\section{Reduction and lower bounds (overview)}
\label{sec:o-lower-bounds}

Recall that almost every pattern on three nodes can be verified in linear time (Theorem~\ref{thm:OnePatternThreeVertices}). The next natural step is to study what happens for patterns on four nodes. 
On the positive side, we prove the existence of faster-than-$n^k$ time algorithms in order to verify any pattern on $k \geq 4$ nodes.
Unfortunately, we also give evidence in this section that many such patterns can only be detected in (at least) quadratic time, even for $k=4$.

\subsection{General reduction from clique detection}
\label{subsec:reduction-clique}

Recall that any pattern on $k$ nodes can be verified in $n^k$ time. 

However, there exists a faster algorithm: in roughly $n^{\omega k/3}$ time where $\omega$ denotes the exponent for square matrix multiplication, thanks to the following theorem.

\begin{theorem}
\label{thm:reduction-k-clique}
There exists a reduction from the $P$-{\sc Detection} problem on an $n$-vertex ordered graph $G=(V,E,\tau)$, for any $k$-node pattern $P$, to the detection of a $k$-clique in a graph $G'$ with $kn$ vertices.
\end{theorem}

Note that the pattern $P$ does not need to be fully specified. 
We now describe the reduction, that is illustrated in~Figure~\ref{fig:reduction} for the case $k=3$.
The vertex set of $G'$ equals $V \times \{1,2,\ldots,k\}$. For every $u,v \in V$ and for every $i,j$ such that $1 \leq i < j \leq k$, we add an edge in $G'$ between $(u,i)$ and $(v,j)$ if and only if one of the following sets of conditions hold:
\begin{itemize}
    \item There is a mandatory edge between the $i^{th}$ and $j^{th}$ vertices of $P$, $u$ and $v$ are adjacent in $G$, $\tau(u) < \tau(v)$;
    \item There is a forbidden edge between the $i^{th}$ and $j^{th}$ vertices of $P$, $u$ and $v$ are nonadjacent in $G$, $\tau(u) < \tau(v)$;
    \item The $i^{th}$ and $j^{th}$ vertices of $P$ form an undecided pair, $\tau(u) < \tau(v)$.
\end{itemize}

\vspace{-0.8cm}
\begin{figure}[!h]
    \centering
    \scalebox{0.9}{
    \input{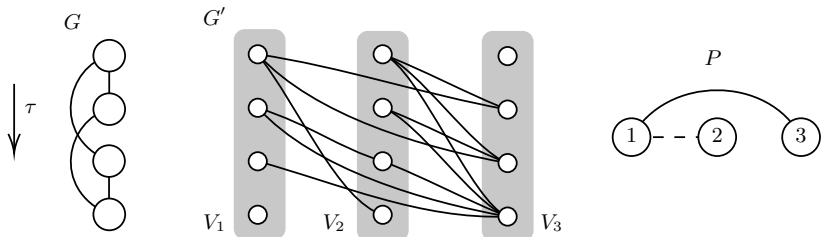}}
    \caption{Illustration of the reduction from $P$-{\sc Detection} to clique detection. On the left the ordered graph $(G,\tau)$ ordered from top to bottom, and on the right the pattern we want to detect. In the middle, the graph $G'$ built by the reduction. Here the pattern can be found in the last three vertices of $(G,\tau)$ and this implies that we can find a $3$-clique in~$G'$: consider the second vertex of $V_1$, the third of $V_2$ and the fourth of $V_3$.}
    \label{fig:reduction}
\end{figure}


This reduction can be done in ${\cal O}((kn)^2)$ time.
One can  check that there is an ordered $k$-subgraph of $G$ that realizes $P$ if and only if there is a $k$-clique in $G'$.

\subsection{Lower bounds}

We prove conditional lower bounds for the pattern detection problem. We actually prove various results, using several techniques, and assuming different complexity hypotheses. 
For the full list of results, we refer to Appendix~\ref{sec:lowerbounds}. 

In this overview, we will focus on the case of patterns of size 4, and give the flavor of our results. We now list the complexity hypotheses we use, which are classic hypotheses.
We refer to~\cite{DVW19,LWWW18,Williams18,WilliamsW18} for a thorough discussion about their plausibility, and their implications in the field of fine-grained complexity. See also Subsection~\ref{sec:lb:hyp} in the appendix.

\begin{description}
    \item[Hypothesis 1] 
    For every $k \geq 3$, deciding whether a graph contains a clique (an independent set, resp.) on $k$ vertices requires $n^{\omega\left(\lfloor k/3 \rfloor, \lceil k/3 \rceil, \lceil (k-1)/3 \rceil\right) - o(1)}$ time, with $\omega(p,q,r)$ the exponent for multiplying two matrices of respective dimensions $n^p \times n^q$ and $n^q \times n^r$.
    \item[Hypothesis 2]
    Every {\em combinatorial} algorithm (not using fast matrix multiplication or other algebraic techniques) for detecting a triangle in a graph requires $n^{3-o(1)}$ time.
    \item[Hypothesis 3]
    Deciding whether a $3$-uniform hypergraph contains a hyperclique with four nodes requires $n^{4-o(1)}$ time.    
\end{description} 

The following theorem follows from gathering most of our results, (and weakening some to make them fit in the same statement). 

\begin{theorem}
\label{thm:lowerbound-four}
Consider a fully-specified pattern $P$ on four vertices. 
Then, assuming Hypotheses 1, 2 and 3, the detection of $P$ cannot be done combinatorially in subquadratic time. 
For {\em most} such patterns $P$, the lower bound also holds for any algorithm (not necessarily combinatorial).
\end{theorem}

The rest of this section is a sketch of the proof of Theorem~\ref{thm:lowerbound-four}, augmented with some discussions. We highlight that this is an informal sketch, and that the proofs are all written in a detailed and formal way in Section~\ref{sec:lowerbounds}.


\paragraph{Probabilistic reduction to subgraph detection.}
Since all three hypotheses listed above are believed to be true even for randomized algorithms, we can use randomized reductions that are correct with small but constant probability.

Roughly, we relate the complexity to verify a fully specified pattern to the structural and algorithmic properties of a related (undirected unordered) graph $H$.
Specifically, a fully-specified pattern is \emph{$H$-based} if when we remove the ordering and the forbidden edges we get $H$. 
Suppose that we can detect efficiently an $H$-based pattern $P_H$. 
Then we prove that we can decide whether a graph $G$ contains $H$ as an induced subgraph with constant probability, in the same complexity (plus an additive linear term). 
For that, given the graph $G$, we generate a random permutation of its vertices (using Fisher-Yates shuffle algorithm\cite{Dur64,FiY63}), run the $P_H$-detection algorithm on this ordered graph and output the same decision. 
If the induced subgraph $H$ is present, then there is a probability at least $1/(k!)$ that the answer is yes, and if it is absent this algorithm always answers no. 
(See Proposition~\ref{prop:red-isi} in the appendix.)

Now, almost all the proof is about how to adapt all the results known for unordered subgraph detection in order to make them go through the reduction above and obtain the result we want. This is much more challenging than it might seem at first because we do not have a full understanding of the complexity of unordered subgraph detection. Hence, sometimes we can reuse results in a black-box manner, but often we adapt or complete known proofs, and in some cases build ad hoc reductions from scratch.

\paragraph{Glossary for $H$ graphs.}

In the proof of Theorem~\ref{thm:lowerbound-four}, we will consider all the patterns of size four, but most of the time we can process all the ones that are based on the same graph $H$ at once.
The exhaustive list of such graphs $H$ on four vertices is given in Figure~\ref{fig:all-four-vertices}, with the names we will use.

\vspace{-0.5cm}
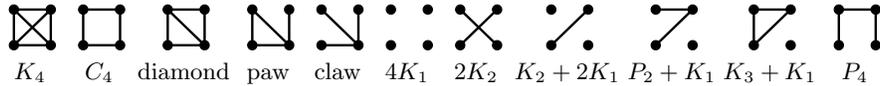
\begin{figure}[!h]
    \centering
    \begin{tabular}{ccccccccccc}
    \scalebox{0.6}{
    \tikzset{every picture/.style={line width=0.75pt}} 

\begin{tikzpicture}[x=0.75pt,y=0.75pt,yscale=-1,xscale=1]

\draw  [fill={rgb, 255:red, 0; green, 0; blue, 0 }  ,fill opacity=1 ] (96,130.85) .. controls (96,128.72) and (97.72,127) .. (99.85,127) .. controls (101.98,127) and (103.7,128.72) .. (103.7,130.85) .. controls (103.7,132.98) and (101.98,134.7) .. (99.85,134.7) .. controls (97.72,134.7) and (96,132.98) .. (96,130.85) -- cycle ;
\draw [line width=1.5]    (99.85,130.85) -- (99.85,160.85) ;
\draw  [fill={rgb, 255:red, 0; green, 0; blue, 0 }  ,fill opacity=1 ] (127,130.85) .. controls (127,128.72) and (128.72,127) .. (130.85,127) .. controls (132.98,127) and (134.7,128.72) .. (134.7,130.85) .. controls (134.7,132.98) and (132.98,134.7) .. (130.85,134.7) .. controls (128.72,134.7) and (127,132.98) .. (127,130.85) -- cycle ;
\draw  [fill={rgb, 255:red, 0; green, 0; blue, 0 }  ,fill opacity=1 ] (96,160.85) .. controls (96,158.72) and (97.72,157) .. (99.85,157) .. controls (101.98,157) and (103.7,158.72) .. (103.7,160.85) .. controls (103.7,162.98) and (101.98,164.7) .. (99.85,164.7) .. controls (97.72,164.7) and (96,162.98) .. (96,160.85) -- cycle ;
\draw  [fill={rgb, 255:red, 0; green, 0; blue, 0 }  ,fill opacity=1 ] (127,160.85) .. controls (127,158.72) and (128.72,157) .. (130.85,157) .. controls (132.98,157) and (134.7,158.72) .. (134.7,160.85) .. controls (134.7,162.98) and (132.98,164.7) .. (130.85,164.7) .. controls (128.72,164.7) and (127,162.98) .. (127,160.85) -- cycle ;
\draw [line width=1.5]    (130.85,130.85) -- (130.85,160.85) ;
\draw [line width=1.5]    (130.85,130.85) -- (99.85,130.85) ;
\draw [line width=1.5]    (130.85,160.85) -- (99.85,160.85) ;
\draw [line width=1.5]    (130.85,160.85) -- (99.85,130.85) ;
\draw [line width=1.5]    (130.85,130.85) -- (99.85,160.85) ;

\end{tikzpicture}}&
    \scalebox{0.6}{
    \tikzset{every picture/.style={line width=0.75pt}} 

\begin{tikzpicture}[x=0.75pt,y=0.75pt,yscale=-1,xscale=1]

\draw  [fill={rgb, 255:red, 0; green, 0; blue, 0 }  ,fill opacity=1 ] (96,130.85) .. controls (96,128.72) and (97.72,127) .. (99.85,127) .. controls (101.98,127) and (103.7,128.72) .. (103.7,130.85) .. controls (103.7,132.98) and (101.98,134.7) .. (99.85,134.7) .. controls (97.72,134.7) and (96,132.98) .. (96,130.85) -- cycle ;

\draw  [fill={rgb, 255:red, 0; green, 0; blue, 0 }  ,fill opacity=1 ] (127,130.85) .. controls (127,128.72) and (128.72,127) .. (130.85,127) .. controls (132.98,127) and (134.7,128.72) .. (134.7,130.85) .. controls (134.7,132.98) and (132.98,134.7) .. (130.85,134.7) .. controls (128.72,134.7) and (127,132.98) .. (127,130.85) -- cycle ;
\draw  [fill={rgb, 255:red, 0; green, 0; blue, 0 }  ,fill opacity=1 ] (96,160.85) .. controls (96,158.72) and (97.72,157) .. (99.85,157) .. controls (101.98,157) and (103.7,158.72) .. (103.7,160.85) .. controls (103.7,162.98) and (101.98,164.7) .. (99.85,164.7) .. controls (97.72,164.7) and (96,162.98) .. (96,160.85) -- cycle ;
\draw  [fill={rgb, 255:red, 0; green, 0; blue, 0 }  ,fill opacity=1 ] (127,160.85) .. controls (127,158.72) and (128.72,157) .. (130.85,157) .. controls (132.98,157) and (134.7,158.72) .. (134.7,160.85) .. controls (134.7,162.98) and (132.98,164.7) .. (130.85,164.7) .. controls (128.72,164.7) and (127,162.98) .. (127,160.85) -- cycle ;

\draw [line width=1.5]    (99.85,130.85) -- (99.85,160.85) ;
\draw [line width=1.5]    (130.85,130.85) -- (130.85,160.85) ;
\draw [line width=1.5]    (130.85,130.85) -- (99.85,130.85) ;
\draw [line width=1.5]    (130.85,160.85) -- (99.85,160.85) ;

\end{tikzpicture}}&
    \scalebox{0.6}{
    \tikzset{every picture/.style={line width=0.75pt}} 

\begin{tikzpicture}[x=0.75pt,y=0.75pt,yscale=-1,xscale=1]

\draw  [fill={rgb, 255:red, 0; green, 0; blue, 0 }  ,fill opacity=1 ] (96,130.85) .. controls (96,128.72) and (97.72,127) .. (99.85,127) .. controls (101.98,127) and (103.7,128.72) .. (103.7,130.85) .. controls (103.7,132.98) and (101.98,134.7) .. (99.85,134.7) .. controls (97.72,134.7) and (96,132.98) .. (96,130.85) -- cycle ;

\draw  [fill={rgb, 255:red, 0; green, 0; blue, 0 }  ,fill opacity=1 ] (127,130.85) .. controls (127,128.72) and (128.72,127) .. (130.85,127) .. controls (132.98,127) and (134.7,128.72) .. (134.7,130.85) .. controls (134.7,132.98) and (132.98,134.7) .. (130.85,134.7) .. controls (128.72,134.7) and (127,132.98) .. (127,130.85) -- cycle ;
\draw  [fill={rgb, 255:red, 0; green, 0; blue, 0 }  ,fill opacity=1 ] (96,160.85) .. controls (96,158.72) and (97.72,157) .. (99.85,157) .. controls (101.98,157) and (103.7,158.72) .. (103.7,160.85) .. controls (103.7,162.98) and (101.98,164.7) .. (99.85,164.7) .. controls (97.72,164.7) and (96,162.98) .. (96,160.85) -- cycle ;
\draw  [fill={rgb, 255:red, 0; green, 0; blue, 0 }  ,fill opacity=1 ] (127,160.85) .. controls (127,158.72) and (128.72,157) .. (130.85,157) .. controls (132.98,157) and (134.7,158.72) .. (134.7,160.85) .. controls (134.7,162.98) and (132.98,164.7) .. (130.85,164.7) .. controls (128.72,164.7) and (127,162.98) .. (127,160.85) -- cycle ;

\draw [line width=1.5]    (99.85,130.85) -- (99.85,160.85) ;
\draw [line width=1.5]    (130.85,130.85) -- (130.85,160.85) ;
\draw [line width=1.5]    (130.85,130.85) -- (99.85,130.85) ;
\draw [line width=1.5]    (130.85,160.85) -- (99.85,160.85) ;
\draw [line width=1.5]    (130.85,160.85) -- (99.85,130.85) ;

\end{tikzpicture}}&
    \scalebox{0.6}{
    \tikzset{every picture/.style={line width=0.75pt}} 

\begin{tikzpicture}[x=0.75pt,y=0.75pt,yscale=-1,xscale=1]

\draw  [fill={rgb, 255:red, 0; green, 0; blue, 0 }  ,fill opacity=1 ] (96,130.85) .. controls (96,128.72) and (97.72,127) .. (99.85,127) .. controls (101.98,127) and (103.7,128.72) .. (103.7,130.85) .. controls (103.7,132.98) and (101.98,134.7) .. (99.85,134.7) .. controls (97.72,134.7) and (96,132.98) .. (96,130.85) -- cycle ;

\draw  [fill={rgb, 255:red, 0; green, 0; blue, 0 }  ,fill opacity=1 ] (127,130.85) .. controls (127,128.72) and (128.72,127) .. (130.85,127) .. controls (132.98,127) and (134.7,128.72) .. (134.7,130.85) .. controls (134.7,132.98) and (132.98,134.7) .. (130.85,134.7) .. controls (128.72,134.7) and (127,132.98) .. (127,130.85) -- cycle ;
\draw  [fill={rgb, 255:red, 0; green, 0; blue, 0 }  ,fill opacity=1 ] (96,160.85) .. controls (96,158.72) and (97.72,157) .. (99.85,157) .. controls (101.98,157) and (103.7,158.72) .. (103.7,160.85) .. controls (103.7,162.98) and (101.98,164.7) .. (99.85,164.7) .. controls (97.72,164.7) and (96,162.98) .. (96,160.85) -- cycle ;
\draw  [fill={rgb, 255:red, 0; green, 0; blue, 0 }  ,fill opacity=1 ] (127,160.85) .. controls (127,158.72) and (128.72,157) .. (130.85,157) .. controls (132.98,157) and (134.7,158.72) .. (134.7,160.85) .. controls (134.7,162.98) and (132.98,164.7) .. (130.85,164.7) .. controls (128.72,164.7) and (127,162.98) .. (127,160.85) -- cycle ;

\draw [line width=1.5]    (99.85,130.85) -- (99.85,160.85) ;
\draw [line width=1.5]    (130.85,130.85) -- (130.85,160.85) ;
\draw [line width=1.5]    (130.85,160.85) -- (99.85,160.85) ;
\draw [line width=1.5]    (130.85,160.85) -- (99.85,130.85) ;

\end{tikzpicture}}&
    \scalebox{0.6}{
    \tikzset{every picture/.style={line width=0.75pt}} 

\begin{tikzpicture}[x=0.75pt,y=0.75pt,yscale=-1,xscale=1]

\draw  [fill={rgb, 255:red, 0; green, 0; blue, 0 }  ,fill opacity=1 ] (96,130.85) .. controls (96,128.72) and (97.72,127) .. (99.85,127) .. controls (101.98,127) and (103.7,128.72) .. (103.7,130.85) .. controls (103.7,132.98) and (101.98,134.7) .. (99.85,134.7) .. controls (97.72,134.7) and (96,132.98) .. (96,130.85) -- cycle ;

\draw  [fill={rgb, 255:red, 0; green, 0; blue, 0 }  ,fill opacity=1 ] (127,130.85) .. controls (127,128.72) and (128.72,127) .. (130.85,127) .. controls (132.98,127) and (134.7,128.72) .. (134.7,130.85) .. controls (134.7,132.98) and (132.98,134.7) .. (130.85,134.7) .. controls (128.72,134.7) and (127,132.98) .. (127,130.85) -- cycle ;
\draw  [fill={rgb, 255:red, 0; green, 0; blue, 0 }  ,fill opacity=1 ] (96,160.85) .. controls (96,158.72) and (97.72,157) .. (99.85,157) .. controls (101.98,157) and (103.7,158.72) .. (103.7,160.85) .. controls (103.7,162.98) and (101.98,164.7) .. (99.85,164.7) .. controls (97.72,164.7) and (96,162.98) .. (96,160.85) -- cycle ;
\draw  [fill={rgb, 255:red, 0; green, 0; blue, 0 }  ,fill opacity=1 ] (127,160.85) .. controls (127,158.72) and (128.72,157) .. (130.85,157) .. controls (132.98,157) and (134.7,158.72) .. (134.7,160.85) .. controls (134.7,162.98) and (132.98,164.7) .. (130.85,164.7) .. controls (128.72,164.7) and (127,162.98) .. (127,160.85) -- cycle ;

\draw [line width=1.5]    (130.85,130.85) -- (130.85,160.85) ;
\draw [line width=1.5]    (130.85,160.85) -- (99.85,160.85) ;
\draw [line width=1.5]    (130.85,160.85) -- (99.85,130.85) ;

\end{tikzpicture}}&
    \scalebox{0.6}{
    \tikzset{every picture/.style={line width=0.75pt}} 

\begin{tikzpicture}[x=0.75pt,y=0.75pt,yscale=-1,xscale=1]

\draw  [fill={rgb, 255:red, 0; green, 0; blue, 0 }  ,fill opacity=1 ] (96,130.85) .. controls (96,128.72) and (97.72,127) .. (99.85,127) .. controls (101.98,127) and (103.7,128.72) .. (103.7,130.85) .. controls (103.7,132.98) and (101.98,134.7) .. (99.85,134.7) .. controls (97.72,134.7) and (96,132.98) .. (96,130.85) -- cycle ;

\draw  [fill={rgb, 255:red, 0; green, 0; blue, 0 }  ,fill opacity=1 ] (127,130.85) .. controls (127,128.72) and (128.72,127) .. (130.85,127) .. controls (132.98,127) and (134.7,128.72) .. (134.7,130.85) .. controls (134.7,132.98) and (132.98,134.7) .. (130.85,134.7) .. controls (128.72,134.7) and (127,132.98) .. (127,130.85) -- cycle ;
\draw  [fill={rgb, 255:red, 0; green, 0; blue, 0 }  ,fill opacity=1 ] (96,160.85) .. controls (96,158.72) and (97.72,157) .. (99.85,157) .. controls (101.98,157) and (103.7,158.72) .. (103.7,160.85) .. controls (103.7,162.98) and (101.98,164.7) .. (99.85,164.7) .. controls (97.72,164.7) and (96,162.98) .. (96,160.85) -- cycle ;
\draw  [fill={rgb, 255:red, 0; green, 0; blue, 0 }  ,fill opacity=1 ] (127,160.85) .. controls (127,158.72) and (128.72,157) .. (130.85,157) .. controls (132.98,157) and (134.7,158.72) .. (134.7,160.85) .. controls (134.7,162.98) and (132.98,164.7) .. (130.85,164.7) .. controls (128.72,164.7) and (127,162.98) .. (127,160.85) -- cycle ;


\end{tikzpicture}}&
    \scalebox{0.6}{
    \tikzset{every picture/.style={line width=0.75pt}} 

\begin{tikzpicture}[x=0.75pt,y=0.75pt,yscale=-1,xscale=1]

\draw  [fill={rgb, 255:red, 0; green, 0; blue, 0 }  ,fill opacity=1 ] (96,130.85) .. controls (96,128.72) and (97.72,127) .. (99.85,127) .. controls (101.98,127) and (103.7,128.72) .. (103.7,130.85) .. controls (103.7,132.98) and (101.98,134.7) .. (99.85,134.7) .. controls (97.72,134.7) and (96,132.98) .. (96,130.85) -- cycle ;

\draw  [fill={rgb, 255:red, 0; green, 0; blue, 0 }  ,fill opacity=1 ] (127,130.85) .. controls (127,128.72) and (128.72,127) .. (130.85,127) .. controls (132.98,127) and (134.7,128.72) .. (134.7,130.85) .. controls (134.7,132.98) and (132.98,134.7) .. (130.85,134.7) .. controls (128.72,134.7) and (127,132.98) .. (127,130.85) -- cycle ;
\draw  [fill={rgb, 255:red, 0; green, 0; blue, 0 }  ,fill opacity=1 ] (96,160.85) .. controls (96,158.72) and (97.72,157) .. (99.85,157) .. controls (101.98,157) and (103.7,158.72) .. (103.7,160.85) .. controls (103.7,162.98) and (101.98,164.7) .. (99.85,164.7) .. controls (97.72,164.7) and (96,162.98) .. (96,160.85) -- cycle ;
\draw  [fill={rgb, 255:red, 0; green, 0; blue, 0 }  ,fill opacity=1 ] (127,160.85) .. controls (127,158.72) and (128.72,157) .. (130.85,157) .. controls (132.98,157) and (134.7,158.72) .. (134.7,160.85) .. controls (134.7,162.98) and (132.98,164.7) .. (130.85,164.7) .. controls (128.72,164.7) and (127,162.98) .. (127,160.85) -- cycle ;

\draw [line width=1.5]    (130.85,160.85) -- (99.85,130.85) ;
\draw [line width=1.5]    (130.85,130.85) -- (99.85,160.85) ;

\end{tikzpicture}}&
    \scalebox{0.6}{
    \tikzset{every picture/.style={line width=0.75pt}} 

\begin{tikzpicture}[x=0.75pt,y=0.75pt,yscale=-1,xscale=1]

\draw  [fill={rgb, 255:red, 0; green, 0; blue, 0 }  ,fill opacity=1 ] (96,130.85) .. controls (96,128.72) and (97.72,127) .. (99.85,127) .. controls (101.98,127) and (103.7,128.72) .. (103.7,130.85) .. controls (103.7,132.98) and (101.98,134.7) .. (99.85,134.7) .. controls (97.72,134.7) and (96,132.98) .. (96,130.85) -- cycle ;

\draw  [fill={rgb, 255:red, 0; green, 0; blue, 0 }  ,fill opacity=1 ] (127,130.85) .. controls (127,128.72) and (128.72,127) .. (130.85,127) .. controls (132.98,127) and (134.7,128.72) .. (134.7,130.85) .. controls (134.7,132.98) and (132.98,134.7) .. (130.85,134.7) .. controls (128.72,134.7) and (127,132.98) .. (127,130.85) -- cycle ;
\draw  [fill={rgb, 255:red, 0; green, 0; blue, 0 }  ,fill opacity=1 ] (96,160.85) .. controls (96,158.72) and (97.72,157) .. (99.85,157) .. controls (101.98,157) and (103.7,158.72) .. (103.7,160.85) .. controls (103.7,162.98) and (101.98,164.7) .. (99.85,164.7) .. controls (97.72,164.7) and (96,162.98) .. (96,160.85) -- cycle ;
\draw  [fill={rgb, 255:red, 0; green, 0; blue, 0 }  ,fill opacity=1 ] (127,160.85) .. controls (127,158.72) and (128.72,157) .. (130.85,157) .. controls (132.98,157) and (134.7,158.72) .. (134.7,160.85) .. controls (134.7,162.98) and (132.98,164.7) .. (130.85,164.7) .. controls (128.72,164.7) and (127,162.98) .. (127,160.85) -- cycle ;

\draw [line width=1.5]    (130.85,130.85) -- (99.85,160.85) ;

\end{tikzpicture}}&
    \scalebox{0.6}{
    \tikzset{every picture/.style={line width=0.75pt}} 

\begin{tikzpicture}[x=0.75pt,y=0.75pt,yscale=-1,xscale=1]

\draw  [fill={rgb, 255:red, 0; green, 0; blue, 0 }  ,fill opacity=1 ] (96,130.85) .. controls (96,128.72) and (97.72,127) .. (99.85,127) .. controls (101.98,127) and (103.7,128.72) .. (103.7,130.85) .. controls (103.7,132.98) and (101.98,134.7) .. (99.85,134.7) .. controls (97.72,134.7) and (96,132.98) .. (96,130.85) -- cycle ;

\draw  [fill={rgb, 255:red, 0; green, 0; blue, 0 }  ,fill opacity=1 ] (127,130.85) .. controls (127,128.72) and (128.72,127) .. (130.85,127) .. controls (132.98,127) and (134.7,128.72) .. (134.7,130.85) .. controls (134.7,132.98) and (132.98,134.7) .. (130.85,134.7) .. controls (128.72,134.7) and (127,132.98) .. (127,130.85) -- cycle ;
\draw  [fill={rgb, 255:red, 0; green, 0; blue, 0 }  ,fill opacity=1 ] (96,160.85) .. controls (96,158.72) and (97.72,157) .. (99.85,157) .. controls (101.98,157) and (103.7,158.72) .. (103.7,160.85) .. controls (103.7,162.98) and (101.98,164.7) .. (99.85,164.7) .. controls (97.72,164.7) and (96,162.98) .. (96,160.85) -- cycle ;
\draw  [fill={rgb, 255:red, 0; green, 0; blue, 0 }  ,fill opacity=1 ] (127,160.85) .. controls (127,158.72) and (128.72,157) .. (130.85,157) .. controls (132.98,157) and (134.7,158.72) .. (134.7,160.85) .. controls (134.7,162.98) and (132.98,164.7) .. (130.85,164.7) .. controls (128.72,164.7) and (127,162.98) .. (127,160.85) -- cycle ;

\draw [line width=1.5]    (130.85,130.85) -- (99.85,130.85) ;
\draw [line width=1.5]    (130.85,130.85) -- (99.85,160.85) ;

\end{tikzpicture}}&
    \scalebox{0.6}{
    \tikzset{every picture/.style={line width=0.75pt}} 

\begin{tikzpicture}[x=0.75pt,y=0.75pt,yscale=-1,xscale=1]

\draw  [fill={rgb, 255:red, 0; green, 0; blue, 0 }  ,fill opacity=1 ] (96,130.85) .. controls (96,128.72) and (97.72,127) .. (99.85,127) .. controls (101.98,127) and (103.7,128.72) .. (103.7,130.85) .. controls (103.7,132.98) and (101.98,134.7) .. (99.85,134.7) .. controls (97.72,134.7) and (96,132.98) .. (96,130.85) -- cycle ;

\draw  [fill={rgb, 255:red, 0; green, 0; blue, 0 }  ,fill opacity=1 ] (127,130.85) .. controls (127,128.72) and (128.72,127) .. (130.85,127) .. controls (132.98,127) and (134.7,128.72) .. (134.7,130.85) .. controls (134.7,132.98) and (132.98,134.7) .. (130.85,134.7) .. controls (128.72,134.7) and (127,132.98) .. (127,130.85) -- cycle ;
\draw  [fill={rgb, 255:red, 0; green, 0; blue, 0 }  ,fill opacity=1 ] (96,160.85) .. controls (96,158.72) and (97.72,157) .. (99.85,157) .. controls (101.98,157) and (103.7,158.72) .. (103.7,160.85) .. controls (103.7,162.98) and (101.98,164.7) .. (99.85,164.7) .. controls (97.72,164.7) and (96,162.98) .. (96,160.85) -- cycle ;
\draw  [fill={rgb, 255:red, 0; green, 0; blue, 0 }  ,fill opacity=1 ] (127,160.85) .. controls (127,158.72) and (128.72,157) .. (130.85,157) .. controls (132.98,157) and (134.7,158.72) .. (134.7,160.85) .. controls (134.7,162.98) and (132.98,164.7) .. (130.85,164.7) .. controls (128.72,164.7) and (127,162.98) .. (127,160.85) -- cycle ;

\draw [line width=1.5]    (99.85,130.85) -- (99.85,160.85) ;
\draw [line width=1.5]    (130.85,130.85) -- (99.85,130.85) ;
\draw [line width=1.5]    (130.85,130.85) -- (99.85,160.85) ;

\end{tikzpicture}}&
    \scalebox{0.6}{
    \tikzset{every picture/.style={line width=0.75pt}} 

\begin{tikzpicture}[x=0.75pt,y=0.75pt,yscale=-1,xscale=1]

\draw  [fill={rgb, 255:red, 0; green, 0; blue, 0 }  ,fill opacity=1 ] (96,130.85) .. controls (96,128.72) and (97.72,127) .. (99.85,127) .. controls (101.98,127) and (103.7,128.72) .. (103.7,130.85) .. controls (103.7,132.98) and (101.98,134.7) .. (99.85,134.7) .. controls (97.72,134.7) and (96,132.98) .. (96,130.85) -- cycle ;

\draw  [fill={rgb, 255:red, 0; green, 0; blue, 0 }  ,fill opacity=1 ] (127,130.85) .. controls (127,128.72) and (128.72,127) .. (130.85,127) .. controls (132.98,127) and (134.7,128.72) .. (134.7,130.85) .. controls (134.7,132.98) and (132.98,134.7) .. (130.85,134.7) .. controls (128.72,134.7) and (127,132.98) .. (127,130.85) -- cycle ;
\draw  [fill={rgb, 255:red, 0; green, 0; blue, 0 }  ,fill opacity=1 ] (96,160.85) .. controls (96,158.72) and (97.72,157) .. (99.85,157) .. controls (101.98,157) and (103.7,158.72) .. (103.7,160.85) .. controls (103.7,162.98) and (101.98,164.7) .. (99.85,164.7) .. controls (97.72,164.7) and (96,162.98) .. (96,160.85) -- cycle ;
\draw  [fill={rgb, 255:red, 0; green, 0; blue, 0 }  ,fill opacity=1 ] (127,160.85) .. controls (127,158.72) and (128.72,157) .. (130.85,157) .. controls (132.98,157) and (134.7,158.72) .. (134.7,160.85) .. controls (134.7,162.98) and (132.98,164.7) .. (130.85,164.7) .. controls (128.72,164.7) and (127,162.98) .. (127,160.85) -- cycle ;

\draw [line width=1.5]    (99.85,130.85) -- (99.85,160.85) ;
\draw [line width=1.5]    (130.85,130.85) -- (130.85,160.85) ;
\draw [line width=1.5]    (130.85,130.85) -- (99.85,130.85) ;

\end{tikzpicture}}\\
    $K_4$ &
    $C_4$ &
    diamond &
    paw & 
    claw &
    $4K_1$ &
    $2K_2$ &
    $K_2+2K_1$ &
    $P_2+K_1$ & 
    $K_3+K_1$ &
    $P_4$
\end{tabular}
    \caption{The 11 (unordered) graphs on four vertices: $K_4$, $C_4$, diamond, paw, claw, their respective complements, and $P_4$ which is its own complement.}
    \label{fig:all-four-vertices}
\end{figure}

\vspace{-1cm}

\paragraph{Cliques, independent sets, and $C_4$.}

The subgraph detection problem is best understood for cliques. But actually the case of the complete pattern (all mandatory edges) directly follows from Hypothesis~1, without needing the probabilistic reduction, since the ordering is irrelevant. 
Now, it might seem natural that if a subgraph $H$ is hard to detect, then the detection of any larger subgraph $H'$ containing it must also be hard. But this is not known to be true in general. 
It was proved for the case of cliques only, in~\cite{DVW19}. 
Thus, for any pattern on four vertices containing a triangle, we get from Hypothesis~1 that it requires time at least $n^{\omega-\varepsilon}$ (where $\omega=\omega(1,1,1) \in [2, 2.3727]$ \cite{Williams12}). 

At that point, we have proved Theorem~\ref{thm:lowerbound-four} for all the patterns that are $H$-based with $H$ being $K_4$, diamond, paw, or $K_3+K_1$. 
Now, thanks to complementation, we also get lower bound for $4K_1$. But for independent sets of size three, we are not aware of any result that we can use directly (that is, of analogues of~\cite{DVW19}). 
Hence, for the case of claw, $K_2+2K_1$ and $P_2+K_1$, we design direct reductions to their detection from detecting $3K_1$. See Lemma~\ref{lem:3-indep-set} in the Appendix. 

The case of $C_4$-based patterns is special, since it is the only case where we need to use Hypothesis~3, and get only a $n^{2-\varepsilon}$ lower bound. This is achieved by using directly the very recent result of~\cite{DaW22}.  

\paragraph{The need for new tools, directed paths and finishing the proof.}

At that point, we are left with $H=P_4$ or $2K_2$. These are the ones for which we need to work the most to prove our lower bounds. 
For these, we cannot use the probabilistic reduction we have used so far, because no nontrivial lower bound is known in the subgraph setting for these. 
Even worse, detecting $P_4$ in a graph is actually not hard, because the $P_4$-free graphs are the cographs, that are known to be recognizable in linear time~\cite{CorneilPS85}. 

We design a new general result based on the notion of directed path. 
Let a \emph{directed path} in a pattern be a path in the pattern, where the vertices follow the order of the pattern. 
We prove that if a pattern $P$ contains a unique directed path on $t$ vertices, then there is a $O(n^2)$ time reduction from $t$-clique detection to $P$-detection. 
See Proposition~\ref{prop:k-path} in the appendix.
In addition to improving on some patterns based on diamond, paw, or $C_4$, this allows us to get $n^{3-o(1)}$ lower bounds for combinatorial algorithms, assuming Hypothesis~2 for all $P_4$-based patterns where the path contains three (not necessarily consecutive) vertices $u,v,w$ that appear in this order in the pattern. 

Thanks to the directed path technique, we solve various $P_4$-based patterns, and also some $2K_2$-based patterns, via complementation. 
We finish the proof with several ad hoc reductions for the remaining cases, in particular from triangle detection. 
See Appendix~\ref{sec:lb-k=4}.

\paragraph{Beyond four vertices.}

Our primary goal was to understand completely the case of fully-specified patterns of four vertices. 
Along the way, several results obtained are actually more general (the probabilistic reduction from induced subgraph detection, the cases of patterns containing cliques, the directed path technique). We prove one more such result: if $P$ is the prefix of a larger pattern $Q$, then detecting $Q$ is at least as hard as detecting $P$ (under randomized reductions). 
See Appendix~\ref{sec:lb>4}.
\section{Parameterized algorithms (overview)}
\label{sec:o-parametrized}

In this section, we will introduce a parameter for patterns, the \emph{merge-width}, and design a polynomial-in-$n$ algorithm whose complexity exponent will depend only on this parameter.
Let us start by giving some intuition about our approach with the example of the pattern below, that we call a \emph{flat cycle} of length~$k$.

\vspace{-0.2cm}
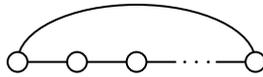
    \begin{center}
    \tikzset{every picture/.style={line width=0.75pt}} 

\begin{tikzpicture}[x=0.75pt,y=0.75pt,yscale=-1,xscale=1]

\draw    (135,85) -- (215,85) ;
\draw  [fill={rgb, 255:red, 255; green, 255; blue, 255 }  ,fill opacity=1 ] (190,85) .. controls (190,82.24) and (192.24,80) .. (195,80) .. controls (197.76,80) and (200,82.24) .. (200,85) .. controls (200,87.76) and (197.76,90) .. (195,90) .. controls (192.24,90) and (190,87.76) .. (190,85) -- cycle ;
\draw  [draw opacity=0] (135,85) .. controls (135,85) and (135,85) .. (135,85) .. controls (135,69.05) and (161.86,56.12) .. (195,56.12) .. controls (228.14,56.12) and (255,69.05) .. (255,85) -- (195,85) -- cycle ; \draw   (135,85) .. controls (135,85) and (135,85) .. (135,85) .. controls (135,69.05) and (161.86,56.12) .. (195,56.12) .. controls (228.14,56.12) and (255,69.05) .. (255,85) ;  
\draw  [fill={rgb, 255:red, 255; green, 255; blue, 255 }  ,fill opacity=1 ] (160,85) .. controls (160,82.24) and (162.24,80) .. (165,80) .. controls (167.76,80) and (170,82.24) .. (170,85) .. controls (170,87.76) and (167.76,90) .. (165,90) .. controls (162.24,90) and (160,87.76) .. (160,85) -- cycle ;
\draw  [fill={rgb, 255:red, 255; green, 255; blue, 255 }  ,fill opacity=1 ] (130,85) .. controls (130,82.24) and (132.24,80) .. (135,80) .. controls (137.76,80) and (140,82.24) .. (140,85) .. controls (140,87.76) and (137.76,90) .. (135,90) .. controls (132.24,90) and (130,87.76) .. (130,85) -- cycle ;
\draw    (235,85) -- (255,85) ;
\draw  [fill={rgb, 255:red, 255; green, 255; blue, 255 }  ,fill opacity=1 ] (250,85) .. controls (250,82.24) and (252.24,80) .. (255,80) .. controls (257.76,80) and (260,82.24) .. (260,85) .. controls (260,87.76) and (257.76,90) .. (255,90) .. controls (252.24,90) and (250,87.76) .. (250,85) -- cycle ;

\draw (216,82.8) node [anchor=north west][inner sep=0.75pt]    {$\dotsc $};

\end{tikzpicture}

    \end{center}
\vspace{-0.2cm}

To detect this pattern in an ordered graph $G$, we use a dynamic programming approach.  
First, for every $x\in[2,k]$, we compute all the pairs $(i,j)$, with $1\leq i<j\leq n$, such that there exists a path of length $x$ in $G$, using increasing indices, with left-most vertex $i$ and right-most vertex $j$. 
For $x=2$, this is straightforward from the adjacency matrix of $G$.
For any~$x\geq 3$, this can be done easily by considering the intervals $[i,j']$ computed for $x-1$, and checking whether there is an edge $(j',j)$, $j'<j$ in $E(G)$. 
Now, for every pair $(i,j)$ corresponding to a path on $k$ vertices, we check whether $(i,j)$ is an edge of $G$ or not. 
The complexity of this algorithm is at most $O(n^3)$. 

Note that the complexity does not depend on $k$. 
Instead, the complexity comes from the fact that we have manipulated pairs of positions, corresponding to the extremities of ``subpatterns''.
If the pattern had been only the path, without the feedback edge, then we could have stored only the right end of the intervals: it would have been enough to remember that there exists a subpath of length $x$ that ends at some given vertex. 
But since we needed to add the feedback edge, it was necessary to store both endpoints.
We can rephrase this as: we needed the two endpoints of the path as \emph{anchors} to \emph{merge} with the feedback edge. 


\paragraph{Intuition of the merge-width and how to use it.}
We define a parameter on patterns, that we call \emph{merge-width}. 
We consider a few operations that allow to build a pattern: creations of vertices and edges, and merges of patterns. 
In the merge of two patterns, we will identify some $r$ vertices of the first pattern with $r$ vertices of the second (in a consistent order).  
At any step of the construction, any vertex can have a special role, that we call being an \emph{anchor}. The vertices that are identified in a merge must be anchors, for example. 
The \emph{merge-width} is the minimum over all possible constructions of a pattern, of the maximum, over all the intermediate patterns, of the number of anchors.

Our algorithm will proceed in a similar way as in the example:
we will follow the construction of the pattern bottom-up, storing at each point all the possible sets of positions for the anchors. 
Then, for a merge, we will compare these lists, detect the ones that match on all the anchors used in the merge, and keep the merge of these anchors list.
The complexity of the algorithm mainly depends on the size of the sets of the positions we manipulate, and this is measured by the merge-width. Hence, our algorithm has complexity $n^{ct}$, where $c$ is a constant and $t$ is the merge-width.

\paragraph{Difficulty: definition of the anchors}

Deciding which vertices should be anchors or not is actually tricky. 
It is not enough to simply keep the vertices that have to be merged at some point, because in the target pattern some vertices should appear in a specific order, and by keeping only the ``to-be-merged vertices" we cannot control this. 
See Figure~\ref{fig:wrong-merge}.

\vspace{-0.5cm}
\begin{figure}[!h]
    \centering
    \scalebox{0.60}{
    \input{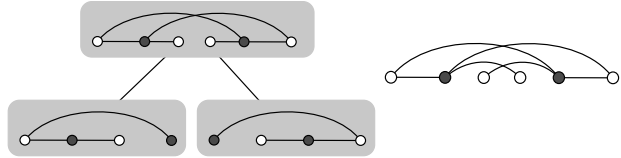}
    }
    \caption{Consider the merge on the left, where the anchors are the black vertices. Without more constraints, one can also obtain the pattern on the right. In this case, the algorithm could output Yes, while the correct answer is No.}
    \label{fig:wrong-merge}
\end{figure}
\vspace{-0.5cm}

To avoid this, we define a set of constraints that the construction of the pattern should satisfy in terms of anchor placement, and prove that this is enough to ensure the correctness of the algorithm.
See Subsection~\ref{subsec:anchored-operations} in the appendix.

\paragraph{Speed-up by matrix multiplication}
We speed up the computation of the anchors positions in case of a merge, thanks to matrix multiplication. 
For, example in the case of the motivating example, matrix multiplication allows getting to $O(n^\omega)$ instead of $O(n^3)$, and this is generalized for merges with more than two anchors.

\paragraph{Applications}
We prove that the patterns that can be drawn without crossing of edges (outerplanar patterns) have merge-width 2, and that if the number of crossings is bounded by $r$ then the merge width is at most $2r+2$.  
We believe that the merge-width should also be bounded for patterns of bounded book thickness, and for some notion of ``pattern tree-width", but we leave this for further work.

\section{Open problems}

We have introduced the pattern detection problem for ordered graphs. Many interesting directions are still to be explored on this new topic. 
Below is a list of some open problems. 

\begin{itemize}
    \item Can we go further on the lower bound side? 
    Using results from subgraph detection might not help much more, since we are already reaching the limit of the current knowledge (\emph{e.g.} the result about $C_4$ that we use is very recent). But maybe one can build more ad hoc proofs? 
    \item Is merge-width the right parameter to capture the difficulty of the problem? Can we show that patterns that can be detected fast have small merge-width?
    \item Is merge-width comparable to a known parameter? The book-thickness could be a candidate here.
    \item Are positive planar forests the largest natural class of patterns that can be recognized in linear time? 
    \item Would the picture be much different if instead of ordered graphs, we looked at partially ordered graphs? 
    \item Regarding ``geometric patterns", can the pattern $P_{ac}$ be tested in linear time, as it is a kind of combination of two linear-time detectable patterns $P_a$ and~$P_c$?
\end{itemize}

\newpage
\bibliography{biblio.bib}{}

\begin{thebibliography}{10}

\bibitem{Chu08}
Frank Pok~Man Chu.
\newblock A simple linear time certifying {LBFS}-based algorithm for
  recognizing trivially perfect graphs and their complements.
\newblock {\em Inf. Process. Lett.}, 107(1):7--12, 2008.

\bibitem{CorneilDHK16}
Derek~G. Corneil, J{\'{e}}r{\'{e}}mie Dusart, Michel Habib, and Ekkehard
  K{\"{o}}hler.
\newblock On the power of graph searching for cocomparability graphs.
\newblock {\em {SIAM} J. Discrete Math.}, 30(1):569--591, 2016.

\bibitem{CorneilK08}
Derek~G. Corneil and Richard Krueger.
\newblock A unified view of graph searching.
\newblock {\em {SIAM} J. Discrete Math.}, 22(4):1259--1276, 2008.

\bibitem{CorneilPS85}
Derek~G. Corneil, Yehoshua Perl, and Lorna~K. Stewart.
\newblock A linear recognition algorithm for cographs.
\newblock {\em {SIAM} J. Comput.}, 14(4):926--934, 1985.

\bibitem{DVW19}
Mina Dalirrooyfard, Thuy~Duong Vuong, and Virginia~Vassilevska Williams.
\newblock {Graph pattern detection: Hardness for all induced patterns and
  faster non-induced cycles}.
\newblock In {\em Proceedings of the 51st Annual ACM SIGACT Symposium on Theory
  of Computing}, pages 1167--1178, 2019.

\bibitem{DaW22}
Mina Dalirrooyfard and Virginia~Vassilevska Williams.
\newblock Induced cycles and paths are harder than you think.
\newblock In {\em 2022 IEEE 63rd Annual Symposium on Foundations of Computer
  Science (FOCS)}, pages 531--542. IEEE, 2022.

\bibitem{Dur64}
Richard Durstenfeld.
\newblock Algorithm 235: random permutation.
\newblock {\em Communications of the ACM}, 7(7):420, 1964.

\bibitem{FH22}
Laurent Feuilloley and Michel Habib.
\newblock Classifying grounded intersection graphs via ordered forbidden
  patterns, December, http://arxiv.org/abs/2112.00629ø, 2021.

\bibitem{FeuilloleyH21}
Laurent Feuilloley and Michel Habib.
\newblock Graph classes and forbidden patterns on three vertices.
\newblock {\em SIAM Journal on Discrete Mathematics (SIDMA)}, 35(1):55--90,
  2021.

\bibitem{FiY63}
Ronald~Aylmer Fisher and Frank Yates.
\newblock {\em Statistical tables for biological, agricultural and medical
  research}.
\newblock Edinburgh: Oliver and Boyd, 1963.

\bibitem{GalinierHP95}
Philippe Galinier, Michel Habib, and Christophe Paul.
\newblock Chordal graphs and their clique graphs.
\newblock In {\em Graph-Theoretic Concepts in Computer Science, 21st
  International Workshop, {WG} '95, Aachen, Germany, June 20-22, 1995,
  Proceedings}, pages 358--371, 1995.

\bibitem{Gallai67}
Tibor Gallai.
\newblock Transitiv orientierbare graphen.
\newblock {\em Acta Math. Hungarica}, 18(1):25--66, 1967.

\bibitem{HabibMPV00}
Michel Habib, Ross~M. McConnell, Christophe Paul, and Laurent Viennot.
\newblock Lex-bfs and partition refinement, with applications to transitive
  orientation, interval graph recognition and consecutive ones testing.
\newblock {\em Theor. Comput. Sci.}, 234(1-2):59--84, 2000.

\bibitem{huang1998fast}
Xiaohan Huang and Victor~Y Pan.
\newblock Fast rectangular matrix multiplication and applications.
\newblock {\em Journal of complexity}, 14(2):257--299, 1998.

\bibitem{LekkeikerkerB62}
C~Lekkeikerker and J~Boland.
\newblock Representation of a finite graph by a set of intervals on the real
  line.
\newblock {\em Fundamenta Mathematicae}, 51(1):45--64, 1962.

\bibitem{LWWW18}
Andrea Lincoln, Virginia~Vassilevska Williams, and R.~Ryan Williams.
\newblock Tight hardness for shortest cycles and paths in sparse graphs.
\newblock In {\em Proceedings of the Twenty-Ninth Annual ACM-SIAM Symposium on
  Discrete Algorithms}, pages 1236--1252. SIAM, 2018.

\bibitem{McConnellMNS11}
Ross~M. McConnell, Kurt Mehlhorn, Stefan N{\"{a}}her, and Pascal Schweitzer.
\newblock Certifying algorithms.
\newblock {\em Comput. Sci. Rev.}, 5(2):119--161, 2011.

\bibitem{McConnellS99}
Ross~M. McConnell and Jeremy~P. Spinrad.
\newblock Modular decomposition and transitive orientation.
\newblock {\em Discrete Mathematics}, 201(1-3):189--241, 1999.

\bibitem{Olariu88}
Stephan Olariu.
\newblock Paw-free graphs.
\newblock {\em Information Processing Letters}, 28(1):53--54, 1988.

\bibitem{RTEL76}
Donald~J. Rose, R.~Endre Tarjan, and George~S. Lueker.
\newblock Algorithmic aspects of vertex elimination on graphs.
\newblock {\em SIAM Journal on Computing}, 5(2):266--283, 1976.

\bibitem{Schaefer12}
Marcus Schaefer.
\newblock The graph crossing number and its variants: A survey.
\newblock {\em The electronic journal of combinatorics}, pages DS21--Apr, 2012.

\bibitem{TY85b}
Robert~E. Tarjan and Mihalis Yannakakis.
\newblock Addendum: Simple linear-time algorithms to test chordality of graphs,
  test acyclicity of hypergraphs, and selectively reduce acyclic hypergraphs.
\newblock {\em SIAM Journal on Computing}, 14(1):254--255, 1985.

\bibitem{TarjanY84}
Robert~Endre Tarjan and Mihalis Yannakakis.
\newblock Simple linear-time algorithms to test chordality of graphs, test
  acyclicity of hypergraphs, and selectively reduce acyclic hypergraphs.
\newblock {\em {SIAM} J. Comput.}, 13(3):566--579, 1984.

\bibitem{TCHP08}
Marc Tedder, Derek Corneil, Michel Habib, and Cristophe Paul.
\newblock Simpler linear-time modular decomposition via recursive factorizing
  permutations.
\newblock In {\em ICALP}, pages 634--645. Springer, 2008.

\bibitem{Williams12}
Virginia~Vassilevska Williams.
\newblock Multiplying matrices faster than coppersmith-winograd.
\newblock In Howard~J. Karloff and Toniann Pitassi, editors, {\em Proceedings
  of the 44th Symposium on Theory of Computing Conference, {STOC} 2012, New
  York, NY, USA, May 19 - 22, 2012}, pages 887--898. {ACM}, 2012.

\bibitem{Williams18}
Virginia~Vassilevska Williams.
\newblock On some fine-grained questions in algorithms and complexity.
\newblock In {\em Proceedings of the International Congress of Mathematicians:
  Rio de Janeiro 2018}, pages 3447--3487. World Scientific, 2018.

\bibitem{WilliamsW18}
Virginia~Vassilevska Williams and R.~Ryan Williams.
\newblock Subcubic equivalences between path, matrix, and triangle problems.
\newblock {\em J. {ACM}}, 65(5):27:1--27:38, 2018.

\bibitem{YCC96}
J.-H Yan, J.-J Chen, and G.J. Chang.
\newblock Quasi-threshold graphs.
\newblock {\em Discrete Applied Math.}, 69 (3), 1996.

\end{thebibliography}
\bibliographystyle{plain}

\newpage

\appendix

\noindent\textbf{Appendix table of contents}

\begin{itemize}
   \item Appendix~\ref{sec:three-vertices}: Detection of patterns on three vertices.
    \item Appendix~\ref{sec:lowerbounds}: Conditional lower bounds for patterns with $k \geq$ 4.
    \item Appendix~\ref{sec:parametrized}: Parametrized algorithm and merge-width.
    \item Appendix~\ref{sec:forests}: Linear-time detection of positive planar forests.
    \item Appendix~\ref{sec:positive-P4}:
    Linear-time detection of every positive $P_4$.
    \item Appendix~\ref{sec:geometry}: Linear-time detection of patterns arising from geometry.
    
\end{itemize}

\section{Detection of patterns on three vertices}
\label{sec:three-vertices}


Patterns on three vertices are very relevant to the study of classic graph classes. Indeed, many classes can be characterized by the existence of an ordering of the vertices that avoids one or several such patterns. 
In this detailed section, we study the detection problem for pattern on three vertices. 

We prove two theorems. The first one is for detecting one pattern.

\thmOnePatternThreeVertices* 

The natural generalization of this is to consider all sets of patterns on three vertices. 
Since this would go through a lengthy case analysis, for this version of the paper we go for a less systematic, but still meaningful version. The list of graph classes defined by such sets of patterns is established in~\cite{FeuilloleyH21}. 
For every such class, we consider one of the sets of patterns that defines it.

\thmThreeVerticesClass*


The rest of this section is devoted to the proofs of these theorems. Theorem~\ref{thm:OnePatternThreeVertices} is proved in two parts: Subsections~\ref{subsec:one-pattern-easy} and~\ref{subsec:one-pattern-chordal} (in the latter, we study the case where the pattern is \pchordal{} or co-\pchordal{}). Theorem~\ref{thm:three-vertices-class} is proved in Subsection~\ref{subsec:more-than-one-pattern}.

\subsection{Name conventions and observations}

In this section, for readability, we will refer to the pattern on three vertices by names that have been assigned in~\cite{FeuilloleyH21}. 
These are listed in Figure~\ref{fig:27-patterns}. 
They are related to the classes in which they are forbidden, and they are written in small capital letters.

\begin{figure}[!h]
\hspace{-0.6cm}
\scalebox{0.7}{
\begin{tabular}{ccc}
\begin{tabular}{cc}
0: \ptriangle{}
&
\begin{tikzpicture}
	[scale=1,auto=left,every node/.style=		
	{circle,draw,fill=black!5}]
	\node (a) at (0,0) {};
	\node (b) at (1,0) {};
	\node (c) at (2,0) {};
	\draw (a) to (b);
	\draw (a) to[bend left=50] (c);
	\draw (b) to (c);
\end{tikzpicture}\\
1: mirror-\pchordal{}
&
\begin{tikzpicture}
	[scale=1,auto=left,every node/.style=		
	{circle,draw,fill=black!5}]
	\node (a) at (0,0) {};
	\node (b) at (1,0) {};
	\node (c) at (2,0) {};
	\draw (a) to (b);
	\draw (a) to[bend left=50] (c);
	\draw[dashed] (b) to (c);
\end{tikzpicture}\\
2: \pcomparability
&
\begin{tikzpicture}
	[scale=1,auto=left,every node/.style=		
	{circle,draw,fill=black!5}]
	\node (a) at (0,0) {};
	\node (b) at (1,0) {};
	\node (c) at (2,0) {};
	\draw (a) to (b);
	\draw[dashed] (a) to[bend left=50] (c);
	\draw (b) to (c);
\end{tikzpicture}\\
3 co-\pchordal{}
&
\begin{tikzpicture}
	[scale=1,auto=left,every node/.style=		
	{circle,draw,fill=black!5}]
	\node (a) at (0,0) {};
	\node (b) at (1,0) {};
	\node (c) at (2,0) {};
	\draw (a) to (b);
	\draw[dashed] (a) to[bend left=50] (c);
	\draw[dashed] (b) to (c);
\end{tikzpicture}\\
4: \pchordal{}
&
\begin{tikzpicture}
	[scale=1,auto=left,every node/.style=		
	{circle,draw,fill=black!5}]
	\node (a) at (0,0) {};
	\node (b) at (1,0) {};
	\node (c) at (2,0) {};
	\draw[dashed] (a) to (b);
	\draw (a) to[bend left=50] (c);
	\draw (b) to (c);
\end{tikzpicture}\\
5: co-\pcomparability
&
\begin{tikzpicture}
	[scale=1,auto=left,every node/.style=		
	{circle,draw,fill=black!5}]
	\node (a) at (0,0) {};
	\node (b) at (1,0) {};
	\node (c) at (2,0) {};
	\draw[dashed] (a) to (b);
	\draw (a) to[bend left=50] (c);
	\draw[dashed] (b) to (c);
\end{tikzpicture}\\
6: mirror-co-\pchordal{}
&
\begin{tikzpicture}
	[scale=1,auto=left,every node/.style=		
	{circle,draw,fill=black!5}]
	\node (a) at (0,0) {};
	\node (b) at (1,0) {};
	\node (c) at (2,0) {};
	\draw[dashed] (a) to (b);
	\draw[dashed] (a) to[bend left=50] (c);
	\draw (b) to (c);
\end{tikzpicture}\\
7: co-\ptriangle{}
&
\begin{tikzpicture}
	[scale=1,auto=left,every node/.style=		
	{circle,draw,fill=black!5}]
	\node (a) at (0,0) {};
	\node (b) at (1,0) {};
	\node (c) at (2,0) {};
	\draw[dashed] (a) to (b);
	\draw[dashed] (a) to[bend left=50] (c);
	\draw[dashed] (b) to (c);
\end{tikzpicture}\\
\end{tabular}

&

\begin{tabular}{cc}
8: \pforest{}
&
\begin{tikzpicture}
	[scale=1,auto=left,every node/.style=		
	{circle,draw,fill=black!5}]
	\node (a) at (0,0) {};
	\node (b) at (1,0) {};
	\node (c) at (2,0) {};
	\draw (a) to[bend left=50] (c);
	\draw (b) to (c);
\end{tikzpicture}\\
9: mirror-\pinterval{}
&
\begin{tikzpicture}
	[scale=1,auto=left,every node/.style=		
	{circle,draw,fill=black!5}]
	\node (a) at (0,0) {};
	\node (b) at (1,0) {};
	\node (c) at (2,0) {};
	\draw (a) to[bend left=50] (c);
	\draw[dashed] (b) to (c);
\end{tikzpicture}\\
10: mirror-co-\pinterval{}
&
\begin{tikzpicture}
	[scale=1,auto=left,every node/.style=		
	{circle,draw,fill=black!5}]
	\node (a) at (0,0) {};
	\node (b) at (1,0) {};
	\node (c) at (2,0) {};
	\draw[dashed] (a) to[bend left=50] (c);
	\draw (b) to (c);
\end{tikzpicture}\\
11: co-\pforest{}
&
\begin{tikzpicture}
	[scale=1,auto=left,every node/.style=		
	{circle,draw,fill=black!5}]
	\node (a) at (0,0) {};
	\node (b) at (1,0) {};
	\node (c) at (2,0) {};
	\draw[dashed] (a) to[bend left=50] (c);
	\draw[dashed] (b) to (c);
\end{tikzpicture}\\
\\
12: \pbipartite{}
&
\begin{tikzpicture}
	[scale=1,auto=left,every node/.style=		
	{circle,draw,fill=black!5}]
	\node (a) at (0,0) {};
	\node (b) at (1,0) {};
	\node (c) at (2,0) {};
	\draw (a) to (b);
	\draw (b) to (c);
\end{tikzpicture}\\
\\
13: \psplit{}
&
\begin{tikzpicture}
	[scale=1,auto=left,every node/.style=		
	{circle,draw,fill=black!5}]
	\node (a) at (0,0) {};
	\node (b) at (1,0) {};
	\node (c) at (2,0) {};
	\draw (a) to (b);
	\draw[dashed] (b) to (c);
\end{tikzpicture}\\
\\
14: mirror-\psplit{}=co-\psplit{}
&
\begin{tikzpicture}
	[scale=1,auto=left,every node/.style=		
	{circle,draw,fill=black!5}]
	\node (a) at (0,0) {};
	\node (b) at (1,0) {};
	\node (c) at (2,0) {};
	\draw[dashed] (a) to (b);
	\draw (b) to (c);
\end{tikzpicture}\\
\\
15: co-\pbipartite{}
&
\begin{tikzpicture}
	[scale=1,auto=left,every node/.style=		
	{circle,draw,fill=black!5}]
	\node (a) at (0,0) {};
	\node (b) at (1,0) {};
	\node (c) at (2,0) {};
	\draw[dashed] (a) to (b);
	\draw[dashed] (b) to (c);
\end{tikzpicture}\\
16: mirror-\pforest{}
&
\begin{tikzpicture}
	[scale=1,auto=left,every node/.style=		
	{circle,draw,fill=black!5}]
	\node (a) at (0,0) {};
	\node (b) at (1,0) {};
	\node (c) at (2,0) {};
	\draw (a) to (b);
	\draw (a) to[bend left=50] (c);
\end{tikzpicture}\\
17: co-\pinterval{}
&
\begin{tikzpicture}
	[scale=1,auto=left,every node/.style=		
	{circle,draw,fill=black!5}]
	\node (a) at (0,0) {};
	\node (b) at (1,0) {};
	\node (c) at (2,0) {};
	\draw (a) to (b);
	\draw[dashed] (a) to[bend left=50] (c);
\end{tikzpicture}\\
18: \pinterval{}
&
\begin{tikzpicture}
	[scale=1,auto=left,every node/.style=		
	{circle,draw,fill=black!5}]
	\node (a) at (0,0) {};
	\node (b) at (1,0) {};
	\node (c) at (2,0) {};
	\draw[dashed] (a) to (b);
	\draw (a) to[bend left=50] (c);
\end{tikzpicture}\\
19: mirror-co-\pforest{}
&
\begin{tikzpicture}
	[scale=1,auto=left,every node/.style=		
	{circle,draw,fill=black!5}]
	\node (a) at (0,0) {};
	\node (b) at (1,0) {};
	\node (c) at (2,0) {};
	\draw[dashed] (a) to (b);
	\draw[dashed] (a) to[bend left=50] (c);
\end{tikzpicture}\\
\end{tabular}

&

\begin{tabular}{cc}
20: mirror-\pstar{}
&
\begin{tikzpicture}
	[scale=1,auto=left,every node/.style=		
	{circle,draw,fill=black!5}]
	\node (a) at (0,0) {};
	\node (b) at (1,0) {};
	\node (c) at (2,0) {};
	\draw (b) to (c);
\end{tikzpicture}\\
\\
21: mirror-co-\pstar{}
&
\begin{tikzpicture}
	[scale=1,auto=left,every node/.style=		
	{circle,draw,fill=black!5}]
	\node (a) at (0,0) {};
	\node (b) at (1,0) {};
	\node (c) at (2,0) {};
	\draw[dashed] (b) to (c);
\end{tikzpicture}\\
22: \ppath{}
&
\begin{tikzpicture}
	[scale=1,auto=left,every node/.style=		
	{circle,draw,fill=black!5}]
	\node (a) at (0,0) {};
	\node (b) at (1,0) {};
	\node (c) at (2,0) {};
	\draw(a) to[bend left=50] (c);
\end{tikzpicture}\\
23: co-\ppath{}
&
\begin{tikzpicture}
	[scale=1,auto=left,every node/.style=		
	{circle,draw,fill=black!5}]
	\node (a) at (0,0) {};
	\node (b) at (1,0) {};
	\node (c) at (2,0) {};
	\draw[dashed] (a) to[bend left=50] (c);
\end{tikzpicture}\\
\\
24: \pstar{}
&
\begin{tikzpicture}
	[scale=1,auto=left,every node/.style=		
	{circle,draw,fill=black!5}]
	\node (a) at (0,0) {};
	\node (b) at (1,0) {};
	\node (c) at (2,0) {};
	\draw (a) to (b);
\end{tikzpicture}\\
\\
25: co-\pstar{}
&
\begin{tikzpicture}
	[scale=1,auto=left,every node/.style=		
	{circle,draw,fill=black!5}]
	\node (a) at (0,0) {};
	\node (b) at (1,0) {};
	\node (c) at (2,0) {};
	\draw[dashed] (a) to (b);
\end{tikzpicture}\\
\\
\\
\\
26: \pnograph{}
&
\begin{tikzpicture}
	[scale=1,auto=left,every node/.style=		
	{circle,draw,fill=black!5}]
	\node (a) at (0,0) {};
	\node (b) at (1,0) {};
	\node (c) at (2,0) {};
\end{tikzpicture}\\
\end{tabular}

\end{tabular}
}
\caption{\label{fig:27-patterns} The 27 patterns on three nodes. By convention, since mirror-\psplit{}=co-\psplit{}, we will ignore the pattern mirror-\psplit{}.}
\end{figure}
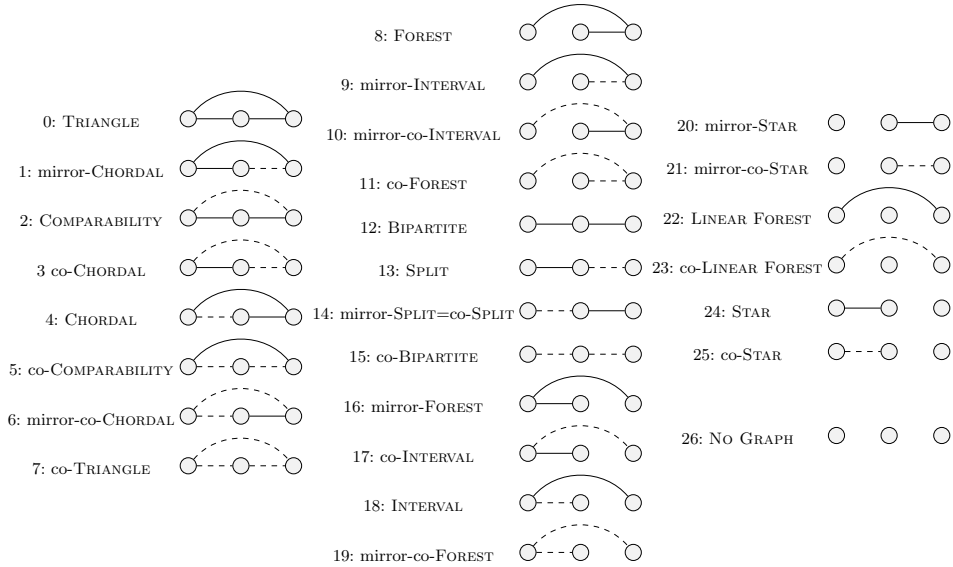

To reduce the number of cases to consider for these 3 nodes patterns we can rely on the following observations.

\begin{remark}\label{mirroring}
Let $f(n, m)$ be the complexity of deciding whether a total  ordering of the vertices of a graph $G$ avoids a pattern $P$, the complexity $mirror$-$f(n, m)$ of deciding that an ordering avoids the mirror pattern $mirror(P)$, satisfies $mirror$-$f(n, m) \leq max\{ f(n,m), g(n, m)\}$, where $g(n, m)$ is the complexity of changing the data in order to transform $\tau$ into $\tau^d$.
\end{remark}

\begin{proof}
If we have an Algorithm $\cal A$ that decides that a given total  ordering $\tau$ of the vertices of a graph $G$ does not have a pattern $P$,
it suffices to apply  $\cal A$ on the graph $G$ and the reverse ordering $\tau^d$ to certify that $\tau^d$ does not have the pattern $P$, and therefore that $\tau$ does not have the pattern $mirror$-$P$. \qed
\end{proof}

In most of the cases, $g(n, m)$ is in $O(n+m)$ since it suffices to sort with respect to $\tau^d$ the adjacency lists of $G$. Therefore if the pattern $P$ can be decided in linear time, then its mirror-pattern $mirror$-$P$ can also be decided in linear time.
 
\begin{remark}\label{unseulealafois}
Let $P_1$ and $P_2$ two patterns and  $f_1(n,m)$ (resp. $f_2(n,m)$) the complexity of deciding that an ordering $\tau$ of the vertices of a graph  $G$ with n vertices and m edges does not yields the pattern  $P_1$ (resp. $P_2$), then the complexity $f_{12}(n,m)$ of deciding that an ordering $\tau$ of the vertices of a graph  $G$ with n vertices and m edges does not yields the patterns   $P_1$  and $P_2$ satisfies 

$f_{12}(n,m) \leq max\{f_{1}(n,m), f_{2}(n,m)\}$.

\end{remark}

Therefore, a set of linear decidable patterns can also be decidable in linear time.

\begin{remark}\label{passageaucomplement}
Let $f(n, m)$ be the complexity of deciding that a total  ordering of the vertices of a graph $G$ avoids a set of patterns $\cal P$ that characterize a graph class $\cal G$.  Then the complement class can be decided in $O(max(\{f(n, m), g(n, m)\})$, where $g(n,m)$ is the complexity to build $\overline{G}$ from $G$.
\end{remark}
\begin{proof}
Starting from a graph $G$ we first construct its complement $\overline{G}$.
If we can characterize $\overline{G}$ by an ordering avoiding the patterns $\cal P$, then it suffices to decide them in $O(f, m)$.\qed
\end{proof}

As a consequence to characterize and decide a complement class we may use the set of patterns made up with the complement patterns (exchanging mandatory edges and forbidden edges). Furthermore in some examples it is not necessary to compute the complement, as it will be used in the following.

\paragraph{Final list.}
If we remove the patterns that are mirror of another pattern, we are left with the following list of 18 patterns (see Figure~\ref{fig:18-patterns}): \ptriangle{}, co-\ptriangle{}, \pcomparability{}, co-\pcomparability{}, \pchordal{}, co-\pchordal{}, \pforest{}, co-\pforest{}, \pinterval{}, co-\pinterval{}, \pbipartite, co-\pbipartite, \psplit{}, \pstar{}, co-\pstar{}, \ppath{}, co-\ppath{}, \pnograph{}.

\begin{figure}[!h]
\scalebox{0.7}{
\begin{tabular}{ccc}
\begin{tabular}{cc}
\ptriangle{}
&
\begin{tikzpicture}
	[scale=1,auto=left,every node/.style=		
	{circle,draw,fill=black!5}]
	\node (a) at (0,0) {};
	\node (b) at (1,0) {};
	\node (c) at (2,0) {};
	\draw (a) to (b);
	\draw (a) to[bend left=50] (c);
	\draw (b) to (c);
\end{tikzpicture}\\
\pcomparability
&
\begin{tikzpicture}
	[scale=1,auto=left,every node/.style=		
	{circle,draw,fill=black!5}]
	\node (a) at (0,0) {};
	\node (b) at (1,0) {};
	\node (c) at (2,0) {};
	\draw (a) to (b);
	\draw[dashed] (a) to[bend left=50] (c);
	\draw (b) to (c);
\end{tikzpicture}\\
co-\pchordal{}
&
\begin{tikzpicture}
	[scale=1,auto=left,every node/.style=		
	{circle,draw,fill=black!5}]
	\node (a) at (0,0) {};
	\node (b) at (1,0) {};
	\node (c) at (2,0) {};
	\draw (a) to (b);
	\draw[dashed] (a) to[bend left=50] (c);
	\draw[dashed] (b) to (c);
\end{tikzpicture}\\
\pchordal{}
&
\begin{tikzpicture}
	[scale=1,auto=left,every node/.style=		
	{circle,draw,fill=black!5}]
	\node (a) at (0,0) {};
	\node (b) at (1,0) {};
	\node (c) at (2,0) {};
	\draw[dashed] (a) to (b);
	\draw (a) to[bend left=50] (c);
	\draw (b) to (c);
\end{tikzpicture}\\
co-\pcomparability
&
\begin{tikzpicture}
	[scale=1,auto=left,every node/.style=		
	{circle,draw,fill=black!5}]
	\node (a) at (0,0) {};
	\node (b) at (1,0) {};
	\node (c) at (2,0) {};
	\draw[dashed] (a) to (b);
	\draw (a) to[bend left=50] (c);
	\draw[dashed] (b) to (c);
\end{tikzpicture}\\
co-\ptriangle{}
&
\begin{tikzpicture}
	[scale=1,auto=left,every node/.style=		
	{circle,draw,fill=black!5}]
	\node (a) at (0,0) {};
	\node (b) at (1,0) {};
	\node (c) at (2,0) {};
	\draw[dashed] (a) to (b);
	\draw[dashed] (a) to[bend left=50] (c);
	\draw[dashed] (b) to (c);
\end{tikzpicture}\\
\end{tabular}

&

\begin{tabular}{cc}
\pforest{}
&
\begin{tikzpicture}
	[scale=1,auto=left,every node/.style=		
	{circle,draw,fill=black!5}]
	\node (a) at (0,0) {};
	\node (b) at (1,0) {};
	\node (c) at (2,0) {};
	\draw (a) to[bend left=50] (c);
	\draw (b) to (c);
\end{tikzpicture}\\
co-\pforest{}
&
\begin{tikzpicture}
	[scale=1,auto=left,every node/.style=		
	{circle,draw,fill=black!5}]
	\node (a) at (0,0) {};
	\node (b) at (1,0) {};
	\node (c) at (2,0) {};
	\draw[dashed] (a) to[bend left=50] (c);
	\draw[dashed] (b) to (c);
\end{tikzpicture}\\
\\
\pbipartite{}
&
\begin{tikzpicture}
	[scale=1,auto=left,every node/.style=		
	{circle,draw,fill=black!5}]
	\node (a) at (0,0) {};
	\node (b) at (1,0) {};
	\node (c) at (2,0) {};
	\draw (a) to (b);
	\draw (b) to (c);
\end{tikzpicture}\\
\\
13: \psplit{}
&
\begin{tikzpicture}
	[scale=1,auto=left,every node/.style=		
	{circle,draw,fill=black!5}]
	\node (a) at (0,0) {};
	\node (b) at (1,0) {};
	\node (c) at (2,0) {};
	\draw (a) to (b);
	\draw[dashed] (b) to (c);
\end{tikzpicture}\\
\\
15: co-\pbipartite{}
&
\begin{tikzpicture}
	[scale=1,auto=left,every node/.style=		
	{circle,draw,fill=black!5}]
	\node (a) at (0,0) {};
	\node (b) at (1,0) {};
	\node (c) at (2,0) {};
	\draw[dashed] (a) to (b);
	\draw[dashed] (b) to (c);
\end{tikzpicture}\\
co-\pinterval{}
&
\begin{tikzpicture}
	[scale=1,auto=left,every node/.style=		
	{circle,draw,fill=black!5}]
	\node (a) at (0,0) {};
	\node (b) at (1,0) {};
	\node (c) at (2,0) {};
	\draw (a) to (b);
	\draw[dashed] (a) to[bend left=50] (c);
\end{tikzpicture}\\
\pinterval{}
&
\begin{tikzpicture}
	[scale=1,auto=left,every node/.style=		
	{circle,draw,fill=black!5}]
	\node (a) at (0,0) {};
	\node (b) at (1,0) {};
	\node (c) at (2,0) {};
	\draw[dashed] (a) to (b);
	\draw (a) to[bend left=50] (c);
\end{tikzpicture}\\
\end{tabular}

&

\begin{tabular}{cc}
\ppath{}
&
\begin{tikzpicture}
	[scale=1,auto=left,every node/.style=		
	{circle,draw,fill=black!5}]
	\node (a) at (0,0) {};
	\node (b) at (1,0) {};
	\node (c) at (2,0) {};
	\draw(a) to[bend left=50] (c);
\end{tikzpicture}\\
co-\ppath{}
&
\begin{tikzpicture}
	[scale=1,auto=left,every node/.style=		
	{circle,draw,fill=black!5}]
	\node (a) at (0,0) {};
	\node (b) at (1,0) {};
	\node (c) at (2,0) {};
	\draw[dashed] (a) to[bend left=50] (c);
\end{tikzpicture}\\
\\
\pstar{}
&
\begin{tikzpicture}
	[scale=1,auto=left,every node/.style=		
	{circle,draw,fill=black!5}]
	\node (a) at (0,0) {};
	\node (b) at (1,0) {};
	\node (c) at (2,0) {};
	\draw (a) to (b);
\end{tikzpicture}\\
\\
co-\pstar{}
&
\begin{tikzpicture}
	[scale=1,auto=left,every node/.style=		
	{circle,draw,fill=black!5}]
	\node (a) at (0,0) {};
	\node (b) at (1,0) {};
	\node (c) at (2,0) {};
	\draw[dashed] (a) to (b);
\end{tikzpicture}\\
\\
\pnograph{}
&
\begin{tikzpicture}
	[scale=1,auto=left,every node/.style=		
	{circle,draw,fill=black!5}]
	\node (a) at (0,0) {};
	\node (b) at (1,0) {};
	\node (c) at (2,0) {};
\end{tikzpicture}\\
\end{tabular}

\end{tabular}
}
\caption{\label{fig:18-patterns} The 18 patterns on three nodes, once we have removed mirrors.}
\end{figure}
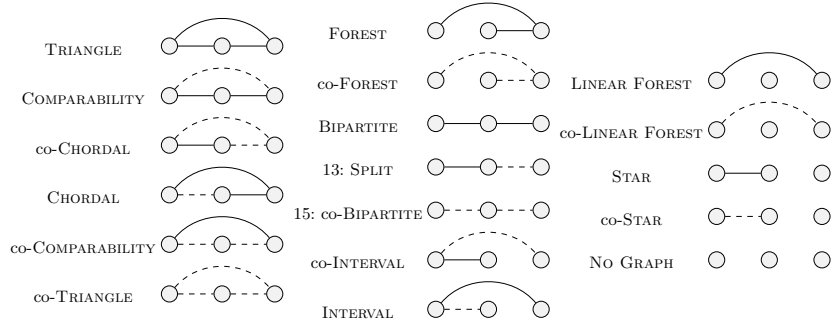

\subsection{One pattern, the easy cases}
\label{subsec:one-pattern-easy}

For most patterns on three vertices, detecting it basically boils down to comparing the neighborhoods of vertices. 
Remember that $N^-(i)$ and $N^+(i)$ denote respectively the predecessors and successors sorted adjacency lists of a vertex.

\begin{proposition}\label{prop:one-pattern-easy}
The detection of each of the following 12 patterns (and of their mirrors) can be done in linear time: \pnograph{}, \psplit{}, \pstar{}, co-\pstar{}, \pbipartite, co-\pbipartite, \pforest, co-\pforest, \ppath{}, co-\ppath{}, \pinterval{}, co-\pinterval{}.
\end{proposition}

\begin{proof}
For each pattern cited, we establish a property such that: \emph{if it holds} at every vertex $i$, then the pattern \emph{does not} appear, other it appears.
(This negation might seem counter-intuitive, but it actually makes the condition easier to state, and it is aligned with the graph class recognition literature.)
These properties deal with the lists $N^-(i)$ and  $N^+(i)$.

\begin{itemize}
\item\pnograph{}:
$|V(G)| \leq 2$.
\item \psplit{}:
If for some $i$, $N^-(i) \neq \emptyset$  then $|N^+(i)|=n-i$.
\item \pstar{}:
If for some $i$, $N^-(i) \neq \emptyset$  then $i=n$.
\item co-\pstar{}:
If for some $i$, $|N^-(i)| < i-1$  then $i=n$.
\item \pbipartite{}:
If for some $i$, $N^-(i) \neq \emptyset$  then $N^+(i)=\emptyset$.
\item co-\pbipartite{}:
If for some $i$, $N^-(i) \neq \emptyset$  then $N^+(i)=\emptyset$.
\item \pforest{}:
$\forall i$, $|N^-(i)| \leq 1$.
\item co-\pforest{}
For all $i$, $|N^-(i)| \geq i-1$.
\item \ppath{}:
$\forall  ij \in E(G)$, with $i < j$ then $i=j-1$.
\item co-\ppath{}:
$\forall i$, either  $|N^+(i)|  = n-i$  or ($|N^+(i)|  =n- i-1$ and $i,i+1 \notin E(G)$ 
\item \pinterval{}:
$\forall i$, $N^+(i)$ if non empty is a unique interval of the type $[i+1, j]$.
\item co-\pinterval{}:
$\forall i$, the non-neighbours of $i$ to its right if exist, yield a unique interval of the type $[i+1, j]$.
\end{itemize}
Since these properties can be checked by scanning the lists $N^-(i)$ and $N^+(i)$ for all $i$ only once, the complexity is $O(m)$.\qed
\end{proof}

\noindent
 By no means the conditions proposed here to detect these patterns are unique, some other conditions could also be correct.

\subsection{One pattern: the special case of \textsc{Chordal} and co-\pchordal{}}
\label{subsec:one-pattern-chordal}

The patterns \pchordal{} and co-\pchordal{} have a special role in this paper. Indeed, they are examples of patterns that are fully specified, thus similar to the patterns studied in the lower bound section and to the patterns \ptriangle{} and \pcomparability{}, but unlike these patterns, these can be detected in linear-time. 

Nevertheless, the proof is less straightforward than in the previous cases. Indeed, the naive algorithm that checks that the left neighborhood of every vertex forms a clique has quadratic complexity.

\paragraph{\pchordal{} detection.}

Some algorithms for detecting the pattern \pchordal{} have been designed before~\cite{GalinierHP95,TY85b}, but they do so only if the graph ordering comes from a graph traversal called LexBFS. These orderings have strong properties that we cannot use, for example the first vertex has some special role in the graph.
We propose the following algorithm (Algorithm~\ref{algo:chordal}).


\vspace{0.5cm}

\begin{algorithm}[H]\label{algo:chordal}
	\DontPrintSemicolon
    \SetAlgoLined
		\SetKwInOut{Input}{Input}
	\SetKwInOut{Output}{Output}
	\SetKw{Continue}{continue}
	\SetKw{Break}{break}

	\Input{$G$ a graph given by its adjacency lists and $\tau$ a total ordering of the vertices}
	\Output{Does $(G,\tau)$ contain \pchordal{}? }
	
	\BlankLine
	Wlog we suppose $\tau=1, 2, \dots,n$\;
    Precompute $N^-(i)$ for every vertex $i$.\;  
	\For {$ i=3$ to $n$}{
	If {$N^-(i) \neq \emptyset$:}
	{Let $j \in N^-(i)$ with largest index \;
	If{ NOT($N^-(i)\setminus j \subseteq  N^-(j)$)} \;
	 {Return ``Yes, \pchordal{} found"} \;
	 }
	 }
	 Return ``No, \pchordal{} not found".
	 
	 \BlankLine
	 \BlankLine	
	 \caption{Detection algorithm for the pattern \pchordal{}.}
\end{algorithm}

\begin{proposition}\label{prop:chordal}
Algorithm~\ref{algo:chordal} detects the pattern \pchordal{} in linear-time.
\end{proposition}

\begin{proof}
We prove first the correctness and then the complexity
\paragraph{Correctness.}
If for some $i$ its largest index neighbor $j$  does not verify $N^-(i) \setminus j \subseteq  N^-(j)$, it yields some vertex $k$ with 
$ki, jk \in E(G)$ and  $kj \notin E(G)$ and therefore \pchordal{} appears.
Conversely, if the algorithm terminates with the answer No, we claim the pattern is not present. 
For the sake of contradiction, suppose the algorithm did not output yes, although the pattern is present. 
Let $i$ be the smallest integer such that the pattern appear with right-most vertex $i$.
The index $i$ must be strictly larger than 3, since otherwise the pattern would be detected.
Now let $z<t$ be the other vertices in the realization of the pattern. 
We have $z, t  \in N^-(i)$ such that $zi, ti \in E(G)$ and  $zt \notin E(G)$. 
The vertex $t$ cannot be the with largest index of $i$, otherwise the algorithm would detect the pattern. 
Thus, there exists $j$, such that $ji \in E$, and $z<t<j<i$. 
Because the algorithm does not output yes, we must have $N^-(i)\setminus j \subseteq  N^-(j)$.  
Then necessarily $zj, tj \in  E(G)$ and it yields a forbidden configuration on $z, t, j$. This contradict the minimality of $i$.

\paragraph{Complexity.}
Each vertex is considered at most once, in the $\textbf{for}$ loop and to evaluate the test $N^-(i)\setminus j \subseteq  N^-(j)$.
Since a vertex j can be several times the closest neighour to the right, we have to be careful to organize this test in order to be linear. One way to obtain this is to associate each time the list $N^-(i)\setminus j$  to $j$. Then in a second scan of the order $\tau$ we compare these associated lists to the list $N^-(j)$. So $N^-(j)$ is only scanned once for all its lists.
Therefore the whole process is linear.\qed
\end{proof}

Note that this detection algorithm can be embedded in known recognition algorithms, using specific searches such as LexBFS, LexDFS or MCS\cite{CorneilK08}.

This algorithm uses a paradigm that we can call \emph{Rely on your rightmost left neighbour}. 
We will generalize this ideas to certify more sophisticated patterns in linear time in the following sections.

\paragraph{co-\pchordal{} detection.}

Note that in linear time, detecting the complement of a pattern $P$ cannot be done by complementing the graph, and then detecting~$P$, since this uses quadratic time. 
In some cases, one can implicitly work on the complement graph, \emph{e.g.} when using partition refinement. Since this is not the case in our algorithm above, we need to define an algorithm for the case of co-\pchordal{}. 
We propose the following algorithm (Algorithm~\ref{algo:co-chordal}).



\begin{algorithm}[ht]
	\DontPrintSemicolon
    \SetAlgoLined
		\SetKwInOut{Input}{Input}
	\SetKwInOut{Output}{Output}
	\SetKw{Continue}{continue}
	\SetKw{Break}{break}

	\Input{$G$ a graph given by its adjacency lists and $\tau$ a total ordering of the vertices}
	\Output{Does $(G,\tau)$ contain co-\pchordal{}? }
	
	\BlankLine
	Wlog we suppose $\tau=1, 2, \dots,n$\;
	Precompute $N^-(i)$ for every vertex $i$.\; 
	\For {$ i=3$ to $n$}
	{
	If {$N^-(i) \neq \emptyset$}
	{Let $j \in \overline{N^-(i)}$ with largest index \;
	If{ NOT($N^-(j) \subseteq  N^-(i)$)} \;
	 {Return "Yes, co-\pchordal{} is found"} \;
	 }
	 }
	 Output "No, co-\pchordal{} not found."
	 
	 \BlankLine
	 \BlankLine	
	 \caption{Certifying a simplicial elemination ordering}\label{algo:co-chordal}
	
\end{algorithm}	

\begin{proposition}\label{prop:co-chordal}
Algorithm \ref{algo:co-chordal} detects co-\pchordal{} in linear time.
\end{proposition}

\begin{proof}
The proof is essentially the same as for the previous algorithm. The only thing to check is that the part "Let $j \in \overline{N^-(i)}$ with largest index"  can be computed in time $O(|N^-(i)|)$, which is the case.\qed
\end{proof}

This concludes the proof of Theorem~\ref{thm:OnePatternThreeVertices}. 
Indeed, there are 18 such patterns (up to mirror), we have given linear detection for 12 such patterns in Proposition~\ref{prop:one-pattern-easy}, and then Propositions~\ref{prop:chordal} and~\ref{prop:co-chordal} gave us two more. The four remaining ones are the ones listed in the theorem as exceptions.




\subsection{Sets corresponding to graph classes}
\label{subsec:more-than-one-pattern}

We now prove Theorem~\ref{thm:three-vertices-class}. For every graph class charaterizable by a set of forbidden patterns on three vertices, we consider such a set of forbidden patterns, and show that it can be detected in linear time, except for the classes comparability, co-comparability, triangle-free and co-triangle-free.

Using the characterization Theorem~2 in~\cite{FeuilloleyH21}, we know that up to complement there exist exactly the 22 non-trivial classes that can be defined using forbidden sets on three nodes. 
(The trivial ones being finite classes, that can be processed in constant time.)

\begin{multicols}{3} 
\begin{enumerate}
\item forests
\item linear forests
\item stars
\item interval 
\item split 
\item bipartite 
\item chordal
\item comparability 
\item triangle-free  
\item permutation 
\item bipartite permutation
\item clique
\item threshold
\item proper interval
\item 1-split 
\item augmented clique
\item 2-star
\item bipartite chain 
\item caterpillar
\item trivially perfect 
\item triangle-free \\$\cap$ co-chordal
\item complete bipartite
\end{enumerate}
\end{multicols}

Most of the classes listed are classic classes. Exceptions are 1-split, 2-star and augmented cliques that are small variations of split, stars and cliques respectively. We do not need to define them properly here, hence we refer to~\cite{FeuilloleyH21}.

\paragraph{Easy cases}
For all this paragraph, we refer to\cite{FeuilloleyH21} for references.

For several classes, the result follows from Theorem~\ref{thm:OnePatternThreeVertices} since they can be characterized by a unique pattern these are: forests, linear forests, stars, interval, split, bipartite and chordal graphs, cliques, and all their respective complements.
For some others, there exists a set of patterns characterizing them such that all the patterns can be detected in linear time. In this case, we are also done since we can run the detection algorithm of each pattern independently one after the other. 
This is the case of:
\begin{itemize}
    \item threshold graphs, characterized by \pchordal{} and co-\pchordal{},
    \item proper interval graphs, characterized by \pchordal{} and mirror-\pchordal{},
    \item 1-split graphs, characterized by \psplit{} and co-\psplit{},
    \item Augmented cliques, characterized by \pchordal{} and \psplit{},
    \item 2-stars, characterized by co-\pchordal{} and \pforest{}.
    \item Bipartite chain graphs, characterized by co-\pchordal{} and \pbipartite{}. 
\end{itemize}

We are left with the classes for which every characterization by patterns on three vertices uses one of the hard patterns: \pcomparability{}, co-\pcomparability{}, \ptriangle{} and co-\ptriangle{}.
That is (if we remove the exceptions to the theorem, and up to complement): permutation, bipartite permutation, caterpillar, trivially perfect, triangle-free $\cap$ co-chordal, and complete bipartite. 

\paragraph{Remaining cases}

We prove the remaining cases with five propositions.

\begin{proposition}\label{permutations}
Permutation and bipartite permutation are characterized by a set of patterns that can be detected in linear time.
\end{proposition}

\begin{proof}
Permutation graphs are characterized by a family of two patterns: \pcomparability{} and co-\pcomparability{}.
And the class is stable by complement, which is why we do not mention the complements.
Also, it is enough to consider the case of permutation graphs, since for bipartite permutation graphs we simply add the bipartite pattern, which can be detected in linear time.

For permutation graph the proposition was already proved by~\cite{McConnellS99} (although not in these terms), and in \cite{CorneilDHK16}, the same result was proposed as an application of a general framework on graph searches. 
The idea is to process the given ordered graph to build a new graph $G'$, such that $G=G'$ if and only if the ordering avoids \pcomparability{} and co-\pcomparability{}. 
More precisely, a depth-first search traversal (with specific tie-break rules) yields a new ordering $\theta$ of $G$. Now $\tau, \theta$ define a representation of $G$ as a permutation graph, and to conclude it suffices to check in linear time that this representation corresponds to $G$.\qed
\end{proof}



\begin{proposition}\label{complete-bip}
Complete bipartite graphs and their complements are characterized by a set of patterns that can be detected in linear time.
\end{proposition}

\begin{proof}
These graphs can be characterized by the patterns co-\pchordal{}, co-\pcomparability{} and \pbipartite{}~\cite{FeuilloleyH21}.
The pattern \pbipartite{} implies that the ordering is of the concatenation of two vertex sets $V_1$, $V_2$ such that the graph has a bipartition $(V_1,V_2)$. 
Then, checking completeness can be done in linear time, by a scan of the edge set.
A similar argument works for the complement class (2 disjoint cliques).\qed
\end{proof}

\begin{proposition}\label{caterpillars}
Caterpillars are characterized by a set of patterns that can be detected in linear time. \
\end{proposition}

\begin{proof}

Caterpillars can be characterized by the patterns co-\pcomparability{} and \pforest{}~\cite{FeuilloleyH21}. Let $G$ be a caterpillar and its spine be the path $[x_1, \dots, x_k]$ maximal with the property that $x_1, x_k$ are not pending vertices. It is wel-known that such a  spine can be discovered with 2 consecutive BFS.
Let $\tau$ an ordering of $G$ that avoids the patterns. Consider $\forall i, 1 \leq i \leq k-1$, we may  assume $x_i <_{\tau} x_{i+1}$  then all vertices in $[x_i, x_{i+1}]_{\tau}$ must be adjacent to one of $x_i$ or $x_{i+1}$ using the co-comparability pattern, but using the forests pattern they could only be adjacent to $x_i$. So in this interval we may have pending vertices to $x_i$ (pending vertices to $x_{i+1}$ are excluded). Similarly if $x_l$ with $l \neq i,i+1$ belongs to the interval it must be adjacent to $x_i$ and $l=i-1$ and it cannot have any pendant vertices, a contradiction. So $\tau$ necessarily orders the spine from $x_1$ to $x_k$ and pending vertices attached to $x_i$ with $i>1$, can only be between $x_i$ and $x_{i+1}$. Conversely such an ordering avoids patterns co-\pcomparability{} and \pforest{}.
Let us call such ordering a non interlacing ordering, and this property can be easily  checked in linear time.

For the complement class, we can easily check these properties without computing the complement (BFS for the spine and the interlacing property).

\qed
\end{proof}

\begin{proposition}\label{trivially-perfect}
Trivially perfect graphs are characterized by a set of patterns that can be detected in linear time.
\end{proposition}

\begin{proof}
Trivially perfect graphs can be characterized by the patterns \ptriangle{} and \pchordal{}~\cite{FeuilloleyH21}. 
We already have proposed Algorithm~\ref{algo:chordal} to check for \pchordal{} in linear time. 
For the pattern \ptriangle{}, using the structure of trivially perfect graphs, we just have to check that for every $x, y$ such that $y$ is the first neighbour of x with respect to $\tau$ that $N'(x) =N'(y) \cup y$  where $N'$ means Neighbourhood after $x$ in $\tau$. This also can be done linearly.\qed
\end{proof}

Note that trivially perfect graphs are also called quasi-threshold in \cite{YCC96} since they can be obtained by iteration the two operations: disjoint  union and adding an universal vertex. 
In fact they are comparability graphs of trees and a subclass of cographs and chordal. 
Following  \cite{Chu08} these graphs can be recognized using a unique special LexBFS graph search guided by the non decreasing degree ordering of the vertices. As usual this LexBFS provides a reverse ordering $\tau$ which
avoids the patterns 1 and 2 iff the graph is trivially perfect. 

\begin{proposition}\label{Tfree-cochordal}
Triangle-free co-chordal graphs are characterized by a set of patterns that can be detected in linear time.
\end{proposition}

\begin{proof}
Triangle-free $\cap$ co-chordal can be characterized by the patterns \ptriangle{} and co-\pchordal{}~\cite{FeuilloleyH21}. 
Using Algorithm \ref{algo:co-chordal} we can check for co-\pchordal{}. 
For the triangle, we can use a linear time algorithm to compute a maximum independent set in the complement as given in \cite{RTEL76}, and check whether it is of size at least 3 or not.\qed
\end{proof}

Note that in linear time one can execute a LexBFS in the complement of the graph and produce an ordering that avoids pattern 0 and 3 iff  the graph is in Triangle-free $\cap$ co-chordal. 

\subsection{Finding an ordering versus verfying it}

As a consequence of the theorems above and of the discussions, we get the following corollary.

\begin{corollary}
For every forbidden set of patterns on 3 elements, it is easier to find an ordering that avoids this set of patterns than  to certify that this ordering avoids these patterns.
\end{corollary}

\begin{proof}
Theorem 5 in \cite{FeuilloleyH21} provides a linear time algorithm to produce an ordering that avoids the pattern for each of the 20 over 22 graph classes, moreover these proposed algorithms are all based on classical graph searches that can all be implemented using partition refinement \cite{HabibMPV00} which implies that the result also holds for their complement classes, since partition refinement applies equally on a graph or its complement without constructing the complement.
Therefore
 using Theorem  \ref{thm:three-vertices-class} for 20 graph classes over 22, their ordering can be certified also in linear time. 
 
 To conclude it suffices to  consider the single patterns comparability, triangle-free and their complements. For the patterns triangle free and its complement, all the orderings are equivalent therefore taking one ordering can be done in linear time. But for comparability and its complement to construct in linear time an ordering that avoids the pattern is more sophisticated but can be done using  \cite{McConnellS99}.
 \qed
\end{proof}

Such a statement \emph{does not hold} for patterns with more than 4 vertices. We recall that the classical Roy-Gallai-Hasse-Vitaver theorem
says that a graph G is $k$-colorable if and only if it admits an orientation
with no directed path on $(k + 1)$-vertices (the pattern $P_k$). Already with $k=4$
 the path on 4 vertices pattern which corresponds to the class of 3 colorable graphs that are NP-complete to recognize.  But such a pattern can be certified in linear time. 
\section{Conditional lower bounds for patterns with $k \geq 4$ vertices}
\label{sec:lowerbounds}

In this section, we prove conditional superlinear lower bounds for large families of fully specified patterns with at least four nodes.
Remember that a pattern is fully specified is it has no undecided pair.

Most of our results are derived from a few powerful reduction techniques.
In particular, we obtain conditional quadratic lower bounds for most fully specified patterns on four nodes (resp., for all such patterns under additional restrictions on the type of algorithms considered).

In Section~\ref{sec:o-lower-bounds}, we proved that detecting a pattern on $k$ vertices is in general not harder than finding a $k$-clique in a graph. 
Roughly, we prove in this section that for several cases of {\em fully specified} $k$-node patterns $P$, there is also a converse reduction from the problem of detecting an $f(k)$-clique, for some function $f$, to the $P$-{\sc Detection} problem.

\medskip
We formally introduce our complexity assumptions in Sec.~\ref{sec:lb:hyp}.
Then, in Sec.~\ref{sec:lb:iso} and~\ref{sec:lb:dag}, we present two different types of reductions. 
Applications to the detection of some fully specified patterns are given.
In Sec.~\ref{sec:lb-k=4}, we complete the complexity classification of $P$-{\sc Detection} for the fully specified patterns $P$ with four nodes. 
We end up presenting a few more results for patterns of size at least five in Sec.~\ref{sec:lb>4}.

\medskip
In what follows, a \underline{$H$-based pattern} refers to a fully specified pattern the mandatory edges of which induce $H$. Our exact results and techniques in what follows depend on $H$.

\subsection{Complexity Hypotheses}\label{sec:lb:hyp}

Let us start introducing our complexity assumptions for what follows.
We refer to~\cite{DVW19,LWWW18,Williams18,WilliamsW18} for a thorough discussion about their plausibility, and their implications in the field of fine-grained complexity.

\begin{restatable}{hypothesis}{hypKClique}
\label{hyp:k-clique}
    For every $k \geq 3$, deciding whether a graph contains a clique (an independent set, resp.) on $k$ vertices requires $n^{\omega\left(\lfloor k/3 \rfloor, \lceil k/3 \rceil, \lceil (k-1)/3 \rceil\right) - o(1)}$ time, with $\omega(p,q,r)$ the exponent for multiplying two matrices of respective dimensions $n^p \times n^q$ and $n^q \times n^r$.
\end{restatable}

\begin{restatable}{hypothesis}{hypCombTriangle}
\label{hyp:comb-triangle}
    Every {\em combinatorial} algorithm (not using fast matrix multiplication or other algebraic techniques) for detecting a triangle in a graph requires $n^{3-o(1)}$ time.
\end{restatable}

\begin{restatable}{hypothesis}{hypFourHyperclique}
\label{hyp:4-hyperclique}
    Deciding whether a $3$-uniform hypergraph contains a hyperclique with four nodes requires $n^{4-o(1)}$ time.    
\end{restatable}

In this section, Hypothesis~\ref{hyp:k-clique} is often used for $k=3$ or $k=4$.
In particular, we simply denote by $\omega = \omega(1,1,1)$ the exponent for square matrix multiplication.
Note that Hypothesis~\ref{hyp:k-clique} is often cited in the literature only for cliques.
This is inconsequential for $k \geq 4$, because all stated lower bounds are superquadratic and we can complement a graph in ${\cal O}(n^2)$ time.
However, this might become an issue for $k=3$ under the widely believed conjecture that $\omega = 2$.
Following~\cite{DVW19}, we posit Hypothesis~\ref{hyp:k-clique} for both cliques and independent sets.

We observe that Hypothesis~\ref{hyp:comb-triangle} is less general than our other two complexity assumptions because it only applies to combinatorial algorithms.

\subsection{General reductions from {\sc Induced subgraph isomorphism}}\label{sec:lb:iso}

Recall that the $H$-{\sc Induced-SI} problem asks whether a given host graph $G$ contains $H$ as an induced subgraph.

\begin{proposition}\label{prop:red-isi}
    Let $H$ be a graph on $k$ vertices, for some constant $k$.
    For every $H$-based pattern $P_H$, there is a randomized ${\cal O}(n)$-time reduction from $H$-{\sc Induced-SI} to $P_H$-{\sc Detection}.
\end{proposition}

\begin{proof}
We first describe our reduction.
Let $G$ be an instance of the $H$-{\sc Induced-SI} problem.
We generate a random permutation $\tau$ of the vertices of $G$, which can be done in ${\cal O}(n)$ time using Fisher-Yates shuffle algorithm~\cite{Dur64,FiY63}.
Then, we output that $G$ is a yes-instance of $H$-{\sc Induced-SI} if and only if $\tau$ contains the pattern $P_H$.

If $\tau$ contains $P_H$, then the vertex set of any ordered subgraph that realizes this pattern induces a copy of $H$ in $G$. In particular, $G$ is a yes-instance of $H$-{\sc Induced-SI} (equivalently, we never accept no-instances with our reduction).
Conversely, let us assume the existence of an induced copy of $H$ in $G$, with vertex set $S$.
We totally order the vertices of $S$ such that the resulting ordered graph realizes $P_H$.
The probability for the vertices of $S$ to appear in this exact order in a random permutation of $V(G)$ equals $1/(k!)$.
As a result, we may reject a yes-instance of $H$-{\sc Induced-SI} with probability at most $1 - 1/(k!)$.
 \qed
\end{proof}

We remark that Proposition~\ref{prop:red-isi} also holds more generally for any family ${\cal F}_H$ of $H$-based patterns.

\medskip
Let us illustrate the application of Proposition~\ref{prop:red-isi} to four-node patterns.
There exist eleven $4$-vertex graphs, namely: $K_4$, $C_4$, the diamond, the paw, the claw, the respective complements of these five graphs, and the path $P_4$.
For all such graphs $H$ but $P_4$, there is no known linear-time algorithm for $H$-{\sc Induced-SI}.
Therefore, the existence of a linear-time algorithm for certifying any $H$-based pattern would be a significant algorithmic breakthrough.
We now give more precise lower bounds (also for some patterns of size larger than four), assuming either Hypotheses~\ref{hyp:k-clique},~\ref{hyp:comb-triangle} or~\ref{hyp:4-hyperclique}.

\begin{corollary}\label{cor:k-clique-pattern}
    Assuming Hypothesis~\ref{hyp:k-clique} for $k \geq 3$, if $P$ is the unique $K_k$-based pattern ($\overline{K_k}$-based pattern, resp.), then every $P$-{\sc Detection} algorithm requires $n^{\omega\left(\lfloor k/3 \rfloor, \lceil k/3 \rceil, \lceil (k-1)/3 \rceil\right) - o(1)}$ time.    
\end{corollary}

Note that there exists only one $K_k$-based pattern, and that searching for it in an arbitrary order of the vertices is equivalent to the problem of deciding whether a graph contains a clique on $k$ vertices. 
Therefore, the $n^{\omega\left(\lfloor k/3 \rfloor, \lceil k/3 \rceil, \lceil (k-1)/3 \rceil\right) + o(1)}$ running time is conditionally tight for this pattern.
The same holds true for the unique $\overline{K_k}$-based pattern.

\medskip
It is known that for every $H$ with a clique on $k$ vertices, the $H$-{\sc Induced-SI} problem is at least as hard as $k$-{\sc Clique Detection}~\cite{DVW19}.
Therefore, 

\begin{corollary}\label{cor:red-k-clique}
    Assuming Hypothesis~\ref{hyp:k-clique} for $k \geq 3$, if $H$ is a graph with a clique on $k$ vertices, then $P_H$-{\sc Detection} requires at least $n^{\omega\left(\lfloor k/3 \rfloor, \lceil k/3 \rceil, \lceil (k-1)/3 \rceil\right) - o(1)}$ time for {\em every} $H$-based pattern $P_H$.    
\end{corollary}

\smallskip
The following graphs $H$ contain a triangle: the diamond, the paw, and $K_3+K_1$ (a.k.a., the complement of the claw).
Therefore:

\begin{corollary}\label{cor:red-k-clique-4node}
    Let $H$ be either the diamond, the paw, or $K_3+K_1$.
    Assuming Hypothesis~\ref{hyp:k-clique} for $k=3$, for every $H$-based pattern $P_H$, $P_H$-{\sc Detection} requires at least $n^{\omega-o(1)}$ time.    
\end{corollary}

Up to complementing the input graph in ${\cal O}(n^2)$ time, the result from~\cite{DVW19} also holds for graphs $H$ with an independent set on $k$ vertices, for every $k \geq 4$.
So does our Corollary~\ref{cor:red-k-clique} for their $H$-based patterns.
However, we are not aware of a similar result for graphs $H$ with an independent set on three vertices. 
We prove next that such a result holds if $H$ is a four-vertex graph.

\begin{lemma}
\label{lem:3-indep-set}
    Let $H$ be either $K_2+2K_1$, the co-paw, or the claw.
    There is a linear-time reduction from the problem of detecting an independent set of size three in an $n$-vertex graph to the $H$-{\sc Induced-SI} problem on ${\cal O}(n)$-vertex graphs.
\end{lemma}
\begin{proof}
    Let $G=(V,E)$ be an arbitrary undirected graph. We transform $G$ into a larger graph $G'$ as follows:
    \begin{itemize}
        \item {\it Case $H$ is the claw}. We add a universal vertex, {\it i.e.}, that is adjacent to every vertex of $V$.
        \item {\it Case $H$ equals $K_2+2K_1$}. We construct two disjoint copies $G_1,G_2$ of $G$, calling $v_1,v_2$ the two copies of a same vertex $v \in V$. 
        For every $v \in V$, we add an edge $v_1v_2$. Finally, for every edge $uv \in E$, we add the edges $u_1v_2,v_1u_2$.

        If $x,y,z$ form an independent set of $G$, then $x_1,y_1,z_1,x_2$ induce a copy of $H$ in $G'$.
        Conversely, assume that $G'$ contains an induced copy of $H$.
        In particular, $G'$ also contains an independent set of size three.
        Wlog at least two vertices of it are contained in $V_1$.
        Let us call them $x_1,y_1$.
        By construction, vertices $x,y$ of $G$ are nonadjacent.
        Let $z_i$ be the third vertex of the independent set, for some $i\in \{1,2\}$.
        If $i=1$, then clearly $x,y,z$ form an independent set of $G$.
        Let us now assume that $i=2$.
        Then, $z \notin \{x,y\}$ and $z$ is nonadjacent to both $x,y$ in $G$ (otherwise, there would exist one of the two edges $x_1z_2,y_1z_2$ in $G'$).
        As a result, we also have in this subcase that $x,y,z$ form an independent set of $G$.
        
        \item {\it Case $H$ is the co-paw}. We first compute the modular decomposition of $G$, that can be done in linear time~\cite{TCHP08}.
        By Gallai's theorem~\cite{Gallai67}, either $G$ is disconnected, co-disconnected, or its quotient graph is prime for modular decomposition.
        Assume that $G$ is disconnected.
        If it has at least three connected components, then it contains an independent set on three vertices.
        Otherwise, $G$ is $I_3$-free if and only if both its connected components are cliques.
        From now on, we assume that $G$ is connected.
        Assume that $\overline{G}$ is disconnected.
        Then, $G$ is $I_3$-free if and only if every co-connected component is $I_3$-free.
        Therefore, we are left considering each co-connected component separately.
        From now on, we assume that $G$ is both connected and co-connected.
        Let $G'$ be its quotient graph.
        Each vertex $x$ of $G'$ replaces some strong module $M_x$ of $G$.
        If such a module $M_x$ is not a clique, then there is an independent set of $G$ that is made of two nonadjacent vertices of $M_x$ and of any vertex in some module nonadjacent to $M_x$ (such a nonadjacent module must exist because $G'$ is prime).
        Therefore, we now assume that for every vertex $x$ of $G'$, the module $M_x$ of $G$ is a clique.
        Then, $G$ is $I_3$-free if and only if $G'$ is $I_3$-free.
        Olariu proved that the graph $\overline{G'}$ is paw-free if and only if each of its connected components is triangle-free or complete multipartite~\cite{Olariu88}.
        Recall that in our case, $\overline{G'}$ is connected.
        Since $G'$ is also connected, its complement $\overline{G'}$ is not complete multipartite.
        Therefore, $G'$ is $I_3$-free if and only if it does not contain an induced co-paw.
    \end{itemize}
    Overall, in all three cases, the resulting graph $G'$ contains an induced copy of $H$ if and only if $G$ contains an independent set on three vertices. 
    \qed
\end{proof}

\begin{corollary}\label{cor:complement-4node}
    Let $H$ be either $K_2+2K_1$, the co-paw, or the claw.
    Assuming Hypothesis~\ref{hyp:k-clique} for $k=3$, for every $H$-based pattern $P_H$, $P_H$-{\sc Detection} requires at least $n^{\omega-o(1)}$ time.    
\end{corollary}

We are left considering $H$-based patterns for the following graphs $H$ with four nodes: $C_4$, $2K_2$ and $P_4$.
It has been proved recently that assuming Hypothesis~\ref{hyp:4-hyperclique}, the $C_4$-{\sc Induced-SI} problem requires at least $n^{2-o(1)}$ time, even for graphs with ${\cal O}(n^{1.5})$ edges~\cite{DaW22}.
Therefore:

\begin{corollary}\label{cor:c4}
    Assuming Hypothesis~\ref{hyp:4-hyperclique}, for every $C_4$-based pattern $P$, every $P$-{\sc Detection} algorithm requires at least $n^{2-o(1)}$ time, even if the input graph has ${\cal O}(n^{1.5})$ edges.   
\end{corollary}

\subsection{Topological orderings}\label{sec:lb:dag}\

Being given a $H$-based pattern $P_H$, for some arbitrary $H$, we may orient each edge from its smallest to largest end-vertex.
In doing so, we get an orientation of $H$, of which the pattern represents a topological ordering.
Let us call it the orientation of $H$ associated to $P_H$, which we denote by $\overrightarrow{Or}(P_H)$.
A folklore result is that a directed acyclic graph admits a unique topological ordering if and only if it contains a Hamiltonian directed path.
Next, we prove a general relation between the complexity of detecting a fully specified pattern and the existence of some directed paths in its associated orientation.

\begin{proposition}\label{prop:k-path}
    Let $P$ be a fully specified pattern of constant size.
    If $\overrightarrow{Or}(P)$ contains a unique directed path on $k$ vertices, for some constant $k$, then there is an ${\cal O}(n^2)$-time reduction from $k$-{\sc Clique Detection} on $n$-vertex graphs to $P$-{\sc Detection} on ${\cal O}(n)$-vertex ordered graphs.
    
    In particular, if $k \geq 4$ then assuming Hypothesis~\ref{hyp:k-clique}, every $P$-{\sc Detection} algorithm requires at least $n^{\omega\left(\lfloor k/3 \rfloor, \lceil k/3 \rceil, \lceil (k-1)/3 \rceil\right) - o(1)}$ time.
    If $k \geq 3$ then assuming Hypothesis~\ref{hyp:comb-triangle}, every combinatorial $P$-{\sc Detection} algorithm requires at least ${\cal O}(n^{3-o(1)})$ time.
\end{proposition}
\begin{proof}
    Let $G=(V,E)$ be an instance of $k$-{\sc Clique Detection}.
    We partition $V(P)$ in $\{v^1,v^2,\ldots,v^k\} \cup \{u^1,u^2,\ldots,u^p\}$ so that the vertices $v^1,v^2,\ldots,v^k$ are on the unique directed $k$-path of the orientation associated to $P$.
    Let $Q$ be the $k$-node subpattern of $P$ that is induced by $v^1,v^2,\ldots,v^k$.

    \smallskip
    We first construct an intermediate ordered graph $G_Q$ as follows:
    \begin{enumerate}
        \item $V(G_Q) = V_1 \cup V_2 \cup \ldots \cup V_k$, such that each $V_i$ is a disjoint copy of $V$. For every $v \in V$, we denote by $v_i$ its copy in $V_i$.
        \item $V_1,V_2,\ldots,V_k$ are independent sets.
        \item Then, let $i,j$ be such that $1 \leq i < j \leq k$. If $v^i$ and $v^j$ are adjacent (that is always true if $j=i+1$) then we add the edges $w_ix_j$ for every adjacent vertices $w,x \in V$. Otherwise, we add the edges $w_ix_j$ for every vertices $w,x \in V$ that are either equal or nonadjacent. 
        \item Finally, we equip $G_Q$ with a total ordering $\tau_Q$ such that: for every $1 \leq i < j \leq k$, all vertices of $V_i$ must appear before every vertex of $V_j$.
    \end{enumerate}
    This construction can be done in ${\cal O}((kn)^2)$ time.

    Assume first that $x^1,x^2,\ldots,x^k$ form a complete subgraph of $G$.
    Then, the ordered subgraph of $G_Q$ that is induced by $x^1_1,x^2_2,\ldots,x^k_k$ realizes the pattern $Q$.
    Conversely, assume the existence of an induced ordered subgraph of $G_Q$ that realizes $Q$.
    Since every two consecutive vertices in this subgraph must be adjacent, its $k$ vertices must be, in order, some $x^1_1 \in V_1, \ x^2_2 \in V_2, \ldots, \ x_k^k \in V_k$.
    By construction, $x_1,x_2,\ldots,x_k$ form a complete subgraph of $G$.

    \smallskip
    We construct $G_P$ from $G_Q$ as follows:
    \begin{enumerate}
        \item We add new vertices $y^1,y^2,\ldots,y^p$.
        \item For every $1 \leq i < j \leq p$, there is an edge between $y^i$ and $y^j$ if and only if $u^i$ and $u^j$ are adjacent in $P$. For every $1 \leq i \leq p$ and every $1 \leq j \leq k$, we add edges between $y^i$ and every vertex of $V_j$ if and only if $u^i$ and $v^j$ are adjacent in $P$.
        \item We equip $G_P$ with a total ordering $\tau_P$ whose restriction to $G_Q$ equals $\tau_Q$. For every $1 \leq i < j \leq p$, $\tau_P(y^i) < \tau_P(y^j)$. For every $1 \leq i \leq p$ and every $1 \leq j \leq k$, $y^i$ must be ordered before (after, resp.) every vertex of $V_j$ if and only if $u^i$ is ordered before (after, resp.) $v^j$ in $P$.
    \end{enumerate}
    Recall that if $G$ has a $k$-clique, then there is an ordered subgraph of $G_Q$ that realizes $Q$.
    By construction, we can extend the latter with some vertices of $y^1,y^2,\ldots,y^p$ in order to obtain an ordered subgraph of $G_P$ that realizes $P$.
    Conversely, assume the existence of an ordered subgraph of $G_P$ that realizes $P$. 
    Let $z^1,z^2,\ldots,z^k$ be the $k$ vertices of this ordered subgraph that realize $Q$.
    Note that for every $1 \leq j \leq k$, we must have $\# V_j \cap \{z^1,z^2,\ldots,z^k\} \leq 1$ because vertices in this subset must appear consecutively on the $k$-path and $V_j$ is an independent set. Then, let $a^1,a^2,\ldots,a^k$ be satisfying:
    $$\begin{cases}
    a^i = v^j \ \text{if} \ z^i \in V_j \\
    a^i = u^j \ \text{if} \ z^i = y^j
    \end{cases}$$
    For every $a^i,a^{i+1}$ such that $a^i = v^j$ and $a^{i+1} = v^{j+q}$ for some $q \geq 2$, we further add the sequence $v^{j+1},v^{j+2},\ldots,v^{j+q-1}$ between $a^i$ and $a^{i+1}$.
    Doing so, we obtain a directed path with at least $k$ vertices in $\overrightarrow{Or}(P)$.
    Then, $a^i = v^i$ for every $1 \leq i \leq k$.
    This implies $z^1,z^2,\ldots,z^k \in V(G_Q)$.
    Therefore, $G$ contains a $k$-clique.
     \qed
\end{proof}

Note that if a pattern $P$ contains a clique on $k$ nodes, then in $\overrightarrow{Or}(P)$ there exists a directed path on $k$ vertices.
Nevertheless, such a directed path may exist even if there is no $k$-clique.
For instance, if $H$ denotes either the diamond, the paw or $C_4$, and we assume Hypothesis~\ref{hyp:k-clique}, then according to Proposition~\ref{prop:k-path}, for every $H$-based pattern $P_H$ where the vertices are ordered according to some Hamiltonian path, the $P$-{\sc Detection} problem requires $n^{\omega(2,1,1)-o(1)}$ time.
The latter improves on Corollaries~\ref{cor:red-k-clique-4node} and~\ref{cor:c4}, respectively.

\medskip
We present another application of Proposition~\ref{prop:k-path} to some $P_4$-based patterns.
Recall that we can recognize $P_4$-free graphs (a.k.a., cographs) in linear time~\cite{CorneilPS85}.
Therefore, for the $P_4$-based patterns no meaningful lower bound can be derived from Proposition~\ref{prop:red-isi}.
We prove in this subsection and Sec.~\ref{sec:lb-k=4} that the existence of a subquadratic-time algorithm for $P$-{\sc Detection} is unlikely, even for $P_4$-based patterns.
There are four orientations of $P_4$, each being associated to one or more $P_4$-based patterns.
These four orientations are listed in Fig.~\ref{fig:p4-oriente}.
In this subsection, we address all these orientations but $\overrightarrow{P}(1,1,1)$.

\begin{figure}[!h]
    \centering
    \begin{tikzpicture}

        \node[circle,fill=black,inner sep=0pt,minimum size=2pt, label=above:{$v_1$}] at (-4,0) {};
        \node[circle,fill=black,inner sep=0pt,minimum size=2pt, label=above:{$v_2$}] at (-3,0) {};
        \node[circle,fill=black,inner sep=0pt,minimum size=2pt, label=above:{$v_3$}] at (-2,0) {};
        \node[circle,fill=black,inner sep=0pt,minimum size=2pt, label=above:{$v_4$}] at (-1,0) {};

        \draw[-stealth,thick] (-4,0) -- (-3,0);
        \draw[-stealth,thick] (-3,0) -- (-2,0);
        \draw[-stealth,thick] (-2,0) -- (-1,0);

        \node at (-2.5,-.5) {$\overrightarrow{P_4}$};

         \node[circle,fill=black,inner sep=0pt,minimum size=2pt, label=above:{$v_1$}] at (2,0) {};
        \node[circle,fill=black,inner sep=0pt,minimum size=2pt, label=above:{$v_2$}] at (3,0) {};
        \node[circle,fill=black,inner sep=0pt,minimum size=2pt, label=above:{$v_3$}] at (4,0) {};
        \node[circle,fill=black,inner sep=0pt,minimum size=2pt, label=above:{$v_4$}] at (5,0) {};

        \draw[-stealth,thick] (2,0) -- (3,0);
        \draw[stealth-,thick] (3,0) -- (4,0);
        \draw[-stealth,thick] (4,0) -- (5,0);

        \node at (3.5,-.5) {$\overrightarrow{P}(1,1,1)$};

        \node[circle,fill=black,inner sep=0pt,minimum size=2pt, label=above:{$v_1$}] at (-4,-2) {};
        \node[circle,fill=black,inner sep=0pt,minimum size=2pt, label=above:{$v_2$}] at (-3,-2) {};
        \node[circle,fill=black,inner sep=0pt,minimum size=2pt, label=above:{$v_3$}] at (-2,-2) {};
        \node[circle,fill=black,inner sep=0pt,minimum size=2pt, label=above:{$v_4$}] at (-1,-2) {};

        \draw[stealth-,thick] (-4,-2) -- (-3,-2);
        \draw[stealth-,thick] (-3,-2) -- (-2,-2);
        \draw[-stealth,thick] (-2,-2) -- (-1,-2);

        \node at (-2.5,-2.5) {$\overrightarrow{P}(2,1)$};

         \node[circle,fill=black,inner sep=0pt,minimum size=2pt, label=above:{$v_1$}] at (2,-2) {};
        \node[circle,fill=black,inner sep=0pt,minimum size=2pt, label=above:{$v_2$}] at (3,-2) {};
        \node[circle,fill=black,inner sep=0pt,minimum size=2pt, label=above:{$v_3$}] at (4,-2) {};
        \node[circle,fill=black,inner sep=0pt,minimum size=2pt, label=above:{$v_4$}] at (5,-2) {};

        \draw[-stealth,thick] (2,-2) -- (3,-2);
        \draw[stealth-,thick] (3,-2) -- (4,-2);
        \draw[stealth-,thick] (4,-2) -- (5,-2);

        \node at (3.5,-2.5) {$\overrightarrow{P}(1,2)$};
        
    \end{tikzpicture}
    \caption{The four orientations of $P_4$.}
    \label{fig:p4-oriente}
\end{figure}
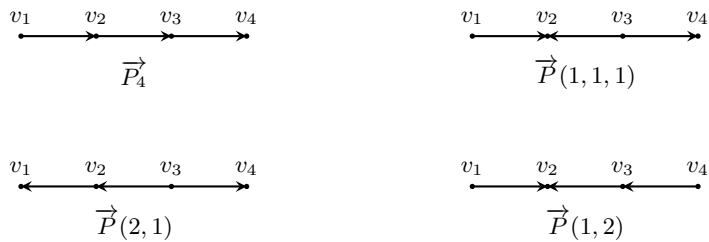

\begin{corollary}\label{cor:p4-1}
    Assuming Hypothesis~\ref{hyp:k-clique} for $k=4$, for the unique pattern $P$ that is associated with $\overrightarrow{P_4}$, $P$-{\sc Detection} requires $n^{\omega(2,1,1)-o(1)}$ time.
\end{corollary}
\begin{proof}
    The result follows from Proposition~\ref{prop:k-path} and the existence of a unique directed $4$-path in $\overrightarrow{P_4}$.
     \qed
\end{proof}

\begin{corollary}\label{cor:p4-2}
    Assuming Hypothesis~\ref{hyp:comb-triangle}, for every $P_4$-based pattern $P$ that is associated with either $\overrightarrow{P}(2,1)$ or $\overrightarrow{P}(1,2)$, every combinatorial $P$-{\sc Detection} algorithm requires at least $n^{3-o(1)}$ time.
\end{corollary}
\begin{proof}
    The result follows from Proposition~\ref{prop:k-path} and the existence of a unique directed $3$-path in $\overrightarrow{P}(2,1)$ and $\overrightarrow{P}(1,2)$.
     \qed
\end{proof}

Finally, the {\em complement} $\overline{P}$ of a pattern $P$ is obtained by replacing every mandatory (forbidden, resp.) edge  of $P$ by a forbidden (mandatory, resp.) of $\overline{P}$.
If an ordering $\tau$ on a graph $G$ avoids $P$, then $\tau$ considered as an ordering on $\overline{G}$ avoids $\overline{P}$.
This implies a trivial ${\cal O}(n^2)$-time reduction from $P$-{\sc Detection} on an $n$-vertex ordered graph to $\overline{P}$-{\sc Detection} on an $n$-vertex ordered graph.
In particular, Proposition~\ref{prop:k-path} also holds for patterns $P$ such that there exists a unique directed $k$-path in $\overrightarrow{Or}(\overline{P})$.

\subsection{Patterns of size four}\label{sec:lb-k=4}

We complete Sec.~\ref{sec:lb:iso} and~\ref{sec:lb:dag} with conditional lower bounds for every fully specified pattern on four nodes.
Let $H$ be an arbitrary four-vertex graph.

\smallskip
The cases $H = K_4$ or $H=\overline{K_4}$ are covered in Corollary~\ref{cor:k-clique-pattern}.

\smallskip
The cases $H$ is the diamond, the paw or $K_3+K_1$ are covered in Corollary~\ref{cor:red-k-clique-4node}.
The cases $H$ is either $K_2+2K_1$, the co-paw, or the claw are covered in Corollary~\ref{cor:complement-4node}.

\smallskip
The case $H=C_4$ is covered in Corollary~\ref{cor:c4}.
We now address the special case of $2K_2$-based patterns.
There are three such patterns, see Fig.~\ref{fig:2k2}.

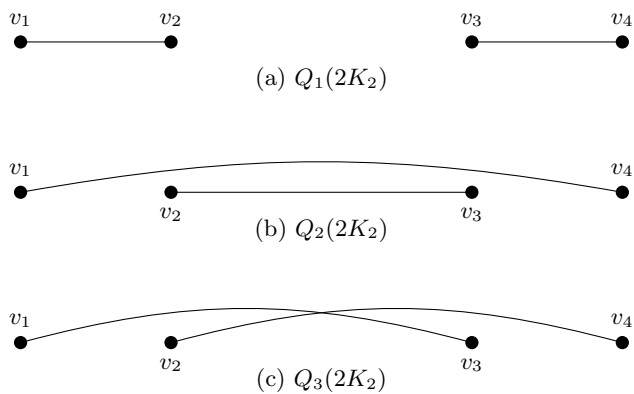
\begin{figure}[!h]
    \centering
    \begin{tikzpicture}

        \node[circle,fill=black,inner sep=0pt,minimum size=5pt, label=above:{$v_1$}] at (-4,0) {};
        \node[circle,fill=black,inner sep=0pt,minimum size=5pt, label=above:{$v_2$}] at (-2,0) {};
        \node[circle,fill=black,inner sep=0pt,minimum size=5pt, label=above:{$v_3$}] at (2,0) {};
        \node[circle,fill=black,inner sep=0pt,minimum size=5pt, label=above:{$v_4$}] at (4,0) {};

        \draw (-4,0) -- (-2,0); \draw (2,0) -- (4,0);

        \node at (0,-.5) {(a) $Q_1(2K_2)$};

        \node[circle,fill=black,inner sep=0pt,minimum size=5pt, label=above:{$v_1$}] at (-4,-2) {};
        \node[circle,fill=black,inner sep=0pt,minimum size=5pt, label=below:{$v_2$}] at (-2,-2) {};
        \node[circle,fill=black,inner sep=0pt,minimum size=5pt, label=below:{$v_3$}] at (2,-2) {};
        \node[circle,fill=black,inner sep=0pt,minimum size=5pt, label=above:{$v_4$}] at (4,-2) {};

        \draw (-4,-2) to[bend left = 10] (4,-2); \draw (-2,-2) -- (2,-2);

        \node at (0,-2.5) {(b) $Q_2(2K_2)$};

        \node[circle,fill=black,inner sep=0pt,minimum size=5pt, label=above:{$v_1$}] at (-4,-4) {};
        \node[circle,fill=black,inner sep=0pt,minimum size=5pt, label=below:{$v_2$}] at (-2,-4) {};
        \node[circle,fill=black,inner sep=0pt,minimum size=5pt, label=below:{$v_3$}] at (2,-4) {};
        \node[circle,fill=black,inner sep=0pt,minimum size=5pt, label=above:{$v_4$}] at (4,-4) {};

        \draw (-4,-4) to[bend left = 15] (2,-4); \draw (-2,-4) to[bend left = 15] (4,-4);

        \node at (0,-4.5) {(c) $Q_3(2K_2)$};

    \end{tikzpicture}
    \caption{All $2K_2$-based patterns (For ease of drawing, forbidden edges are omitted).}
    \label{fig:2k2}
\end{figure}

\begin{proposition}\label{lem:lower-bound-2k2}
    Let $P$ be some $2K_2$-based pattern.
    \begin{itemize}
        \item If $P=Q_3(2K_2)$, then assuming Hypothesis~\ref{hyp:k-clique} for $k=4$, $P$-{\sc Detection} requires at least $n^{\omega(2,1,1)-o(1)}$ time.
        \item Otherwise, assuming Hypothesis~\ref{hyp:comb-triangle}, every combinatorial $P$-{\sc Detection} algorithm requires at least $n^{3-o(1)}$ time.  
    \end{itemize}
\end{proposition}
\begin{proof}
    If $P=Q_3(2K_2)$, then there exists a unique directed $4$-path in $\overrightarrow{Or}(\overline{P})$, and the result follows from Proposition~\ref{prop:k-path}.
    Otherwise, the result follows from Lemmas~\ref{lem:q1-2k2} and~\ref{lem:q2-2k2}, respectively.
     \qed
\end{proof}

Our reductions in what follows are similar for both patterns considered.

\begin{lemma}\label{lem:q1-2k2}
    There is an ${\cal O}(n^2)$-time reduction from {\sc Triangle Detection} to $Q_1(2K_2)$-{\sc Detection} on an ordered graph with ${\cal O}(n)$ vertices.    
\end{lemma}
\begin{proof}
    Let $G=(V,E)$ be an instance of {\sc Triangle Detection}.
    We construct an ordered graph $G'$ as follows:
    \begin{enumerate}
        \item $V(G') = \{u\} \cup V_1 \cup V_2 \cup V_3$, such that each $V_i$ is a disjoint copy of $V$. For every $v \in V$, we denote by $v_i$ its copy in $V_i$.
        \item $N_{G'}(u) = V_1$.
        \item $V_1$ is a clique and $V_2,V_3$ are independent sets.
        \item For every edge $vw \in E$, we add an edge $v_2w_3$.
        \item For every vertex $v \in V$ we add edges $v_1v_2, \ v_1v_3$.
        For every $w \in V \setminus \{v\}$ such that $v,w$ are nonadjacent, we also add edges $v_1w_2, \ v_1w_3$.
        \item Finally, we equip $G'$ with a total ordering $\tau'$ such that: vertex $u$ appears first, followed by all vertices of $V_1$, all vertices of $V_2$, then all vertices of $V_3$.
    \end{enumerate}
    This construction can be done in ${\cal O}(n^2)$ time.

    Assume first the existence of some triangle $xyz$ in $G$. Then, the ordered subgraph of $G'$ induced by $u,x_1,y_2,z_3$ realizes $Q_1(2K_2)$.
    Conversely, assume the existence of an ordered subgraph $T'$ that realizes $Q_1(2K_2)$.
    Since $\{u\} \cup V_1$ induces a clique, the last two vertices of $T'$ (w.r.t. the ordering $\tau'$) must be in $V_2 \cup V_3$.
    In particular, these two vertices must be $y_2,z_3$ for some adjacent vertices $y,z$ of $G$.
    Furthermore, since all vertices of $V_2$ are ordered before all vertices of $V_3$, the two first vertices of $T'$ cannot be also in $V_2 \cup V_3$.
    Hence, at least one such a vertex must be $x_1$, for some vertex $x$ of $G$.
    Since we assume that $T'$ realizes $Q_1(2K_2)$, vertex $x_1$ cannot be adjacent to either $y_2$ nor $z_3$.
    Then, in $G$, by construction vertex $x$ must be adjacent to both $y,z$.
    As a result, there is a triangle in $G$.
     \qed
\end{proof}

\begin{lemma}\label{lem:q2-2k2}
    There is an ${\cal O}(n^2)$-time reduction from {\sc Triangle Detection} to $Q_2(2K_2)$-{\sc Detection} on an ordered graph with ${\cal O}(n)$ vertices.
\end{lemma}
\begin{proof}
    Let $G=(V,E)$ be an instance of {\sc Triangle Detection}.
    We construct an ordered graph $G'$ as follows:
    \begin{enumerate}
        \item $V(G') = V_1 \cup V_2 \cup V_3 \cup V_4$, such that each $V_i$ is a disjoint copy of $V$. For every $v \in V$, we denote by $v_i$ its copy in $V_i$.
        \item $V_1,V_4$ are cliques and $V_2,V_3$ are independent sets.
        \item For every vertex $v \in V$, we add the edges $v_1v_2,v_2v_3,v_3v_4$. For every vertex $w \in V \setminus \{v\}$, we add an edge $v_1w_4$ if $v,w$ are adjacent, and we add edges $v_1w_2,v_3w_4$ if $v,w$ are nonadjacent.
        \item Finally, let $\tau$ be an arbitrary total ordering of $V$, and let $\tau_1,\tau_2,\tau_3,\tau_4$ denote the respective copies of this ordering for $V_1,V_2,V_3,V_4$. We equip $G'$ with a total ordering $\tau'$, which is just the concatenation of $\tau_1,\tau_2,\tau_3,\tau_4$.
    \end{enumerate}
    This construction can be done in ${\cal O}(n^2)$ time.

    Assume first the existence of some triangle $xyz$ in $G$. Then, the ordered subgraph of $G'$ induced by $x_1,y_2,y_3,z_4$ realizes $Q_2(2K_2)$.
    Conversely, assume the existence of an ordered subgraph $T'$ that realizes $Q_2(2K_2)$.
    We consider two different cases in what follows.
    \begin{itemize}
        \item {\it Case $T'$ contains a vertex of $V_1$}. In particular, the first vertex of $T'$ must be some vertex $x_1 \in V_1$. Since $V_1$ is a clique and $x_1$ must be nonadjacent to the second vertex of $T'$, all other three vertices of $T'$ must be in $V_2 \cup V_3 \cup V_4$. Since $V_2$ is an independent set and there is one edge between the second and third vertices of $T'$, not all three last vertices of $T'$ can be contained in $V_2$. In particular, the last vertex of $T'$ must be contained in $V_3 \cup V_4$. Since there is no edge between $V_1$ and $V_3$, and $x_1$ must be adjacent to the last vertex of $T'$, this last vertex must be some $z_4 \in V_4$. Since both $V_2,V_3$ are independent sets, it implies that the second and third vertices of $T'$ are some $y_2 \in V_2$ and $y_3' \in V_3$. Since $y_2,y_3'$ are adjacent, we must have $y = y'$. Since $x_1,z_4$ are adjacent, we must have $xz \in E$. Finally, since $x_1,y_2$ (resp., $y_3,z_4$) are nonadjacent, we must have $xy \in E$ (resp., $yz \in E$). As a result, $xyz$ is a triangle of $G$.
        \item Else, all four vertices of $T'$ must be contained in $V_2 \cup V_3 \cup V_4$. Since all vertices of $V_2$ are ordered before all vertices of $V_3$, and the total orderings $\tau_2,\tau_3$ of $V_2,V_3$ are the same, not all vertices of $T'$ can be contained in $V_2 \cup V_3$ (otherwise, $T'$ could not possibly realize $Q_2(2K_2)$). In particular, the last vertex of $T'$ must be some $z_4 \in V_4$. Recall that $z_4$ must be nonadjacent to the third vertex of $T'$, and so, all other three vertices of $T'$ must be contained in $V_2 \cup V_3$. Furthermore, since there is no edge between $V_2,V_4$, the first vertex of $T'$ cannot be contained in $V_2$. As a result, all other three vertices of $T'$ must be contained in $V_3$. However, the second and third vertices of $T'$ must be adjacent, which contradicts that $V_3$ is an independent set.
    \end{itemize}
    Therefore, only the first case can happen, and it implies the existence of some triangle in $G$.
     \qed
\end{proof}

Finally, we consider the odd case of $P_4$-based patterns $Q$.
The case $\overrightarrow{Or}(Q) = \overrightarrow{P_4}$ is covered in Corollary~\ref{cor:p4-1}.
The cases $\overrightarrow{Or}(Q) = \overrightarrow{P}(2,1)$ and $\overrightarrow{Or}(Q) = \overrightarrow{P}(1,2)$ are covered in Corollary~\ref{cor:p4-2}.
As a result, we are left considering the case $\overrightarrow{Or}(Q) = \overrightarrow{P}(1,1,1)$.
The latter orientation has five different topological orderings (and so, five different patterns are associated with it).
In what follows, the labeling of the nodes of the pattern is the one given in Fig.~\ref{fig:p4-oriente}.

\begin{proposition}
        Let $Q$ be a $P_4$-based pattern that is associated with $\overrightarrow{P}(1,1,1)$.
    \begin{enumerate}
        \item If the vertices of $Q$ are ordered as $v_3,v_1,v_4,v_2$, then assuming Hypothesis~\ref{hyp:k-clique} (for $k=4$), every $Q$-{\sc Detection} algorithm requires $n^{\omega(2,1,1)-o(1)}$ time.
        \item Otherwise, assuming Hypothesis~\ref{hyp:comb-triangle}, every combinatorial algorithm for $Q$-{\sc Detection} requires at least $n^{3-o(1)}$ time.
    \end{enumerate}
\end{proposition}

\begin{proof}
In what follows, let $G=(V,E)$ be an arbitrary graph.
There are several cases to be considered. Each case depends on the total ordering $v_{\sigma(1)},v_{\sigma(2)},v_{\sigma(3)},v_{\sigma(4)}$ of the four vertices in $Q$.  

\smallskip
\underline{Case 1}: $\sigma=(3,1,4,2)$.
There exists a unique directed $4$-path in $\overrightarrow{Or}(\overline{Q})$.
Therefore, Proposition~\ref{prop:k-path} can be applied.

\smallskip
\underline{Case 2}: $\sigma = (3,4,1,2)$.
We construct an ordered graph $G'$ such that:
\begin{enumerate}
    \item $V(G') = V_1 \cup V_2 \cup V_3 \cup V_4$, where each $V_i$ is a disjoint copy of $V$, For every vertex $v \in V$, let $v_i$ denote its copy in $V_i$.
    \item The subsets $V_3,V_2$ are independent sets, while the subsets $V_4,V_1$ are cliques.
    \item For every distinct vertices $v,w \in V$, we add the edge $v_4w_1$, and we add the edges $v_3w_4,v_3w_2,v_1w_2$ if $vw \in E$.
    \item Finally, we equip $G'$ with a total ordering $\tau'$ such that: the vertices of $V_3$ appear first, followed by all vertices of $V_4$, all vertices of $V_1$, then all vertices of $V_2$.
\end{enumerate}
This construction takes ${\cal O}(n^2)$ time.
Furthermore, if $xyz$ is a triangle of $G$, then the ordered subgraph that is induced by $x_3,y_4,y_1,z_2$ realizes $Q$.
Conversely, assume the existence of an induced ordered subgraph $T$ that realizes $Q$.
If every vertex is in a different subset amongst $V_1,V_2,V_3,V_4$, then let $V(T) = (x_3,y_4,y_1',z_2)$.
By construction, $y=y'$ and $xyz$ is a triangle of $G$.
Thus, suppose by contradiction the existence of two vertices of $T$ within a same subset $V_i$.
These two vertices cannot be the second and third vertices of $T$, because the latter are nonadjacent and only $V_3,V_2$ are independent sets.
Therefore, these must be the first and second vertices of $T$, or its third and fourth vertices (both subcases are not necessarily exclusive).
Then, there are only two possibilities: the two first vertices of $T$ are in $V_4$, and at least one of the two last vertices is in $V_1$; or there is at least one of the two first vertices in $V_4$, and the two last vertices are in $V_1$.
However, if the first possibility happens then the third vertex of $T$ (that is in $V_1$) has two nonadjacent vertices in $V_4$, and if the second possibility happens then the second vertex of $T$ (that is in $V_4$) has two nonadjacent vertices in $V_1$.
A contradiction.

\smallskip
\underline{Case 3}: $\sigma=(1,3,2,4)$.
We construct an ordered graph $G'$ such that: 
\begin{enumerate}
    \item $V(G') = V_1 \cup V_2 \cup V_3 \cup V_4$, where each $V_i$ is a disjoint copy of $V$, For every vertex $v \in V$, let $v_i$ denote its copy in $V_i$.
    \item $V_1,V_4$ are cliques, while $V_2,V_3$ are independent sets.
    \item For every vertex $v \in V$, we add an edge $v_2v_3$.
    \item Then, we consider each pair $v,w \in V$ sequentially.
    We add the edge $v_1w_4$ if $v,w$ are either equal or nonadjacent.
    We add the edges $v_1v_2,v_3v_4$, if $v,w$ are adjacent.
    \item Finally, we equip $G'$ with a total ordering $\tau'$ such that: the vertices of $V_1$ appear first, followed by all vertices of $V_3$, all vertices of $V_2$, then all vertices of $V_4$. 
\end{enumerate}
This construction takes ${\cal O}(n^2)$ time.
Furthermore, if $xyz$ is a triangle of $G$, then the ordered subgraph that is induced by $x_1,y_3,y_2,z_4$ realizes $Q$.
Conversely, assume the existence of an induced ordered subgraph $T$ that realizes $Q$.
As in the previous case, one can easily check that if all vertices of $T$ are in different subsets amongst $V_1,V_2,V_3,V_4$, then there is a triangle in $G$.
Therefore, suppose by contradiction the existence of two vertices of $T$ within a same subset $V_i$.
These cannot be the second and third vertices because the latter are adjacent and only $V_1,V_4$ are cliques.
Therefore, these must be the first and second vertices of $T$, or its third and fourth vertices (both subcases are not necessarily exclusive).
In both subcases, these two consecutive vertices are nonadjacent.
Then, there are only two possibilities: the two first vertices of $T$ are in $V_3$, and at least one of the two last vertices is in $V_2$; or there is at least one of the two first vertices in $V_3$, and the two last vertices are in $V_2$.
In particular, the second and third vertices of $T$ are $y_3,y_2$ for some $y \in V$.
However, the two first vertices of $T$ cannot be both in $V_3$ because the third vertex must be adjacent to both. 
Similarly, the two last vertices of $T$ cannot be both in $V_2$ because the second vertex is adjacent to both.
A contradiction.

\underline{Case $4$}: $\sigma = (1,3,4,2)$.
We construct an ordered graph $G'$ such that: 
\begin{enumerate}
    \item $V(G') = V_1 \cup V_2 \cup V_3 \cup V_4$, where each $V_i$ is a disjoint copy of $V$, For every vertex $v \in V$, let $v_i$ denote its copy in $V_i$.
    \item The subsets $V_1,V_2,V_3,V_4$ are cliques.
    \item For every $v \in V$, we add the edge $v_3v_4$.
    \item For every vertices $v,w \in V$, we add the edges $v_1v_2,v_3v_2$ if $v,w$ are adjacent, and we add the edges $v_1v_3,v_1v_4$ if $v,w$ are either equal or nonadjacent.
    \item Finally, we equip $G'$ with a total ordering $\tau'$ such that: the vertices of $V_1$ appear first, followed by all vertices of $V_3$, all vertices of $V_4$, then all vertices of $V_2$. 
\end{enumerate}
This construction takes ${\cal O}(n^2)$ time.
Furthermore, if $xyz$ is a triangle of $G$, then the ordered subgraph that is induced by $x_1,y_3,y_4,z_2$ realizes $Q$.
Conversely, assume the existence of an induced ordered subgraph $T$ that realizes $Q$.
As in the two previous cases, one can easily check that if all vertices of $T$ are in different subsets amongst $V_1,V_2,V_3,V_4$, then there is a triangle in $G$.
Therefore, from now on we assume the existence of two vertices of $T$ within a same subset $V_i$.
Since $V_i$ is a clique, these must be the second and third vertices of $T$.
Then, both vertices are either in $V_3$ or in $V_4$. 
They cannot be contained in $V_4$ because the second and last vertices of $T$ are adjacent whereas there is no edge between $V_4,V_2$. 
Hence, the second and third vertices of $T$ are some $y_3,z_3 \in V_3$.
It implies that the first vertex of $T$ is some $x_1 \in V_1$.
Suppose by contradiction the last vertex of $T$ to be some $t_4 \in V_4$.
Since $z_3,t_4$ are adjacent in $G'$, we obtain by construction $z=t$.
Since $x_1,t_4$ are adjacent in $G'$, by construction $x,z$ are either equal or nonadjacent in $G$.
But then, vertices $x_1,z_3$ should be adjacent in $G'$, a contradiction.
Therefore, the last vertex of $T$ is some $t_2 \in V_2$.
The existence of the edges $x_1t_2,y_3t_2$ in $G'$ implies that $t$ is adjacent to both $x,y$ in $G$.
The nonexistence of the edge $x_1y_3$ in $G'$ also implies that $x,y$ are adjacent in $G$.
As a result, there exists a triangle $xyt$ in $G$.

\smallskip
\underline{Case $5$}: $\sigma = (3,1,2,4)$.
This is the mirror pattern of the previous case.
 \qed
\end{proof}

\subsection{Patterns of size larger than four}\label{sec:lb>4}

We end up presenting another general type of reductions between patterns.
A prefix (suffix, resp.) of a pattern $P$ is the subpattern that is induced by its $i$ first vertices (by its $i$ last vertices, resp.), for some $i \leq |V(P)|$.  

\begin{proposition}
    Let some pattern $P$ be a prefix (suffix, resp.) of some larger pattern $Q$ of constant size $k+1$.
    Then, $Q$-{\sc Detection} is at least as hard as $P$-{\sc Detection} (under randomized ${\cal O}(n)$-time reductions).
\end{proposition}
\begin{proof}
    By symmetry, we may only consider the case where $P$ is a prefix of $Q$.
    Furthermore, it is sufficient to only consider the special case where $|V(Q)| = |V(P)|+1$.
    Let $G$ be an arbitrary ordered graph.
    The ordered graph $G'$ is obtained from $G$ by adding a new vertex $x$, which appears last in the ordering, and such that, for every vertex $v$ of $G$, we add an edge $vx$ independently at random with probability $1/2$.

    Assume first the existence of an ordered subgraph $H$ in $G$ that realizes $P$.
    Let $H'$ be induced by $V(H) \cup \{x\}$.
    In order for $H'$ to realize $Q$, it suffices for vertex $x$ to satisfy some adjacency or nonadjacency relations with some vertices of $H$.
    By construction, all these relations are satisfied with probability at least $1/2^k$.
    Conversely, assume the existence of an ordered subgraph $H'$ in $G'$ that realizes $Q$.
    Since vertex $x$ can only appear in $H'$ as being its last vertex, there is some ordered subgraph $H$ of $H'$, with all vertices in $V(G)$, that must realize~$P$.
     \qed
\end{proof}

For every $k \geq 5$, this above result, combined with all our reductions for fully specified patterns of size four (see Sec.~\ref{sec:lb-k=4}), implies conditional superquadratic lower bounds for a broad family of fully specified patterns of size $k$ (resp, for all these patterns if we are restricting ourselves to combinatorial algorithms). 
Stronger lower bounds can be applied to every pattern such that either Proposition~\ref{prop:red-isi} or~\ref{prop:k-path} can be directly applied.
\section{Parametrized algorithm and merge-width}
\label{sec:parametrized}

In all this section, the forbidden edges have exactly the same role as the mandatory edges, hence for simplicity, we simply say ``edges", and denote the set of edges of a pattern $P$ by $E(P)$ (that is, $E(P)=M(P)\cup F(P)$). We will highlight the place where this can be ambiguous.

\subsection{Anchored pattern operations}
\label{subsec:anchored-operations}

\begin{definition}
    An \emph{anchored pattern} is a pattern with a subset of vertices that are marked as anchors.
\end{definition}

We will now define the operations we consider on patterns. Actually, for each operation, there are two dual forms: one that decompose the pattern in smaller pieces and one that compose smaller pieces into larger ones. 
We define both because on the one hand decomposition operations are natural, but on the other hand our bottom-up algorithm will use the composition operations.

\begin{definition}\label{def:operations}
We define the \emph{pattern operations} has the following operations on anchored patterns.
\begin{enumerate}
    \item (Divide/merge) Given an anchored pattern $P=(V,E)$, the \emph{divide operation} consists creating two patterns $P_1=(V_1,E_1)$ and $P_2=(V_2,E_2)$ such that the following constraints are satisfied:
    \begin{enumerate}
        \item $(E_1,E_2)$ is a partition of $E$ (and the vertex sets inherit the endpoints of the edges, that is, for $i=1,2$, $\forall (x,y)\in E_i$, $x,y\in V_i$).
        \item $V_1 \cup V_2 = V$ 
        \item \label{item:anchors-in-common} Any vertex in $V_1 \cap V_2$ is an anchor in both $V_1$ and $V_2$.
        \item Any vertex that was an anchor in $P$ is still an anchor in the pattern or patterns that has/have received this vertex.
        \item \label{item:consecutive-anchors} There cannot be two consecutive vertices in $V$ such that one is only in $V_1$ and the other only in $V_2$.
    \end{enumerate}
    The reverse operation is called a \emph{merge operation} (a merge is correct only if its reverse operation is a correct divide operation).
    \item (Vertex deletion/creation) If a vertex is isolated in a pattern, the 
    \emph{vertex deletion operation} removes it. 
    After the operation, the vertex just before and just after (if they exist) must be anchors. 
    The reverse operation is a \emph{vertex creation}.
    \item (Edge deletion/creation) If the pattern consists in only one edge, the \emph{edge deletion} removes it (and both its endpoints). The reverse operation is an \emph{edge creation}.  (Here we actually need two cases: one for mandatory edges and one for forbidden edges.)
\end{enumerate}
\end{definition}

For a divide/merge operation, the anchors that are present in both patterns are called \emph{active anchors}, the other are \emph{passive anchors}. Intuitively, the active anchors are the ones that are used for the current merge, whereas the passive anchors are the ones that are needed only for other steps of the process.

An important observation is that when we perform a merge, a vertex that was an active anchor in the two original patterns might not be an anchored pattern in the merged pattern. For example, for the case of the flat cycle described above, when merging 1-2 and 2-3, the vertex 2 lost its anchor status. This happens when a vertex is useful for the current merge but will not be useful later. Since, it will be useful to us to minimize the number of anchors, this is an important aspect. 
The holds for the vertex creation, where the vertices of the outcome might lose their anchor status. 

Constraint~\ref{item:consecutive-anchors} is the key part of this definition and we will justify it in a few paragraphs. 
For now, let us prove the following lemma, in which ``creating a pattern" means creating some anchored version of it, and then removing the anchors to get a standard pattern. 

\begin{lemma}\label{lem:any-pattern-created}
Any pattern can be created by edge creations, vertices creations and merge operations. 
\end{lemma}

\begin{proof}
Consider an arbitrary pattern $P$. 
We describe how we can build it. See Figure~\ref{fig:naive-pattern-construction} for an example. 
Note that here we do not optimize the number of anchors. 
We start by creating all the edges independently, with both endpoints being anchors.
Then, for every such subpattern, we do $k-2$ anchor vertex creations to add all the other vertices of the pattern (on $k$ vertices), in the correct order.  
Then, we merge these subpatterns little by little, until we get the pattern we want. 
All these merges are correct, since all the vertices present are anchored in all the merged intermediate patterns.\qed
\end{proof}

\begin{figure}[!h]
    \centering
    \scalebox{0.8}{
    \input{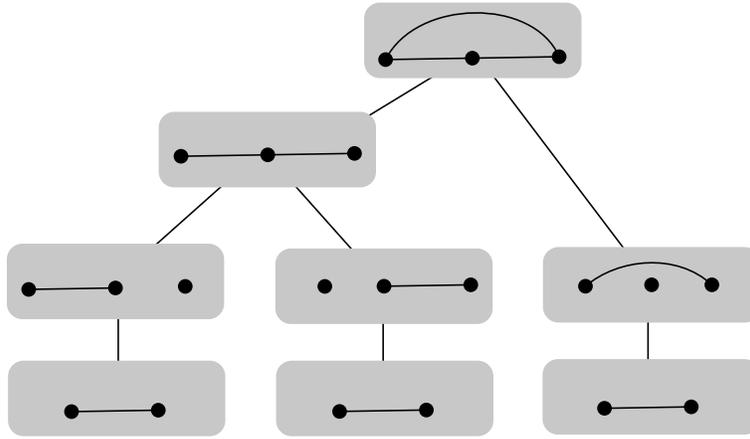}
    }
    \caption{We illustrate the construction of the triangle pattern with the technique of the proof of Lemma~\ref{lem:any-pattern-created}. This is also an illustration of the merge tree defined in Definition~\ref{def:merge-tree}.
    Here, all nodes are anchors at every step.
    We start from the leaves by creating three edges. 
    Then we add vertices, such that the edges are respectively $(1,2)$, $(2,3)$ and $(1,3)$. 
    We then perform two successive merges. Note that here, for the patterns merged are on the same vertex set, thus all anchors are active ; this is a special case, in general we try to minimize the  number of anchors, thus also of active anchors.}
    \label{fig:naive-pattern-construction}
\end{figure}

The key property for our notion of divide/merge is the following.

\begin{lemma}\label{lem:consistent-merge}
Let $P$ be a pattern that can be divided into $P_1$ and $P_2$, and $G$ be an ordered graph.
If there is a realization of $P_1$ and a realization of $P_2$ in $G$, that agree on the positions of the active anchors, then there is a realization of $P$ in $G$.
\end{lemma}

Before we prove this statement, let us point to Figure~\ref{fig:wrong-merge-bis} that illustrates why this statement would not be true without Item~\ref{item:consecutive-anchors} in Definition~\ref{def:operations}.

\begin{figure}[!h]
    \centering
    \scalebox{0.85}{
    \input{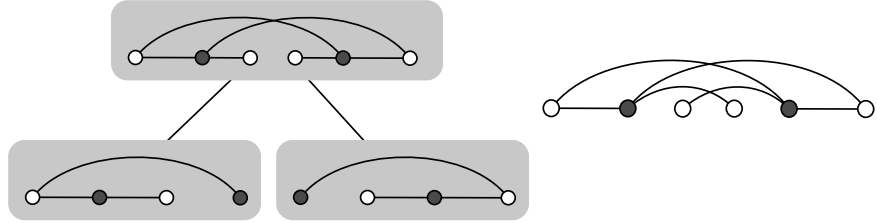}
    }
    \caption{Illustration of why we need Condition~\ref{item:consecutive-anchors} in Definition~\ref{def:operations}. 
    Suppose we decide only keep the anchors that are identified during a merge. 
    On the right, we have such a modified merge, where the black vertices are the anchors. 
    Now, consider the ordered graph on the right. It does have a realization of the two subpatterns, with the correct positions for the anchors, but it does not have a realization of the merged pattern. The problem being the interleaving of the non-anchored vertices.
    An alternative definition of a correct merge would be one that is unambiguous: there is only one pattern that can originates from it.}
    \label{fig:wrong-merge-bis}
\end{figure}

\begin{proof}
Consider $P$, $P_1$, $P_2$ and $G$ as in the lemma.
Consider the realizations of $P_1$ and $P_2$ that agree on the set of active anchors. 
Let $H$ be the union of these two realizations.  
All the vertices in $H$ that are not active anchors are different, by Item~\ref{item:anchors-in-common} thus $H$ has the same number of vertices as $P$, with the same active anchors. 
In $H$, between two active anchors (and before the first active anchor, after the last active anchor), all the vertices originate from the realization of one of the two pattern only because of Item~\ref{item:consecutive-anchors}. Since the active anchors of $H$ and of $P$ are consistent, and the ordering of the vertices of the realization of $P_1$ (resp. $P_2$) is consistent with $P$, and every segment between active anchors belongs only to one of the realization, 
the vertices of $H$ appear in the same order as in $P$. And the then the edges are also consistent, which means that $H$ is a realization of $P$.\qed
\end{proof}

\subsection{Merge trees and merge-width}

Based on the operations listed above, we define a natural notion of \emph{merge tree}, that formalizes the intuition given in Figure~\ref{fig:naive-pattern-construction}. 

\begin{definition}\label{def:merge-tree}
    A \emph{merge tree} for a pattern $P$ is a tree where:
    \begin{enumerate}
        \item Every node is labeled with an anchored pattern. 
        \item The root is labeled with an anchored version of pattern $P$. 
        \item Every node is in one of the following cases:
        \begin{enumerate}
            \item It is a leaf, and then it is labeled with a unique edge.
            \item It has a unique child, and it is labeled with a pattern that can be obtained via a vertex creation from its child's pattern.
            \item It has two children, and it is labeled by a pattern that can be obtained by a merge of the patterns of its children.
        \end{enumerate}
    \end{enumerate}
\end{definition}

We can now define the width of a merge tree as well as our parameter \emph{merge-width}. 

\begin{definition}
The width of a merge tree is the maximum, over all nodes of the merge tree, of the number of anchors in the pattern that labels this node.
The \emph{merge-width} of a pattern $P$ is the minimum, over all correct merge trees of $P$, of the width of that tree.
\end{definition} 

This parameter follows from the fact that in the algorithm will define next, the complexity of the dynamic programming will depend on the number of anchors. Remember that in our merge and vertex creation operation (and unlike in the proof of Lemma~\ref{lem:any-pattern-created}), a vertex could lose its anchor status, which is the way we can hope to maintain a small number of anchors.

\subsection{Algorithm}

Consider an ordered graph $G$, a pattern $P$ and its merge tree. The vertices of $G$ are named by their order in the graph.

\paragraph{General idea.}

The general idea of the algorithm is the following. 
For every node of the merge tree, labeled with an anchored subpattern $P'$, with anchors $a_1,...,a_\ell$, we want to list all the $\ell$-tuple of vertices of $G$ such that there exists a realization of $P'$ in $G$ with anchors at these positions. 
Given this, it is easy to answer our original problem: the list for the root pattern $P$ is non-empty, if and only if there is no occurrence of $P$ in $G$.  

We will process the nodes bottom-up in the tree, using the list of the children to create the list of the parent.  

\paragraph{Notations.}
Consider a node $u$ of the tree, labeled with the anchored pattern $P_u$, having $\ell$ anchors.
We will store the $\ell$-tuples mentioned above in the form of a $\ell$-dimensional matrix $T_u$, where in every dimension the size of $T_u$ is $n$. 
After we have processed node $u$, this matrix must satisfy the following:
\begin{itemize}
    \item $T_u[v_1,...,v_\ell]=1$ if there exists an occurrence of $P_u$ in $G$ for which the corresponding anchors are at the positions $v_1, ..., v_\ell$
    \item $T_u[v_1,...,v_a]=0$ otherwise.
\end{itemize}

In the next paragraph, we describe how we compute these matrices, depending on the type of operation used.

\paragraph{Edge creation nodes.}

The leaf nodes of the merge tree correspond to edge creation. The matrices are filled the following way:
\begin{itemize}
    \item If both endpoints of the edge pattern are anchors, then the matrix is 2-dimensional, and has a 1, at every position $(i,j)$ such that $i<j$, and $(i,j)\in E$.
    \item If only the right vertex (resp. left vertex) is an anchor, then the matrix is 1-dimensional and $T_u[y]=1$, if and only if, there exists an edge $(x,y)$ (resp. an edge $(y,x)$). 
\end{itemize}

\paragraph{Vertex creation nodes.}

Before the creation of the vertex, the two surrounding vertices must be anchors by definition. 
The matrix that we manipulate depends on which vertices are anchors after the creation. 
Let us consider first the case where the new vertex and both its surrounding vertices are anchors after the creation. 
Let $P$ be the pattern after creation, and let $s$ be the position of the newly created vertex. 

Let $u$ be the node after creation, and let $w$ be its children
The matrix $T_u$ is filled the following way:
$T_u[v_1,...,v_{s-1},v_s,v_{s+1},...,v_\ell]=1$ if and only if $T_w[v_1,...,v_{s-1},v_{s+1},...,v_\ell]=1$, and $v_{s-1},v_{s+1}$ are not consecutive vertices in $G$.

Now, if some vertices in $\{s-1,s,s+1\}$ are not anchors in $P_u$, we do the analogous operation, but the matrix $T_u$ does not contain the field corresponding to these. 






\paragraph{Example of merge node.}

Before explaining the general algorithm to merge two anchored patterns, we present a small example through Figure \ref{fig:merge-example}.

\begin{figure}[ht]
    \centering
    \scalebox{1}{
    \tikzset{every picture/.style={line width=0.75pt}} 

\begin{tikzpicture}[x=0.75pt,y=0.75pt,yscale=-1,xscale=1]

\draw    (75,175) -- (135,95) ;
\draw    (135,95) -- (195,175) ;
\draw  [draw opacity=0][fill={rgb, 255:red, 200; green, 200; blue, 200 }  ,fill opacity=1 ] (60,72) .. controls (60,65.37) and (65.37,60) .. (72,60) -- (198,60) .. controls (204.63,60) and (210,65.37) .. (210,72) -- (210,108) .. controls (210,114.63) and (204.63,120) .. (198,120) -- (72,120) .. controls (65.37,120) and (60,114.63) .. (60,108) -- cycle ;
\draw  [fill={rgb, 255:red, 0; green, 0; blue, 0 }  ,fill opacity=1 ] (70,95) .. controls (70,92.24) and (72.24,90) .. (75,90) .. controls (77.76,90) and (80,92.24) .. (80,95) .. controls (80,97.76) and (77.76,100) .. (75,100) .. controls (72.24,100) and (70,97.76) .. (70,95) -- cycle ;
\draw  [fill={rgb, 255:red, 0; green, 0; blue, 0 }  ,fill opacity=1 ] (190,95) .. controls (190,92.24) and (192.24,90) .. (195,90) .. controls (197.76,90) and (200,92.24) .. (200,95) .. controls (200,97.76) and (197.76,100) .. (195,100) .. controls (192.24,100) and (190,97.76) .. (190,95) -- cycle ;
\draw    (75,95) -- (195,95) ;
\draw  [fill={rgb, 255:red, 255; green, 255; blue, 255 }  ,fill opacity=1 ] (100,95) .. controls (100,92.24) and (102.24,90) .. (105,90) .. controls (107.76,90) and (110,92.24) .. (110,95) .. controls (110,97.76) and (107.76,100) .. (105,100) .. controls (102.24,100) and (100,97.76) .. (100,95) -- cycle ;
\draw  [fill={rgb, 255:red, 255; green, 255; blue, 255 }  ,fill opacity=1 ] (160,95) .. controls (160,92.24) and (162.24,90) .. (165,90) .. controls (167.76,90) and (170,92.24) .. (170,95) .. controls (170,97.76) and (167.76,100) .. (165,100) .. controls (162.24,100) and (160,97.76) .. (160,95) -- cycle ;
\draw  [draw opacity=0] (75,95) .. controls (75,95) and (75,95) .. (75,95) .. controls (75,81.19) and (88.43,70) .. (105,70) .. controls (121.57,70) and (135,81.19) .. (135,95) -- (105,95) -- cycle ; \draw   (75,95) .. controls (75,95) and (75,95) .. (75,95) .. controls (75,81.19) and (88.43,70) .. (105,70) .. controls (121.57,70) and (135,81.19) .. (135,95) ;  
\draw  [fill={rgb, 255:red, 255; green, 255; blue, 255 }  ,fill opacity=1 ] (130,95) .. controls (130,92.24) and (132.24,90) .. (135,90) .. controls (137.76,90) and (140,92.24) .. (140,95) .. controls (140,97.76) and (137.76,100) .. (135,100) .. controls (132.24,100) and (130,97.76) .. (130,95) -- cycle ;
\draw  [draw opacity=0][fill={rgb, 255:red, 200; green, 200; blue, 200 }  ,fill opacity=1 ] (30,162) .. controls (30,155.37) and (35.37,150) .. (42,150) -- (108,150) .. controls (114.63,150) and (120,155.37) .. (120,162) -- (120,198) .. controls (120,204.63) and (114.63,210) .. (108,210) -- (42,210) .. controls (35.37,210) and (30,204.63) .. (30,198) -- cycle ;
\draw  [fill={rgb, 255:red, 0; green, 0; blue, 0 }  ,fill opacity=1 ] (40,185) .. controls (40,182.24) and (42.24,180) .. (45,180) .. controls (47.76,180) and (50,182.24) .. (50,185) .. controls (50,187.76) and (47.76,190) .. (45,190) .. controls (42.24,190) and (40,187.76) .. (40,185) -- cycle ;
\draw    (45,185) -- (105,185) ;
\draw  [fill={rgb, 255:red, 255; green, 255; blue, 255 }  ,fill opacity=1 ] (70,185) .. controls (70,182.24) and (72.24,180) .. (75,180) .. controls (77.76,180) and (80,182.24) .. (80,185) .. controls (80,187.76) and (77.76,190) .. (75,190) .. controls (72.24,190) and (70,187.76) .. (70,185) -- cycle ;
\draw  [draw opacity=0] (45,185) .. controls (45,185) and (45,185) .. (45,185) .. controls (45,171.19) and (58.43,160) .. (75,160) .. controls (91.57,160) and (105,171.19) .. (105,185) -- (75,185) -- cycle ; \draw   (45,185) .. controls (45,185) and (45,185) .. (45,185) .. controls (45,171.19) and (58.43,160) .. (75,160) .. controls (91.57,160) and (105,171.19) .. (105,185) ;  
\draw  [fill={rgb, 255:red, 0; green, 0; blue, 0 }  ,fill opacity=1 ] (100,185) .. controls (100,182.24) and (102.24,180) .. (105,180) .. controls (107.76,180) and (110,182.24) .. (110,185) .. controls (110,187.76) and (107.76,190) .. (105,190) .. controls (102.24,190) and (100,187.76) .. (100,185) -- cycle ;
\draw  [draw opacity=0][fill={rgb, 255:red, 200; green, 200; blue, 200 }  ,fill opacity=1 ] (150,162) .. controls (150,155.37) and (155.37,150) .. (162,150) -- (228,150) .. controls (234.63,150) and (240,155.37) .. (240,162) -- (240,198) .. controls (240,204.63) and (234.63,210) .. (228,210) -- (162,210) .. controls (155.37,210) and (150,204.63) .. (150,198) -- cycle ;
\draw  [fill={rgb, 255:red, 0; green, 0; blue, 0 }  ,fill opacity=1 ] (160,185) .. controls (160,182.24) and (162.24,180) .. (165,180) .. controls (167.76,180) and (170,182.24) .. (170,185) .. controls (170,187.76) and (167.76,190) .. (165,190) .. controls (162.24,190) and (160,187.76) .. (160,185) -- cycle ;
\draw    (165,185) -- (225,185) ;
\draw  [fill={rgb, 255:red, 255; green, 255; blue, 255 }  ,fill opacity=1 ] (190,185) .. controls (190,182.24) and (192.24,180) .. (195,180) .. controls (197.76,180) and (200,182.24) .. (200,185) .. controls (200,187.76) and (197.76,190) .. (195,190) .. controls (192.24,190) and (190,187.76) .. (190,185) -- cycle ;
\draw  [fill={rgb, 255:red, 0; green, 0; blue, 0 }  ,fill opacity=1 ] (220,185) .. controls (220,182.24) and (222.24,180) .. (225,180) .. controls (227.76,180) and (230,182.24) .. (230,185) .. controls (230,187.76) and (227.76,190) .. (225,190) .. controls (222.24,190) and (220,187.76) .. (220,185) -- cycle ;

\draw (129,45.4) node [anchor=north west][inner sep=0.75pt]    {$P$};
\draw (64.5,215.4) node [anchor=north west][inner sep=0.75pt]    {$P_{u}$};
\draw (184.5,215.4) node [anchor=north west][inner sep=0.75pt]    {$P_{v}$};
\draw (67,104.4) node [anchor=north west][inner sep=0.75pt]    {$a_{1}$};
\draw (187,104.4) node [anchor=north west][inner sep=0.75pt]    {$a_{2}$};
\draw (130,101.4) node [anchor=north west][inner sep=0.75pt]    {$b$};
\draw (37,194.4) node [anchor=north west][inner sep=0.75pt]    {$a_{3}$};
\draw (97,194.4) node [anchor=north west][inner sep=0.75pt]    {$a_{4}$};
\draw (157,194.4) node [anchor=north west][inner sep=0.75pt]    {$a_{5}$};
\draw (217,194.4) node [anchor=north west][inner sep=0.75pt]    {$a_{6}$};

\end{tikzpicture}
    }
    \caption{Example of a merge. Anchors are in black, and other vertices in white.
    Anchors $a_4$ in pattern $P_u$ and $a_5$ in pattern $P_v$ are identified as vertex $b$, to produce pattern $P$. These two anchors are active anchors, as they are needed for the merge. Other anchors are passive, meaning they are not needed for this merge, but may be needed for a future merge not depicted in this figure.}
    \label{fig:merge-example}
\end{figure}
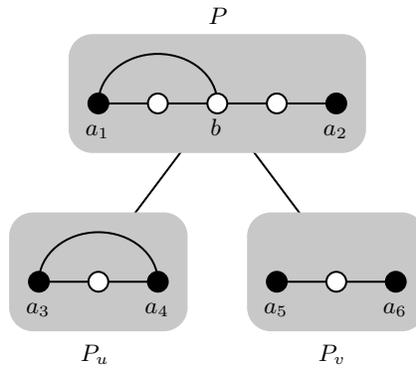

We assume we have computed matrices $T_u$ corresponding to $P_u$ and matrix $T_v$ corresponding to $P_v$.
This means we have a 2-dimensional matrix $T_u$, filled with 0's and 1's, such that for every vertices $v_1, v_2$ in the vertex set of the graph $G$, $T_u[v_1, v_2] = 1$ if and only if $P_u$ can be realized in $G$ with $v_1 = a_3$ and $v_2 = a_4$. We can use a standard 2-dimensional matrix because $P_u$ is anchored with 2 vertices.
Similarly, we have a 2-dimensional matrix $T_v$ such that for every vertices $v_1, v_2$ in the vertex set of the graph $G$, $T_v[v_1, v_2] = 1$ if and only if $P_v$ can be realized in $G$ with $v_1 = a_5$ and $_2 = a_6$.

We want to compute the content of matrix $T$, corresponding to pattern $P$.
To compute entry $T[v_1, v_2]$, we must know if $P$ can be realized in $G$ with $v_1 = a_1$ and $v_2 = a_2$.
To do so, we discuss on vertex $b$ of $P$. It must be some vertex $v_3$ of $G$.
Hence, $T[v_1, v_2] = \max_{v_3 \in V(G)} T_u[v_1, v_3] \cdot T_v[v_3, v_2]$.
This is the formula for matrix multiplication, but with a maximum instead of a sum. This can be addresses by doing a classical matrix multiplication, and then replacing a non-zero answer by a 1.
Using this technique, we can compute $T$ from $T_u$ and $T_v$ as fast as matrix multiplication. Said with merge vocabulary, we can merge $P_u$ and $P_v$ into $P$ as fast as matrix multiplication.

The goal is to extend this idea to any merge. 
In our example, it is relatively natural because it corresponded to the the classic matrix product. 
When addressing patterns with more anchors, we need to compute a $d$-dimensional matrix from a $d_1$-dimensional and a $d_2$-dimensional matrix, and matrix multiplication cannot be performed as such. Let's see what can be done through a second example, depicted in Figure \ref{fig:merge-exemple-big}.
In this case, the formula becomes :
\begin{align*}
    T[v_1, v_2] = \max_{v_3, v_4 \in V(G)} T_u[v_1, v_3, v_4] \cdot T_v[v_3, v_4, v_2].
\end{align*}

\begin{figure}[ht]
    \centering
    \scalebox{1}{
    \tikzset{every picture/.style={line width=0.75pt}} 

\begin{tikzpicture}[x=0.75pt,y=0.75pt,yscale=-1,xscale=1]

\draw    (95,195) -- (155,115) ;
\draw    (155,115) -- (215,195) ;
\draw  [draw opacity=0][fill={rgb, 255:red, 200; green, 200; blue, 200 }  ,fill opacity=1 ] (65,92) .. controls (65,85.37) and (70.37,80) .. (77,80) -- (233,80) .. controls (239.63,80) and (245,85.37) .. (245,92) -- (245,128) .. controls (245,134.63) and (239.63,140) .. (233,140) -- (77,140) .. controls (70.37,140) and (65,134.63) .. (65,128) -- cycle ;
\draw  [fill={rgb, 255:red, 0; green, 0; blue, 0 }  ,fill opacity=1 ] (225,115) .. controls (225,112.24) and (227.24,110) .. (230,110) .. controls (232.76,110) and (235,112.24) .. (235,115) .. controls (235,117.76) and (232.76,120) .. (230,120) .. controls (227.24,120) and (225,117.76) .. (225,115) -- cycle ;
\draw    (140,115) -- (230,115) ;
\draw  [draw opacity=0] (110,115) .. controls (110,115) and (110,115) .. (110,115) .. controls (110,101.19) and (130.15,90) .. (155,90) .. controls (179.85,90) and (200,101.19) .. (200,115) -- (155,115) -- cycle ; \draw   (110,115) .. controls (110,115) and (110,115) .. (110,115) .. controls (110,101.19) and (130.15,90) .. (155,90) .. controls (179.85,90) and (200,101.19) .. (200,115) ;  
\draw  [fill={rgb, 255:red, 255; green, 255; blue, 255 }  ,fill opacity=1 ] (165,115) .. controls (165,112.24) and (167.24,110) .. (170,110) .. controls (172.76,110) and (175,112.24) .. (175,115) .. controls (175,117.76) and (172.76,120) .. (170,120) .. controls (167.24,120) and (165,117.76) .. (165,115) -- cycle ;
\draw  [draw opacity=0][fill={rgb, 255:red, 200; green, 200; blue, 200 }  ,fill opacity=1 ] (20,182) .. controls (20,175.37) and (25.37,170) .. (32,170) -- (128,170) .. controls (134.63,170) and (140,175.37) .. (140,182) -- (140,218) .. controls (140,224.63) and (134.63,230) .. (128,230) -- (32,230) .. controls (25.37,230) and (20,224.63) .. (20,218) -- cycle ;
\draw  [draw opacity=0] (65,205) .. controls (65,205) and (65,205) .. (65,205) .. controls (65,191.19) and (78.43,180) .. (95,180) .. controls (111.57,180) and (125,191.19) .. (125,205) -- (95,205) -- cycle ; \draw   (65,205) .. controls (65,205) and (65,205) .. (65,205) .. controls (65,191.19) and (78.43,180) .. (95,180) .. controls (111.57,180) and (125,191.19) .. (125,205) ;  
\draw  [fill={rgb, 255:red, 0; green, 0; blue, 0 }  ,fill opacity=1 ] (120,205) .. controls (120,202.24) and (122.24,200) .. (125,200) .. controls (127.76,200) and (130,202.24) .. (130,205) .. controls (130,207.76) and (127.76,210) .. (125,210) .. controls (122.24,210) and (120,207.76) .. (120,205) -- cycle ;
\draw  [draw opacity=0][fill={rgb, 255:red, 200; green, 200; blue, 200 }  ,fill opacity=1 ] (170,182) .. controls (170,175.37) and (175.37,170) .. (182,170) -- (278,170) .. controls (284.63,170) and (290,175.37) .. (290,182) -- (290,218) .. controls (290,224.63) and (284.63,230) .. (278,230) -- (182,230) .. controls (175.37,230) and (170,224.63) .. (170,218) -- cycle ;
\draw  [fill={rgb, 255:red, 0; green, 0; blue, 0 }  ,fill opacity=1 ] (180,205) .. controls (180,202.24) and (182.24,200) .. (185,200) .. controls (187.76,200) and (190,202.24) .. (190,205) .. controls (190,207.76) and (187.76,210) .. (185,210) .. controls (182.24,210) and (180,207.76) .. (180,205) -- cycle ;
\draw    (185,205) -- (275,205) ;
\draw  [fill={rgb, 255:red, 255; green, 255; blue, 255 }  ,fill opacity=1 ] (210,205) .. controls (210,202.24) and (212.24,200) .. (215,200) .. controls (217.76,200) and (220,202.24) .. (220,205) .. controls (220,207.76) and (217.76,210) .. (215,210) .. controls (212.24,210) and (210,207.76) .. (210,205) -- cycle ;
\draw  [fill={rgb, 255:red, 0; green, 0; blue, 0 }  ,fill opacity=1 ] (240,205) .. controls (240,202.24) and (242.24,200) .. (245,200) .. controls (247.76,200) and (250,202.24) .. (250,205) .. controls (250,207.76) and (247.76,210) .. (245,210) .. controls (242.24,210) and (240,207.76) .. (240,205) -- cycle ;
\draw  [fill={rgb, 255:red, 0; green, 0; blue, 0 }  ,fill opacity=1 ] (30,205) .. controls (30,202.24) and (32.24,200) .. (35,200) .. controls (37.76,200) and (40,202.24) .. (40,205) .. controls (40,207.76) and (37.76,210) .. (35,210) .. controls (32.24,210) and (30,207.76) .. (30,205) -- cycle ;
\draw  [draw opacity=0] (35,205) .. controls (35,205) and (35,205) .. (35,205) .. controls (35,191.19) and (48.43,180) .. (65,180) .. controls (81.57,180) and (95,191.19) .. (95,205) -- (65,205) -- cycle ; \draw   (35,205) .. controls (35,205) and (35,205) .. (35,205) .. controls (35,191.19) and (48.43,180) .. (65,180) .. controls (81.57,180) and (95,191.19) .. (95,205) ;  
\draw  [fill={rgb, 255:red, 0; green, 0; blue, 0 }  ,fill opacity=1 ] (90,205) .. controls (90,202.24) and (92.24,200) .. (95,200) .. controls (97.76,200) and (100,202.24) .. (100,205) .. controls (100,207.76) and (97.76,210) .. (95,210) .. controls (92.24,210) and (90,207.76) .. (90,205) -- cycle ;
\draw  [fill={rgb, 255:red, 255; green, 255; blue, 255 }  ,fill opacity=1 ] (60,205) .. controls (60,202.24) and (62.24,200) .. (65,200) .. controls (67.76,200) and (70,202.24) .. (70,205) .. controls (70,207.76) and (67.76,210) .. (65,210) .. controls (62.24,210) and (60,207.76) .. (60,205) -- cycle ;
\draw  [fill={rgb, 255:red, 0; green, 0; blue, 0 }  ,fill opacity=1 ] (270,205) .. controls (270,202.24) and (272.24,200) .. (275,200) .. controls (277.76,200) and (280,202.24) .. (280,205) .. controls (280,207.76) and (277.76,210) .. (275,210) .. controls (272.24,210) and (270,207.76) .. (270,205) -- cycle ;
\draw  [fill={rgb, 255:red, 0; green, 0; blue, 0 }  ,fill opacity=1 ] (75,115) .. controls (75,112.24) and (77.24,110) .. (80,110) .. controls (82.76,110) and (85,112.24) .. (85,115) .. controls (85,117.76) and (82.76,120) .. (80,120) .. controls (77.24,120) and (75,117.76) .. (75,115) -- cycle ;
\draw  [draw opacity=0] (80,115) .. controls (80,115) and (80,115) .. (80,115) .. controls (80,101.19) and (93.43,90) .. (110,90) .. controls (126.57,90) and (140,101.19) .. (140,115) -- (110,115) -- cycle ; \draw   (80,115) .. controls (80,115) and (80,115) .. (80,115) .. controls (80,101.19) and (93.43,90) .. (110,90) .. controls (126.57,90) and (140,101.19) .. (140,115) ;  
\draw  [fill={rgb, 255:red, 255; green, 255; blue, 255 }  ,fill opacity=1 ] (105,115) .. controls (105,112.24) and (107.24,110) .. (110,110) .. controls (112.76,110) and (115,112.24) .. (115,115) .. controls (115,117.76) and (112.76,120) .. (110,120) .. controls (107.24,120) and (105,117.76) .. (105,115) -- cycle ;
\draw  [fill={rgb, 255:red, 255; green, 255; blue, 255 }  ,fill opacity=1 ] (135,115) .. controls (135,112.24) and (137.24,110) .. (140,110) .. controls (142.76,110) and (145,112.24) .. (145,115) .. controls (145,117.76) and (142.76,120) .. (140,120) .. controls (137.24,120) and (135,117.76) .. (135,115) -- cycle ;
\draw  [fill={rgb, 255:red, 255; green, 255; blue, 255 }  ,fill opacity=1 ] (195,115) .. controls (195,112.24) and (197.24,110) .. (200,110) .. controls (202.76,110) and (205,112.24) .. (205,115) .. controls (205,117.76) and (202.76,120) .. (200,120) .. controls (197.24,120) and (195,117.76) .. (195,115) -- cycle ;

\draw (149,65.4) node [anchor=north west][inner sep=0.75pt]    {$P$};
\draw (69.5,235.4) node [anchor=north west][inner sep=0.75pt]    {$P_{u}$};
\draw (220,235.4) node [anchor=north west][inner sep=0.75pt]    {$P_{v}$};
\draw (72,124.4) node [anchor=north west][inner sep=0.75pt]    {$a_{1}$};
\draw (222,124.4) node [anchor=north west][inner sep=0.75pt]    {$a_{2}$};
\draw (27,214.4) node [anchor=north west][inner sep=0.75pt]    {$a_{3}$};
\draw (117,214.4) node [anchor=north west][inner sep=0.75pt]    {$a_{5}$};
\draw (177,214.4) node [anchor=north west][inner sep=0.75pt]    {$a_{6}$};
\draw (237,214.4) node [anchor=north west][inner sep=0.75pt]    {$a_{7}$};
\draw (87,214.4) node [anchor=north west][inner sep=0.75pt]    {$a_{4}$};
\draw (267,214.4) node [anchor=north west][inner sep=0.75pt]    {$a_{8}$};
\draw (132,121.4) node [anchor=north west][inner sep=0.75pt]    {$b_{1}$};
\draw (192,121.4) node [anchor=north west][inner sep=0.75pt]    {$b_{2}$};

\end{tikzpicture}
    }
    \caption{Example of a merge. Anchors are in black, and other vertices in white.
    Anchors $a_4$ in $P_u$ and $a_6$ in pattern $P_v$ are identified as vertex $b_1$ in pattern $P$.
    Anchors $a_5$ in $P_u$ and $a_7$ in pattern $P_v$ are identified as vertex $b_2$ in pattern $P$.
    These four anchors are active anchors, as they are needed for the merge. Other anchors are passive, meaning they are not needed for this merge, but may be needed for a future merge not depicted in this figure.}
    \label{fig:merge-exemple-big}
\end{figure}
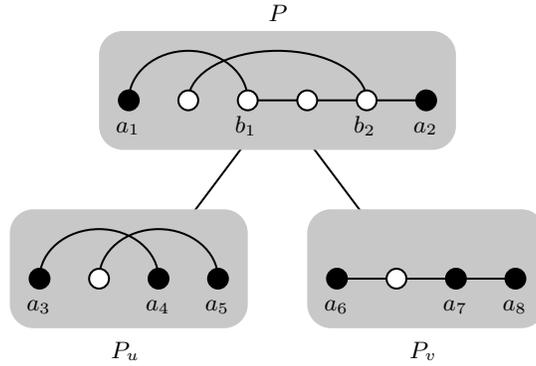

It is not related to a standard a square matrix multiplication anymore, but we can transfer back to this setting by merging some dimensions together.
Let $M_u$ be a $n$ by $n^2$ matrix. Rows are indexed by vertices of $G$ and columns by pair of vertices of $G$. For every vertices $v_1, v_2, v_3 \in V(G)$, we set
$M_u[v_1, (v_2, v_3)] = T_u[v_1, v_2, v_3]$.
Similarly, we define $M_v$ as a $n^2$ by $n$ matrix such that
$M_v[(v_1, v_2), v_3] = T_v[v_1, v_2, v_3]$.
Now, the formula becomes similar to a matrix multiplication:
\begin{align*}
    T[v_1, v_2] = \max_{(v_3, v_4) \in V(G) \times V(G)} M_u[v_1, (v_3, v_4)] \cdot M_v[(v_3, v_4), v_2].
\end{align*}
This means that we can compute $T$ by multiplying $M_u$ by $M_v$.

\paragraph{Merge nodes and matrix multiplication.}


We now show formally how we handle the merge.
At the end, we give the complexity of that computation, parameterized by the maximum number of anchors in any of the three patterns involved.

Let $A_u$ be the set of anchors of $P_u$, $A_v$ the set of anchors of $P_v$, and consider the merge $P'$ of $P_u$ and $P_v$.
Let $A'$ be the set of anchors of the merged pattern $P'$.
We know by definition that $A' \subseteq A_u \cup A_v$.
Let $a = |A_u \cup A_v|$.


To use matrix multiplication, we first need to convert each multi-dimensional matrix into a 2-dimensional matrix.
Let's split the different anchors into several sets:
\begin{itemize}
    \item $U = A_u \setminus A_v$
    \item $V = A_v \setminus A_u$
    \item $X =\overline {A'} \cap A_u \cap A_v$
    \item $Y = A' \cap A_u \cap A_v$
\end{itemize}
This gives a partition of anchors of $A_u \cup A_v$, because  $A' \subseteq A_u \cup A_v$, and every anchor in $U \cup V$ is in $A'$ (this translates the fact that there is no operation to delete passive anchors).
Let fix all anchors in $Y$, and compute a matrix summarizing for all anchors in $U \cup V$ whether there is a realization of $P'$ with anchors in $Y \cup U \cup V = A'$.
The two 2-dimensional matrices are:

\begin{itemize}
    \item $M_u$, where the row are indexes with elements of $U$ and columns with elements of $X$, such that
    \begin{align*}
        M_u[(u_1, \ldots, u_\lambda), (v_1, \ldots,  v_\mu)] = T_u[u_1, \ldots, u_\lambda, v_1, \ldots,  v_\mu].
    \end{align*}
    \item $M_v$, where the row are indexes with elements of $X$ and columns with elements of $V$, such that
    \begin{align*}
        M_v[(v_1, \ldots, v_\mu), (w_1, \ldots,  w_\rho)] = T_v[v_1, \ldots, v_\mu, w_1, \ldots,  w_\rho].
    \end{align*}
\end{itemize}

Doing the multiplication $M_u$ by $M_v$ gives a matrix $M'$ with rows indexed by $U$, columns by $V$, such that $M'[(u_1, \ldots, u_\lambda), (w_1, \ldots,  w_\rho)]$ is positive if and only if there is a realization of $P'$ with anchors $u_1, \ldots, u_\lambda, w_1, \ldots,  w_\rho$.
By doing the multiplication for all realizations of all anchors in $Y$, we can fill $T'$.

This has a complexity of $O(n^x)$ for some $x$. Let's compute $x$. 
Recall that the cost of multiplying two $n$ by $n$ matrices is $O(n^\omega)$, and the cost of multiplying a $n^a$ by $n^b$ matrix by a $n^b$ by $n^c$ matrix is  $O(n^{\omega(a,b,c)})$.
Using clever block multiplications, we have\cite{huang1998fast}, $\omega(a,b,c) \leq \omega \times \max(a+b,b+c,a+c)/2$.
In our case, we do a matrix multiplication between a $|U|$ by $|X|$ matrix and a $|X|$ by $|V|$ matrix for each realization of anchors in $|Y|$, i.e. $O(n^{|Y|})$ times. This gives $x = |Y| + \omega \times \max(|U| + |X|, |X| + |V|, |U| + |V|)/2 \leq \omega \times \max(|Y| + |U| + |X|, |Y| + |X| + |V|, |Y| + |U| + |V|)/2$. By setting $p$ to be the maximum number of anchors (i.e. $p = \max(|A_u|, |A_v|, |A'|)$), we obtain $x \leq \omega p/2$.

Hence, the following theorem:

\begin{theorem}\label{thm:certify-merge-width}
    The complexity to detect a pattern of merge-with $p$ is $O(n^{\omega p/2})$, where the complexity to multiply two $n$ by $n$ matrices is $O(n^\omega)$.
\end{theorem}

\subsection{Patterns with bounded merge-width}

In this section, we show that there is a natural family of patterns (that we call \emph{outerplanar patterns}) that do have bounded merge-width, which justify the study of this parameter. 
We also show that a generalization of this result. 
We believe that several other classes of patterns have bounded merge-width, but we leave this for further work.

\begin{definition}
An outerplanar pattern is a pattern such that when it is drawn with all vertices placed on a line in increasing order, with all edges being semi-circle on top of the line, edges do not intersect each others.
\end{definition}

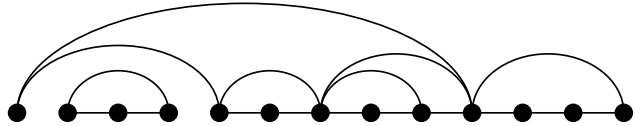
\begin{figure}[!h]
    \centering
    \scalebox{0.85}{
    \tikzset{every picture/.style={line width=0.75pt}} 

\begin{tikzpicture}[x=0.75pt,y=0.75pt,yscale=-1,xscale=1]

\draw  [fill={rgb, 255:red, 0; green, 0; blue, 0 }  ,fill opacity=1 ] (105,135) .. controls (105,132.24) and (107.24,130) .. (110,130) .. controls (112.76,130) and (115,132.24) .. (115,135) .. controls (115,137.76) and (112.76,140) .. (110,140) .. controls (107.24,140) and (105,137.76) .. (105,135) -- cycle ;
\draw  [fill={rgb, 255:red, 0; green, 0; blue, 0 }  ,fill opacity=1 ] (165,135) .. controls (165,132.24) and (167.24,130) .. (170,130) .. controls (172.76,130) and (175,132.24) .. (175,135) .. controls (175,137.76) and (172.76,140) .. (170,140) .. controls (167.24,140) and (165,137.76) .. (165,135) -- cycle ;
\draw  [draw opacity=0] (140,135) .. controls (140,135) and (140,135) .. (140,135) .. controls (140,121.19) and (153.43,110) .. (170,110) .. controls (186.57,110) and (200,121.19) .. (200,135) -- (170,135) -- cycle ; \draw   (140,135) .. controls (140,135) and (140,135) .. (140,135) .. controls (140,121.19) and (153.43,110) .. (170,110) .. controls (186.57,110) and (200,121.19) .. (200,135) ;  
\draw  [fill={rgb, 255:red, 0; green, 0; blue, 0 }  ,fill opacity=1 ] (225,135) .. controls (225,132.24) and (227.24,130) .. (230,130) .. controls (232.76,130) and (235,132.24) .. (235,135) .. controls (235,137.76) and (232.76,140) .. (230,140) .. controls (227.24,140) and (225,137.76) .. (225,135) -- cycle ;
\draw  [fill={rgb, 255:red, 0; green, 0; blue, 0 }  ,fill opacity=1 ] (285,135) .. controls (285,132.24) and (287.24,130) .. (290,130) .. controls (292.76,130) and (295,132.24) .. (295,135) .. controls (295,137.76) and (292.76,140) .. (290,140) .. controls (287.24,140) and (285,137.76) .. (285,135) -- cycle ;
\draw  [fill={rgb, 255:red, 0; green, 0; blue, 0 }  ,fill opacity=1 ] (315,135) .. controls (315,132.24) and (317.24,130) .. (320,130) .. controls (322.76,130) and (325,132.24) .. (325,135) .. controls (325,137.76) and (322.76,140) .. (320,140) .. controls (317.24,140) and (315,137.76) .. (315,135) -- cycle ;
\draw  [fill={rgb, 255:red, 0; green, 0; blue, 0 }  ,fill opacity=1 ] (375,135) .. controls (375,132.24) and (377.24,130) .. (380,130) .. controls (382.76,130) and (385,132.24) .. (385,135) .. controls (385,137.76) and (382.76,140) .. (380,140) .. controls (377.24,140) and (375,137.76) .. (375,135) -- cycle ;
\draw  [fill={rgb, 255:red, 0; green, 0; blue, 0 }  ,fill opacity=1 ] (405,135) .. controls (405,132.24) and (407.24,130) .. (410,130) .. controls (412.76,130) and (415,132.24) .. (415,135) .. controls (415,137.76) and (412.76,140) .. (410,140) .. controls (407.24,140) and (405,137.76) .. (405,135) -- cycle ;
\draw  [fill={rgb, 255:red, 0; green, 0; blue, 0 }  ,fill opacity=1 ] (465,135) .. controls (465,132.24) and (467.24,130) .. (470,130) .. controls (472.76,130) and (475,132.24) .. (475,135) .. controls (475,137.76) and (472.76,140) .. (470,140) .. controls (467.24,140) and (465,137.76) .. (465,135) -- cycle ;
\draw  [draw opacity=0][line width=0.75]  (110,135) .. controls (110,99.1) and (170.44,70) .. (245,70) .. controls (319.56,70) and (380,99.1) .. (380,135) -- (245,135) -- cycle ; \draw  [line width=0.75]  (110,135) .. controls (110,99.1) and (170.44,70) .. (245,70) .. controls (319.56,70) and (380,99.1) .. (380,135) ;  
\draw  [fill={rgb, 255:red, 0; green, 0; blue, 0 }  ,fill opacity=1 ] (135,135) .. controls (135,132.24) and (137.24,130) .. (140,130) .. controls (142.76,130) and (145,132.24) .. (145,135) .. controls (145,137.76) and (142.76,140) .. (140,140) .. controls (137.24,140) and (135,137.76) .. (135,135) -- cycle ;
\draw  [fill={rgb, 255:red, 0; green, 0; blue, 0 }  ,fill opacity=1 ] (195,135) .. controls (195,132.24) and (197.24,130) .. (200,130) .. controls (202.76,130) and (205,132.24) .. (205,135) .. controls (205,137.76) and (202.76,140) .. (200,140) .. controls (197.24,140) and (195,137.76) .. (195,135) -- cycle ;
\draw  [fill={rgb, 255:red, 0; green, 0; blue, 0 }  ,fill opacity=1 ] (255,135) .. controls (255,132.24) and (257.24,130) .. (260,130) .. controls (262.76,130) and (265,132.24) .. (265,135) .. controls (265,137.76) and (262.76,140) .. (260,140) .. controls (257.24,140) and (255,137.76) .. (255,135) -- cycle ;
\draw  [fill={rgb, 255:red, 0; green, 0; blue, 0 }  ,fill opacity=1 ] (345,135) .. controls (345,132.24) and (347.24,130) .. (350,130) .. controls (352.76,130) and (355,132.24) .. (355,135) .. controls (355,137.76) and (352.76,140) .. (350,140) .. controls (347.24,140) and (345,137.76) .. (345,135) -- cycle ;
\draw  [fill={rgb, 255:red, 0; green, 0; blue, 0 }  ,fill opacity=1 ] (435,135) .. controls (435,132.24) and (437.24,130) .. (440,130) .. controls (442.76,130) and (445,132.24) .. (445,135) .. controls (445,137.76) and (442.76,140) .. (440,140) .. controls (437.24,140) and (435,137.76) .. (435,135) -- cycle ;
\draw    (140,135) -- (200,135) ;
\draw  [draw opacity=0] (230,135) .. controls (230,135) and (230,135) .. (230,135) .. controls (230,121.19) and (243.43,110) .. (260,110) .. controls (276.57,110) and (290,121.19) .. (290,135) -- (260,135) -- cycle ; \draw   (230,135) .. controls (230,135) and (230,135) .. (230,135) .. controls (230,121.19) and (243.43,110) .. (260,110) .. controls (276.57,110) and (290,121.19) .. (290,135) ;  
\draw  [draw opacity=0] (290,135) .. controls (290,135) and (290,135) .. (290,135) .. controls (290,121.19) and (303.43,110) .. (320,110) .. controls (336.57,110) and (350,121.19) .. (350,135) -- (320,135) -- cycle ; \draw   (290,135) .. controls (290,135) and (290,135) .. (290,135) .. controls (290,121.19) and (303.43,110) .. (320,110) .. controls (336.57,110) and (350,121.19) .. (350,135) ;  
\draw  [draw opacity=0] (290,135) .. controls (290,135) and (290,135) .. (290,135) .. controls (290,115.67) and (310.15,100) .. (335,100) .. controls (359.85,100) and (380,115.67) .. (380,135) -- (335,135) -- cycle ; \draw   (290,135) .. controls (290,135) and (290,135) .. (290,135) .. controls (290,115.67) and (310.15,100) .. (335,100) .. controls (359.85,100) and (380,115.67) .. (380,135) ;  
\draw    (230,135) -- (380,135) ;
\draw  [draw opacity=0] (110,135) .. controls (110,112.91) and (136.86,95) .. (170,95) .. controls (203.14,95) and (230,112.91) .. (230,135) -- (170,135) -- cycle ; \draw   (110,135) .. controls (110,112.91) and (136.86,95) .. (170,95) .. controls (203.14,95) and (230,112.91) .. (230,135) ;  
\draw  [draw opacity=0] (380,135) .. controls (380,135) and (380,135) .. (380,135) .. controls (380,115.67) and (400.15,100) .. (425,100) .. controls (449.85,100) and (470,115.67) .. (470,135) -- (425,135) -- cycle ; \draw   (380,135) .. controls (380,135) and (380,135) .. (380,135) .. controls (380,115.67) and (400.15,100) .. (425,100) .. controls (449.85,100) and (470,115.67) .. (470,135) ;  
\draw    (380,135) -- (470,135) ;

\end{tikzpicture}
    }
    \caption{Example of an outerplanar pattern with 14 vertices and 17 edges. For readability, edges may not be represented by semi-circles, but rather by semi-ellipses or even straight lines.}
    \label{fig:planar-example}
\end{figure}

The first part of this section (Lemmas \ref{lem:pattern-decomposition} and \ref{lem:merge-decomposition}) is dedicated to prove the following theorem:

\begin{theorem}\label{thm:outerplanar-merge-width}
Outerplanar patterns have merge-width at most 2. 
\end{theorem}

Combined with Theorem \ref{thm:certify-merge-width}, we get the following.

\begin{corollary}
There exists an algorithm to detect an outerplanar pattern in time $O(n^\omega)$, i.e. as efficient as a matrix multiplication algorithm.
\end{corollary}

The proof of Theorem~\ref{thm:outerplanar-merge-width} consists of two lemmas: Lemma~\ref{lem:pattern-decomposition} that shows how to decompose an outerplanar pattern, and Lemma~\ref{lem:merge-decomposition} that shows how to transfer this to a correct merge tree. 

\begin{lemma}\label{lem:pattern-decomposition}
Any outerplanar pattern $P$ on a vertex set $p_1,...,p_k$ with $k>1$ can be decomposed in one of the following ways:
\begin{enumerate}
    \item (Covering edge) It is has an edge $(p_1,p_k)$, and we denote by $P\setminus (p_1,p_k)$ the same pattern without this edge. See Figure~\ref{fig:planar-covering-edge}.
    \item (Cut vertex) It has a vertex $p_t$ with $1<t<k$, and  an edge $(p_1,p_t)$ such that there is no edge of the form $(p_i,p_j)$ with $i<t<j$.
    We denote by $P[1,t]$ and $P[t,k]$ the two patterns using the vertex sets $(p_1,...,p_t)$ and $(p_t,...,p_k)$ respectively. See Figure~\ref{fig:planar-cut-vertex}.
    \item (Isolated first) The first vertex $p_1$ is isolated.
    We denote by $P[2,k]$ the rest of the pattern.
    See Figure~\ref{fig:planar-isolated-first}.
\end{enumerate}
\end{lemma}

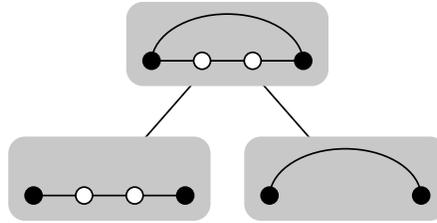
\begin{figure}[!h]
    \centering
    \scalebox{0.85}{
    \tikzset{every picture/.style={line width=0.75pt}} 

\begin{tikzpicture}[x=0.75pt,y=0.75pt,yscale=-1,xscale=1]

\draw    (230,155) -- (160,75) ;
\draw    (90,155) -- (160,75) ;
\draw  [draw opacity=0][fill={rgb, 255:red, 200; green, 200; blue, 200 }  ,fill opacity=1 ] (30,140) .. controls (30,134.48) and (34.48,130) .. (40,130) -- (140,130) .. controls (145.52,130) and (150,134.48) .. (150,140) -- (150,170) .. controls (150,175.52) and (145.52,180) .. (140,180) -- (40,180) .. controls (34.48,180) and (30,175.52) .. (30,170) -- cycle ;
\draw  [fill={rgb, 255:red, 0; green, 0; blue, 0 }  ,fill opacity=1 ] (40,165) .. controls (40,162.24) and (42.24,160) .. (45,160) .. controls (47.76,160) and (50,162.24) .. (50,165) .. controls (50,167.76) and (47.76,170) .. (45,170) .. controls (42.24,170) and (40,167.76) .. (40,165) -- cycle ;
\draw    (45,165) -- (135,165) ;
\draw  [fill={rgb, 255:red, 255; green, 255; blue, 255 }  ,fill opacity=1 ] (70,165) .. controls (70,162.24) and (72.24,160) .. (75,160) .. controls (77.76,160) and (80,162.24) .. (80,165) .. controls (80,167.76) and (77.76,170) .. (75,170) .. controls (72.24,170) and (70,167.76) .. (70,165) -- cycle ;
\draw  [fill={rgb, 255:red, 255; green, 255; blue, 255 }  ,fill opacity=1 ] (100,165) .. controls (100,162.24) and (102.24,160) .. (105,160) .. controls (107.76,160) and (110,162.24) .. (110,165) .. controls (110,167.76) and (107.76,170) .. (105,170) .. controls (102.24,170) and (100,167.76) .. (100,165) -- cycle ;
\draw  [fill={rgb, 255:red, 0; green, 0; blue, 0 }  ,fill opacity=1 ] (130,165) .. controls (130,162.24) and (132.24,160) .. (135,160) .. controls (137.76,160) and (140,162.24) .. (140,165) .. controls (140,167.76) and (137.76,170) .. (135,170) .. controls (132.24,170) and (130,167.76) .. (130,165) -- cycle ;
\draw  [draw opacity=0][fill={rgb, 255:red, 200; green, 200; blue, 200 }  ,fill opacity=1 ] (170,140) .. controls (170,134.48) and (174.48,130) .. (180,130) -- (280,130) .. controls (285.52,130) and (290,134.48) .. (290,140) -- (290,170) .. controls (290,175.52) and (285.52,180) .. (280,180) -- (180,180) .. controls (174.48,180) and (170,175.52) .. (170,170) -- cycle ;
\draw  [fill={rgb, 255:red, 0; green, 0; blue, 0 }  ,fill opacity=1 ] (180,165) .. controls (180,162.24) and (182.24,160) .. (185,160) .. controls (187.76,160) and (190,162.24) .. (190,165) .. controls (190,167.76) and (187.76,170) .. (185,170) .. controls (182.24,170) and (180,167.76) .. (180,165) -- cycle ;
\draw  [fill={rgb, 255:red, 0; green, 0; blue, 0 }  ,fill opacity=1 ] (270,165) .. controls (270,162.24) and (272.24,160) .. (275,160) .. controls (277.76,160) and (280,162.24) .. (280,165) .. controls (280,167.76) and (277.76,170) .. (275,170) .. controls (272.24,170) and (270,167.76) .. (270,165) -- cycle ;
\draw  [draw opacity=0] (185,165) .. controls (185.93,149.58) and (205.71,137.25) .. (229.97,137.25) .. controls (252.69,137.25) and (271.49,148.07) .. (274.53,162.12) -- (229.97,166.14) -- cycle ; \draw   (185,165) .. controls (185.93,149.58) and (205.71,137.25) .. (229.97,137.25) .. controls (252.69,137.25) and (271.49,148.07) .. (274.53,162.12) ;  
\draw  [draw opacity=0][fill={rgb, 255:red, 200; green, 200; blue, 200 }  ,fill opacity=1 ] (100,60) .. controls (100,54.48) and (104.48,50) .. (110,50) -- (210,50) .. controls (215.52,50) and (220,54.48) .. (220,60) -- (220,90) .. controls (220,95.52) and (215.52,100) .. (210,100) -- (110,100) .. controls (104.48,100) and (100,95.52) .. (100,90) -- cycle ;
\draw  [fill={rgb, 255:red, 0; green, 0; blue, 0 }  ,fill opacity=1 ] (110,85) .. controls (110,82.24) and (112.24,80) .. (115,80) .. controls (117.76,80) and (120,82.24) .. (120,85) .. controls (120,87.76) and (117.76,90) .. (115,90) .. controls (112.24,90) and (110,87.76) .. (110,85) -- cycle ;
\draw    (115,85) -- (205,85) ;
\draw  [fill={rgb, 255:red, 255; green, 255; blue, 255 }  ,fill opacity=1 ] (140,85) .. controls (140,82.24) and (142.24,80) .. (145,80) .. controls (147.76,80) and (150,82.24) .. (150,85) .. controls (150,87.76) and (147.76,90) .. (145,90) .. controls (142.24,90) and (140,87.76) .. (140,85) -- cycle ;
\draw  [fill={rgb, 255:red, 255; green, 255; blue, 255 }  ,fill opacity=1 ] (170,85) .. controls (170,82.24) and (172.24,80) .. (175,80) .. controls (177.76,80) and (180,82.24) .. (180,85) .. controls (180,87.76) and (177.76,90) .. (175,90) .. controls (172.24,90) and (170,87.76) .. (170,85) -- cycle ;
\draw  [fill={rgb, 255:red, 0; green, 0; blue, 0 }  ,fill opacity=1 ] (200,85) .. controls (200,82.24) and (202.24,80) .. (205,80) .. controls (207.76,80) and (210,82.24) .. (210,85) .. controls (210,87.76) and (207.76,90) .. (205,90) .. controls (202.24,90) and (200,87.76) .. (200,85) -- cycle ;
\draw  [draw opacity=0] (115,85) .. controls (115.93,69.58) and (135.71,57.25) .. (159.97,57.25) .. controls (182.69,57.25) and (201.49,68.07) .. (204.53,82.12) -- (159.97,86.14) -- cycle ; \draw   (115,85) .. controls (115.93,69.58) and (135.71,57.25) .. (159.97,57.25) .. controls (182.69,57.25) and (201.49,68.07) .. (204.53,82.12) ;

\end{tikzpicture}
    }
    \caption{Covering edge}
    \label{fig:planar-covering-edge}
\end{figure}

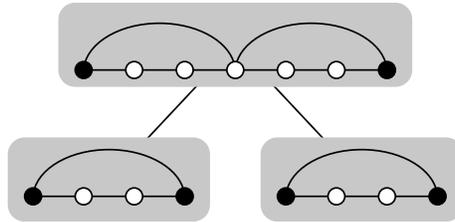
\begin{figure}[!h]
    \centering
    \scalebox{0.85}{
    \tikzset{every picture/.style={line width=0.75pt}} 

\begin{tikzpicture}[x=0.75pt,y=0.75pt,yscale=-1,xscale=1]

\draw    (290,175) -- (215,95) ;
\draw    (140,175) -- (215,95) ;
\draw  [draw opacity=0][fill={rgb, 255:red, 200; green, 200; blue, 200 }  ,fill opacity=1 ] (80,160) .. controls (80,154.48) and (84.48,150) .. (90,150) -- (190,150) .. controls (195.52,150) and (200,154.48) .. (200,160) -- (200,190) .. controls (200,195.52) and (195.52,200) .. (190,200) -- (90,200) .. controls (84.48,200) and (80,195.52) .. (80,190) -- cycle ;
\draw  [fill={rgb, 255:red, 0; green, 0; blue, 0 }  ,fill opacity=1 ] (90,185) .. controls (90,182.24) and (92.24,180) .. (95,180) .. controls (97.76,180) and (100,182.24) .. (100,185) .. controls (100,187.76) and (97.76,190) .. (95,190) .. controls (92.24,190) and (90,187.76) .. (90,185) -- cycle ;
\draw    (95,185) -- (185,185) ;
\draw  [fill={rgb, 255:red, 255; green, 255; blue, 255 }  ,fill opacity=1 ] (120,185) .. controls (120,182.24) and (122.24,180) .. (125,180) .. controls (127.76,180) and (130,182.24) .. (130,185) .. controls (130,187.76) and (127.76,190) .. (125,190) .. controls (122.24,190) and (120,187.76) .. (120,185) -- cycle ;
\draw  [fill={rgb, 255:red, 255; green, 255; blue, 255 }  ,fill opacity=1 ] (150,185) .. controls (150,182.24) and (152.24,180) .. (155,180) .. controls (157.76,180) and (160,182.24) .. (160,185) .. controls (160,187.76) and (157.76,190) .. (155,190) .. controls (152.24,190) and (150,187.76) .. (150,185) -- cycle ;
\draw  [fill={rgb, 255:red, 0; green, 0; blue, 0 }  ,fill opacity=1 ] (180,185) .. controls (180,182.24) and (182.24,180) .. (185,180) .. controls (187.76,180) and (190,182.24) .. (190,185) .. controls (190,187.76) and (187.76,190) .. (185,190) .. controls (182.24,190) and (180,187.76) .. (180,185) -- cycle ;
\draw  [draw opacity=0] (95,185) .. controls (95.93,169.58) and (115.71,157.25) .. (139.97,157.25) .. controls (162.69,157.25) and (181.49,168.07) .. (184.53,182.12) -- (139.97,186.14) -- cycle ; \draw   (95,185) .. controls (95.93,169.58) and (115.71,157.25) .. (139.97,157.25) .. controls (162.69,157.25) and (181.49,168.07) .. (184.53,182.12) ;  
\draw  [draw opacity=0][fill={rgb, 255:red, 200; green, 200; blue, 200 }  ,fill opacity=1 ] (230,160) .. controls (230,154.48) and (234.48,150) .. (240,150) -- (340,150) .. controls (345.52,150) and (350,154.48) .. (350,160) -- (350,190) .. controls (350,195.52) and (345.52,200) .. (340,200) -- (240,200) .. controls (234.48,200) and (230,195.52) .. (230,190) -- cycle ;
\draw  [fill={rgb, 255:red, 0; green, 0; blue, 0 }  ,fill opacity=1 ] (240,185) .. controls (240,182.24) and (242.24,180) .. (245,180) .. controls (247.76,180) and (250,182.24) .. (250,185) .. controls (250,187.76) and (247.76,190) .. (245,190) .. controls (242.24,190) and (240,187.76) .. (240,185) -- cycle ;
\draw    (245,185) -- (335,185) ;
\draw  [fill={rgb, 255:red, 255; green, 255; blue, 255 }  ,fill opacity=1 ] (270,185) .. controls (270,182.24) and (272.24,180) .. (275,180) .. controls (277.76,180) and (280,182.24) .. (280,185) .. controls (280,187.76) and (277.76,190) .. (275,190) .. controls (272.24,190) and (270,187.76) .. (270,185) -- cycle ;
\draw  [fill={rgb, 255:red, 255; green, 255; blue, 255 }  ,fill opacity=1 ] (300,185) .. controls (300,182.24) and (302.24,180) .. (305,180) .. controls (307.76,180) and (310,182.24) .. (310,185) .. controls (310,187.76) and (307.76,190) .. (305,190) .. controls (302.24,190) and (300,187.76) .. (300,185) -- cycle ;
\draw  [fill={rgb, 255:red, 0; green, 0; blue, 0 }  ,fill opacity=1 ] (330,185) .. controls (330,182.24) and (332.24,180) .. (335,180) .. controls (337.76,180) and (340,182.24) .. (340,185) .. controls (340,187.76) and (337.76,190) .. (335,190) .. controls (332.24,190) and (330,187.76) .. (330,185) -- cycle ;
\draw  [draw opacity=0] (245,185) .. controls (245.93,169.58) and (265.71,157.25) .. (289.97,157.25) .. controls (312.69,157.25) and (331.49,168.07) .. (334.53,182.12) -- (289.97,186.14) -- cycle ; \draw   (245,185) .. controls (245.93,169.58) and (265.71,157.25) .. (289.97,157.25) .. controls (312.69,157.25) and (331.49,168.07) .. (334.53,182.12) ;  
\draw  [draw opacity=0][fill={rgb, 255:red, 200; green, 200; blue, 200 }  ,fill opacity=1 ] (110,80) .. controls (110,74.48) and (114.48,70) .. (120,70) -- (310,70) .. controls (315.52,70) and (320,74.48) .. (320,80) -- (320,110) .. controls (320,115.52) and (315.52,120) .. (310,120) -- (120,120) .. controls (114.48,120) and (110,115.52) .. (110,110) -- cycle ;
\draw  [fill={rgb, 255:red, 0; green, 0; blue, 0 }  ,fill opacity=1 ] (120,109.98) .. controls (120,107.22) and (122.24,104.98) .. (125,104.98) .. controls (127.76,104.98) and (130,107.22) .. (130,109.98) .. controls (130,112.74) and (127.76,114.98) .. (125,114.98) .. controls (122.24,114.98) and (120,112.74) .. (120,109.98) -- cycle ;
\draw    (125,109.98) -- (215,109.98) ;
\draw  [fill={rgb, 255:red, 255; green, 255; blue, 255 }  ,fill opacity=1 ] (150,109.98) .. controls (150,107.22) and (152.24,104.98) .. (155,104.98) .. controls (157.76,104.98) and (160,107.22) .. (160,109.98) .. controls (160,112.74) and (157.76,114.98) .. (155,114.98) .. controls (152.24,114.98) and (150,112.74) .. (150,109.98) -- cycle ;
\draw  [fill={rgb, 255:red, 255; green, 255; blue, 255 }  ,fill opacity=1 ] (180,109.98) .. controls (180,107.22) and (182.24,104.98) .. (185,104.98) .. controls (187.76,104.98) and (190,107.22) .. (190,109.98) .. controls (190,112.74) and (187.76,114.98) .. (185,114.98) .. controls (182.24,114.98) and (180,112.74) .. (180,109.98) -- cycle ;
\draw  [draw opacity=0] (125,109.98) .. controls (125.93,94.55) and (145.71,82.23) .. (169.97,82.23) .. controls (192.69,82.23) and (211.49,93.05) .. (214.53,107.1) -- (169.97,111.12) -- cycle ; \draw   (125,109.98) .. controls (125.93,94.55) and (145.71,82.23) .. (169.97,82.23) .. controls (192.69,82.23) and (211.49,93.05) .. (214.53,107.1) ;  
\draw    (215,109.98) -- (305,109.98) ;
\draw  [fill={rgb, 255:red, 255; green, 255; blue, 255 }  ,fill opacity=1 ] (240,109.98) .. controls (240,107.22) and (242.24,104.98) .. (245,104.98) .. controls (247.76,104.98) and (250,107.22) .. (250,109.98) .. controls (250,112.74) and (247.76,114.98) .. (245,114.98) .. controls (242.24,114.98) and (240,112.74) .. (240,109.98) -- cycle ;
\draw  [fill={rgb, 255:red, 255; green, 255; blue, 255 }  ,fill opacity=1 ] (270,109.98) .. controls (270,107.22) and (272.24,104.98) .. (275,104.98) .. controls (277.76,104.98) and (280,107.22) .. (280,109.98) .. controls (280,112.74) and (277.76,114.98) .. (275,114.98) .. controls (272.24,114.98) and (270,112.74) .. (270,109.98) -- cycle ;
\draw  [fill={rgb, 255:red, 0; green, 0; blue, 0 }  ,fill opacity=1 ] (300,109.98) .. controls (300,107.22) and (302.24,104.98) .. (305,104.98) .. controls (307.76,104.98) and (310,107.22) .. (310,109.98) .. controls (310,112.74) and (307.76,114.98) .. (305,114.98) .. controls (302.24,114.98) and (300,112.74) .. (300,109.98) -- cycle ;
\draw  [draw opacity=0] (215,109.98) .. controls (215.93,94.55) and (235.71,82.23) .. (259.97,82.23) .. controls (282.69,82.23) and (301.49,93.05) .. (304.53,107.1) -- (259.97,111.12) -- cycle ; \draw   (215,109.98) .. controls (215.93,94.55) and (235.71,82.23) .. (259.97,82.23) .. controls (282.69,82.23) and (301.49,93.05) .. (304.53,107.1) ;  
\draw  [fill={rgb, 255:red, 255; green, 255; blue, 255 }  ,fill opacity=1 ] (210,109.98) .. controls (210,107.22) and (212.24,104.98) .. (215,104.98) .. controls (217.76,104.98) and (220,107.22) .. (220,109.98) .. controls (220,112.74) and (217.76,114.98) .. (215,114.98) .. controls (212.24,114.98) and (210,112.74) .. (210,109.98) -- cycle ;

\end{tikzpicture}
    }
    \caption{Cut vertex}
    \label{fig:planar-cut-vertex}
\end{figure}

\begin{figure}[!h]
    \centering
    \scalebox{0.85}{
    \tikzset{every picture/.style={line width=0.75pt}} 

\begin{tikzpicture}[x=0.75pt,y=0.75pt,yscale=-1,xscale=1]

\draw    (175,55) -- (175,135) ;
\draw  [draw opacity=0][fill={rgb, 255:red, 200; green, 200; blue, 200 }  ,fill opacity=1 ] (100,40) .. controls (100,34.48) and (104.48,30) .. (110,30) -- (240,30) .. controls (245.52,30) and (250,34.48) .. (250,40) -- (250,70) .. controls (250,75.52) and (245.52,80) .. (240,80) -- (110,80) .. controls (104.48,80) and (100,75.52) .. (100,70) -- cycle ;
\draw    (145,65) -- (235,65) ;
\draw  [fill={rgb, 255:red, 255; green, 255; blue, 255 }  ,fill opacity=1 ] (170,65) .. controls (170,62.24) and (172.24,60) .. (175,60) .. controls (177.76,60) and (180,62.24) .. (180,65) .. controls (180,67.76) and (177.76,70) .. (175,70) .. controls (172.24,70) and (170,67.76) .. (170,65) -- cycle ;
\draw  [fill={rgb, 255:red, 255; green, 255; blue, 255 }  ,fill opacity=1 ] (200,65) .. controls (200,62.24) and (202.24,60) .. (205,60) .. controls (207.76,60) and (210,62.24) .. (210,65) .. controls (210,67.76) and (207.76,70) .. (205,70) .. controls (202.24,70) and (200,67.76) .. (200,65) -- cycle ;
\draw  [fill={rgb, 255:red, 0; green, 0; blue, 0 }  ,fill opacity=1 ] (230,65) .. controls (230,62.24) and (232.24,60) .. (235,60) .. controls (237.76,60) and (240,62.24) .. (240,65) .. controls (240,67.76) and (237.76,70) .. (235,70) .. controls (232.24,70) and (230,67.76) .. (230,65) -- cycle ;
\draw  [draw opacity=0] (145,65) .. controls (145.93,49.58) and (165.71,37.25) .. (189.97,37.25) .. controls (212.69,37.25) and (231.49,48.07) .. (234.53,62.12) -- (189.97,66.14) -- cycle ; \draw   (145,65) .. controls (145.93,49.58) and (165.71,37.25) .. (189.97,37.25) .. controls (212.69,37.25) and (231.49,48.07) .. (234.53,62.12) ;  
\draw  [fill={rgb, 255:red, 0; green, 0; blue, 0 }  ,fill opacity=1 ] (110,65) .. controls (110,62.24) and (112.24,60) .. (115,60) .. controls (117.76,60) and (120,62.24) .. (120,65) .. controls (120,67.76) and (117.76,70) .. (115,70) .. controls (112.24,70) and (110,67.76) .. (110,65) -- cycle ;
\draw  [fill={rgb, 255:red, 255; green, 255; blue, 255 }  ,fill opacity=1 ] (140,65) .. controls (140,62.24) and (142.24,60) .. (145,60) .. controls (147.76,60) and (150,62.24) .. (150,65) .. controls (150,67.76) and (147.76,70) .. (145,70) .. controls (142.24,70) and (140,67.76) .. (140,65) -- cycle ;
\draw  [draw opacity=0][fill={rgb, 255:red, 200; green, 200; blue, 200 }  ,fill opacity=1 ] (115,120) .. controls (115,114.48) and (119.48,110) .. (125,110) -- (225,110) .. controls (230.52,110) and (235,114.48) .. (235,120) -- (235,150) .. controls (235,155.52) and (230.52,160) .. (225,160) -- (125,160) .. controls (119.48,160) and (115,155.52) .. (115,150) -- cycle ;
\draw    (130,145) -- (220,145) ;
\draw  [fill={rgb, 255:red, 255; green, 255; blue, 255 }  ,fill opacity=1 ] (155,145) .. controls (155,142.24) and (157.24,140) .. (160,140) .. controls (162.76,140) and (165,142.24) .. (165,145) .. controls (165,147.76) and (162.76,150) .. (160,150) .. controls (157.24,150) and (155,147.76) .. (155,145) -- cycle ;
\draw  [fill={rgb, 255:red, 255; green, 255; blue, 255 }  ,fill opacity=1 ] (185,145) .. controls (185,142.24) and (187.24,140) .. (190,140) .. controls (192.76,140) and (195,142.24) .. (195,145) .. controls (195,147.76) and (192.76,150) .. (190,150) .. controls (187.24,150) and (185,147.76) .. (185,145) -- cycle ;
\draw  [fill={rgb, 255:red, 0; green, 0; blue, 0 }  ,fill opacity=1 ] (215,145) .. controls (215,142.24) and (217.24,140) .. (220,140) .. controls (222.76,140) and (225,142.24) .. (225,145) .. controls (225,147.76) and (222.76,150) .. (220,150) .. controls (217.24,150) and (215,147.76) .. (215,145) -- cycle ;
\draw  [draw opacity=0] (130,145) .. controls (130.93,129.58) and (150.71,117.25) .. (174.97,117.25) .. controls (197.69,117.25) and (216.49,128.07) .. (219.53,142.12) -- (174.97,146.14) -- cycle ; \draw   (130,145) .. controls (130.93,129.58) and (150.71,117.25) .. (174.97,117.25) .. controls (197.69,117.25) and (216.49,128.07) .. (219.53,142.12) ;  
\draw  [fill={rgb, 255:red, 0; green, 0; blue, 0 }  ,fill opacity=1 ] (125,145) .. controls (125,142.24) and (127.24,140) .. (130,140) .. controls (132.76,140) and (135,142.24) .. (135,145) .. controls (135,147.76) and (132.76,150) .. (130,150) .. controls (127.24,150) and (125,147.76) .. (125,145) -- cycle ;

\end{tikzpicture}
    }
    \caption{Isolated first}
    \label{fig:planar-isolated-first}
\end{figure}
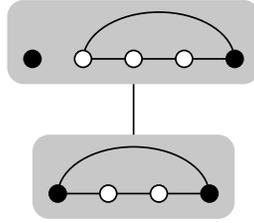

\begin{proof}
Note that we claim that the only pattern that cannot be decomposed is the single vertex pattern. We do a case analysis on the neighborhood of $p_1$, if the pattern is not a single vertex. 
If $p_1$ has degree 0, this is an \emph{Isolated first} case.
If there is an edge between $p_1$ and $p_k$, it is \emph{Covering edge} case.
Otherwise, it means that there is an edge between $p_1$ and some others vertices $p_t$ with $t < k$. Let $t$ be the greatest such a vertex. We check that  there is no edge of the form $(p_i,p_j)$ with $i<t<j$.
For $i=1$, since $t$ is maximal, there is no edge of the form $(p_1,p_j)$ with $t<j$.
For $i>1$, since $P$ is outerplanar, there is no edge that crosses edge $(p_1, p_t)$, meaning that there is no edge of the form $(p_i,p_j)$ with $1<i<t<j$.
Hence, this is \emph{Cut Vertex} case.\qed
\end{proof}

\begin{remark}
The operations of Lemma~\ref{lem:pattern-decomposition} keep the outerplanarity: the pattern(s) obtained by decomposing an outerplanar pattern are outerplanar. 
\end{remark}

\begin{lemma}\label{lem:merge-decomposition}
    Let $P$ be an outerplanar pattern.
    There is a merge tree for pattern~$P$ such that every nodes of the tree is labeled by a pattern with 2 anchors, which are located on the leftmost and the rightmost vertices of the pattern. 
\end{lemma}

In Figure~\ref{fig:planar-covering-edge}, \ref{fig:planar-cut-vertex} and~\ref{fig:planar-isolated-first}, these anchors are in black.

\begin{proof}
    We prove the lemma by induction based on the decomposition of Lemma~\ref{lem:pattern-decomposition}.
    Note that the case where the decomposition does not apply (the single vertex pattern) is well-suited to the merge tree definition, since creating a vertex is an allowed operation in Definition~\ref{def:operations}, and we can put those at the leaves.
    
    Let $P$ be a pattern on a vertex set $p_1,...,p_k$, $k>1$, we first describe the merge tree of $P$ by a case analysis depending on the case of Lemma~\ref{lem:pattern-decomposition}, and they justify its correctness.
    \begin{enumerate}
    \item (Covering edge) There is an edge $(p_1,p_k)$, and let $P' = P\setminus (p_1,p_k)$, i.e. the same pattern without this edge. 
    Let $T'$ be the merge tree corresponding to~$P'$, built by induction. 
    Then the root of $T'$ is labeled with an anchored version of $P'$, with the vertices $p_1$ and $p_k$ being anchored.
    We build the merge tree $T$ corresponding to $P$ as follows:
    \begin{itemize}
        \item The root is labeled with an anchored version of $P$, with the vertices $p_1$ and $p_k$ being anchored.
        \item The left subtree of the root is $T'$.
        \item The right subtree of the root is single node, labeled with edge $(p_1,p_k)$ and with both endpoints anchored.
    \end{itemize}
    \item (Cut vertex) There is a vertex $p_t$ with $1<t<k$, and an edge $(p_1,p_t)$ such that there is no edge of the form $(p_i,p_j)$ with $i<t<j$.
    Let $P_1 = P[1,t]$ and $P_2 = P[t,k]$ be the two patterns using the vertex sets $(p_1,...,p_t)$ and $(p_t,...,p_k)$ respectively.
    Let $T_1$ be the merge tree corresponding to $P_1$, built by induction, and $T_2$ be the merge tree corresponding to $P_2$.
     We build the merge tree $T$ corresponding to $P$ as followed:
    \begin{itemize}
        \item The root is labeled with an anchored version of $P$, with the vertices $p_1$ and $p_k$ being anchored.
        \item The left subtree of the root is tree $T_1$. Note that the root of $T_1$ is labeled with an anchored version of $P_1$, with anchored on vertices $p_1$ and $p_t$.
        \item The right subtree of the root is tree $T_2$. Note that the root of $T_2$ is labeled with an anchored version of $P_2$, with anchored on vertices $p_t$ and $p_k$.
    \end{itemize}
    \item (Isolated first) The first vertex $p_1$ is isolated.
    Let $P' = P[2,k]$ be the rest of the pattern.
    Let $T'$ be the merge tree corresponding to $P'$, built by induction. It means that the root of $T'$ is labeled with an anchored version of $P'$, with the vertices $p_2$ and $p_k$ being anchored.
    We build the merge tree $T$ corresponding to $P$ as followed:
    \begin{itemize}
        \item The root is labeled with an anchored version of $P$, with the vertices $p_1$ and $p_k$ being anchored.
        \item The only child of the root is $T'$.
    \end{itemize}
    \end{enumerate}
    To complete the proof, we only need to check that every merge is correct according to Definition \ref{def:operations}, i.e. it fulfills conditions (a) through (e).
    The cases Covering edge and Isolated vertex are straightforward.
    We focus on the Cut Vertex case. 
    We have to prove that the pattern $P=(V,E)$ with vertices $p_1$ and $p_k$ anchored can be divided into pattern $P_1=(V_1, E_1)$ with vertices $p_1$ and $p_t$ anchored, and pattern $P_2=(V_2, E_2)$ with vertices $p_t$ and $p_k$ anchored:
    \begin{enumerate}[label=(\alph*)]
        \item $(E_1, E_2)$ is a partition of $E$: let $(p_i,p_j)$ be an edge. If $i<j\leq t$, this edges is in $E_1$. If $t\leq i<j$, it is in $E_2$. We know there is no other type of edge, since it is impossible to have $i<t<j$ according to the definition of Cut vertex case.
        We also need to check that $\forall (p_i, p_j) \in E_1, p_i, p_j \in V_1$: since $(p_i, p_j) \in E_1$, we know that $i < j \leq t$, which proves the point as $V_1 = \{p_1, \ldots, p_t\}$.
        The same is true regarding $E_2$ and $V_2$.
        \item $V_1 \cup V_2 = V$, since $V_1 = \{p_1, \ldots, p_t\}$, $V_2 = \{p_t, \ldots, p_k\}$ and $V = \{p_1, \ldots, p_k\}$.
        \item $V_1 \cap V_2 = \{p_t\}$, which is an anchor in both $P_1$ and $P_2$.
        \item Anchors of pattern $P$ are $p_1$ and $p_k$, they are indeed anchors in every pattern they appear in: $p_1$ is an anchor in $P_1$ and $p_k$ is an anchor in $P_2$.
        \item The distribution of vertices across pattern $P_1$ and $P_2$ is as followed: vertices $p_1, \ldots, p_{t-1}$ are only in $P_1$, vertex $p_t$ is in both $P_1$ and $P_2$, and vertices $p_{t+1}, \ldots, p_k$ are only in $P_2$. Hence, there is no consecutive vertices such that one is only in $P_1$ and the other only in $P_2$ thanks to vertex $p_t$ which stands in the middle.\qed
    \end{enumerate}
\end{proof}


\paragraph{Extension.}
We will now prove an extension of this result. To do so, let us define a notion of distance. 

\begin{definition}
The \emph{edit distance to outerplanarity} of a pattern $P$, denoted by $\edit{}(P)$, is the minimal number of edges to remove from $P$ so it becomes an outerplanar pattern.
\end{definition}

\begin{theorem}\label{thm:cross-merge-width}
The merge-width of any pattern $P$ is bounded by $2\edit(P)+2$. 
\end{theorem}

Combined with Theorem \ref{thm:certify-merge-width}, we get that:

\begin{theorem}
Let $P$ be a pattern.
There exists an algorithm to detect $P$ in time $O(n^{\omega(\edit(P)+1)})$, i.e. as efficient as multiplying two $n^{\edit+1}$ by $n^{\edit+1}$ matrices.
\end{theorem}

Before proving Theorem \ref{thm:cross-merge-width}, let us start with the following general lemma, where we ``force'' an anchor. 

\begin{lemma}\label{lem:add-anchor}
    Let $P$ be a pattern of merge-width $w$ and let $p_i$ be a vertices of $P$. There exists a merge tree $T'$ for pattern $P$ such that vertex $p_i$ is anchored in the root of $T'$ and such that $T'$ has width at most $w+1$.
\end{lemma}

\begin{proof}
    Let $T$ be a merge tree of $P$ with width $w$.
    We build $T'$ from $T$, by anchoring every instance of vertex $p_i$ in labels of $T$.
    By definition, the width of~$T'$ is at most $w+1$.
    It remains to check that $T'$ is correct, i.e. that each pattern $P=(V,E)$ that labels a node with two children can be obtained by a correct merge of patterns of its two children $P_1=(V_1,E_1)$ and $P_2=(V_2,E_2)$, according to Definition \ref{def:operations}.
    We have to check that each pattern fulfills conditions (a) through (e): Since $V_1, V_2, E_1$ and $E_2$ are the same in $T$ and $T'$, conditions (a), (b) and (c) remains true in $T'$.
    For condition (d), since we anchor vertex $p_i$ in every pattern of $T'$, it remains true.
    Finally, since the vertex is anchored in every pattern where it appears, in a merge it will be anchored in both $V_1$ and~$V_2$, thus condition~(e) is also satisfied.
    \qed
\end{proof}

\begin{proof}[(Theorem~\ref{thm:cross-merge-width})]
    Let $P=(V,E)$ be a pattern with vertex set $V={p_1, \ldots, p_k}$.
    Let $P_1=(V_1,E_1)$ be an outerplanar pattern obtained by removing $\edit(P)$ edges from $P$. 
    Note that $V_1=V$ since we do not remove any vertex.
    Let $P_2=(V_2,E_2)$ be the pattern induced by the edges removed in $P_1$, meaning that $V_2$ is the set of vertices that are endpoints of an edge of $E_2$. 
    Note that $P_2$ has $\edit(P)$ edges and at most $2\edit(P)$ vertices.
    Since $P_1$ is outerplanar, by Lemma \ref{lem:merge-decomposition}, it has a merge tree of width at most 2, with the label of the root anchored in $p_1$ and $p_k$.
    We use Lemma \ref{lem:add-anchor} to add all vertices from $V_2$ as anchors to the root of this merge tree. This gives a merge tree $T_1$ of width at most $2\edit(P)+2$.
    Let $A = \{p_1, p_k\} \cup V_2$ be the set of anchors of the label of the root of $T_1$. Note that $|A| \leq 2\edit(P)+2$.
    Since $P_2$ has at most $2\edit(P)$ vertices, it has a merge tree $T_2$ of width at most $2\edit(P)$ using a trivial tree as described in Lemma \ref{lem:any-pattern-created}. In particular, every vertex of $V_2$ is an anchor in every label of $T_2$.
    
    Let $T$ be the merge tree where the root is a node labeled by $P$ anchored with all vertices of $A$, and with two subtrees.
    The left subtree is tree $T_1$ and the right subtree is tree $T_2$.
    The root of $T_1$ (i.e. $P_1$) is labeled with a pattern anchored with all vertices of $A$, and
    the root of $T_2$ (i.e. $P_2$) is labeled with a pattern anchored at all vertices of $V_2$.
    It remains to check that the root of $T$ (i.e. $P$ anchored with $A$) yields a correct merge with its two children $P_1$ and $P_2$, by verifying conditions (a) through (e) of Definition 
    \ref{def:operations}:
    \begin{enumerate}[label=(\alph*)]
        \item $(E_1, E_2)$ is a partition of $E$ by definition.
        $\forall (p_i, p_j) \in E_1$, we have $p_i, p_j \in V_1=V$.
        $\forall (p_i, p_j) \in E_2$, we have $p_i, p_j \in V_2$ by definition of $V_2$.
        \item $V_1 \cup V_2 = V$, since $V_1 = V$.
        \item $V_1 \cap V_2 = V_2$, and every vertex in $V_2$ is an anchor in both $P_1$ and $P_2$.
        \item Anchors of pattern $P$ are vertices from $A$.
        $P_1$ receives all vertices from $P$, and has an anchor set also equals to $A$.
        $P_2$ receives vertices in $V_2$ from $P$, and has an anchor set equals to $V_2 = A \cap V_2$.
        \item Every vertex is in $V_1 = V$, so there is no vertex not in $V_1$.\qed
    \end{enumerate}
\end{proof}

We finish with another corollary. The \emph{outerplanar crossing number} (also called \emph{convex crossing number}, \emph{1-page book crossing number} or \emph{circular clrossing number}) is a popular graph parameter (see the entry \emph{Convex crossing number} in the survey~\cite{Schaefer12}). We adapt it to patterns and show that when it is bounded, the merge-width is bounded.

\begin{definition}
The \emph{outerplanar crossing number} of a pattern $P$ is the number of pairs of crossing edges of $P$, where two edges $(i,j)$ and $(i',j')$ cross when either $i<i'<j<j'$ or $i'<i<j'<j$.
\end{definition}

\begin{corollary}
If a pattern has outerplanar crossing number $c$ then it has merge-width at most $2c+2$. 
\end{corollary}

The proposition follows from the fact that the outerplanar crossing number of a pattern is always greater than or equal to the edit distance to outerplanar graph of that pattern. Indeed, for each pair of crossing edges, remove one of the two edges. By doing so, there is no more crossing edges, meaning that the pattern has become outerplanar.

\section{Linear-time detection of positive outerplanar forests}
\label{sec:forests}

In this section, we show that a natural family of pattern of arbitrarily large size such that any such pattern can be detected in linear time.
This is the family of patterns that are outerplanar (that is, without crossing of edges, see Section~\ref{sec:parametrized}), positive (that is, without forbidden edges), and acyclic. 
We call these \emph{positive outerplanar forests}.

At an intuitive level and at the current state of our knowledge, these three constraints seem necessary. 
Indeed, without the outerplanarity (or near outerplanarity) we have no idea how large patterns behave, and our lower bounds indicate that sparse graphs are easier to tackle. 
Also, when there are forbidden edges, one is often tempted to complement the input graph and this is impossible in linear time. 
Finally, as soon as there is a cycle, it seems that we are in a scenario that resembles the one of triangle detection, where we need matrix multiplication time.

This section is devoted to the proof of the following theorem.

\begin{theorem}\label{thm:forest}
    Positive outerplanar forests can be detected in time $O(n+m)$.
\end{theorem}

Let $P=(V,M,F,U)$ be a positive outerplanar forest.
Since there is no forbidden edges, we will simply say ``edges" to refer to the mandatory edges.

Let introduce some terminology related to the relative placement of edges.
See Figure~\ref{fig:tree-decomp}.
Consider a pair of edges $(i,j), (i',j')$ with $i<j$, $i'<j'$ and $i\leq i'$.  
In an outerplanar pattern, there are two configurations: 
\begin{itemize}
    \item either they are \emph{nested}: $i\leq i'<j'\leq j$,
    \item or \emph{side by side}: $i<j\leq i'<j'$.
\end{itemize}
(They cannot be crossing, which would be $i<i'<j<j'$.)
When edges are side by side, we can refer to the \emph{left edge} and to the \emph{right edge}.

When edges are nested ($i\leq i'<j'\leq j$), the nested edge (i.e. $(i',j')$) can be further specified:
\begin{itemize}
    \item $(i',j')$ is \emph{directly nested} inside  $(i,j)$ if there is no edge nested in between,
    \item otherwise, $(i',j')$ is \emph{deeply nested} inside $(i,j)$.
\end{itemize}

Next, we want to partition all edges nested inside one specific edge $(i,j)$.
To do so, we can easily partition edges directly nested inside $(i,j)$ into three parts, depending on how they are connected to edge $(i,j)$ without passing throught $(i,j)$:
\begin{itemize}
    \item edges connected by a path to $i$ (without using edge $(i,j)$), called \emph{left-nested edges};
    \item edges connected by a path  to $j$ (without using edge $(i,j)$), called \emph{right-nested edges};
    \item and edges not connected to neither $i$ nor $j$, called \emph{centered-nested edges}.
\end{itemize}
Note that since our pattern is acyclic, it is impossible to have en edge connected both to $i$ and $j$, as this would create a cycle. 
This kind of edge is called a double-nested edge, and we know we never encounter this kind of edge.
Then, for edges deeply nested, it means that they are nested inside some other edge $(i',j')$ which has been already classified (since every directly nested edge can be classified as either left-, right- or centered-nested edge). We give this edge the same classification as $(i',j')$.
All in all, we have classified every edge nested inside $(i,j)$ as either \emph{left-nested}, \emph{right-nested} or \emph{centered-nested edge}, with respect to $(i,j)$.

\begin{figure}[!h]
    \centering
    \tikzset{every picture/.style={line width=0.75pt}} 

\begin{tikzpicture}[x=0.75pt,y=0.75pt,yscale=-1,xscale=1]

\draw  [fill={rgb, 255:red, 0; green, 0; blue, 0 }  ,fill opacity=1 ] (90,115) .. controls (90,112.24) and (92.24,110) .. (95,110) .. controls (97.76,110) and (100,112.24) .. (100,115) .. controls (100,117.76) and (97.76,120) .. (95,120) .. controls (92.24,120) and (90,117.76) .. (90,115) -- cycle ;
\draw  [fill={rgb, 255:red, 0; green, 0; blue, 0 }  ,fill opacity=1 ] (140,115) .. controls (140,112.24) and (142.24,110) .. (145,110) .. controls (147.76,110) and (150,112.24) .. (150,115) .. controls (150,117.76) and (147.76,120) .. (145,120) .. controls (142.24,120) and (140,117.76) .. (140,115) -- cycle ;
\draw  [draw opacity=0] (95,115) .. controls (95,115) and (95,115) .. (95,115) .. controls (95,101.19) and (106.19,90) .. (120,90) .. controls (133.81,90) and (145,101.19) .. (145,115) -- (120,115) -- cycle ; \draw   (95,115) .. controls (95,115) and (95,115) .. (95,115) .. controls (95,101.19) and (106.19,90) .. (120,90) .. controls (133.81,90) and (145,101.19) .. (145,115) ;  
\draw  [fill={rgb, 255:red, 0; green, 0; blue, 0 }  ,fill opacity=1 ] (190,115) .. controls (190,112.24) and (192.24,110) .. (195,110) .. controls (197.76,110) and (200,112.24) .. (200,115) .. controls (200,117.76) and (197.76,120) .. (195,120) .. controls (192.24,120) and (190,117.76) .. (190,115) -- cycle ;
\draw  [draw opacity=0] (145,115) .. controls (145,115) and (145,115) .. (145,115) .. controls (145,101.19) and (156.19,90) .. (170,90) .. controls (183.81,90) and (195,101.19) .. (195,115) -- (170,115) -- cycle ; \draw   (145,115) .. controls (145,115) and (145,115) .. (145,115) .. controls (145,101.19) and (156.19,90) .. (170,90) .. controls (183.81,90) and (195,101.19) .. (195,115) ;  
\draw  [fill={rgb, 255:red, 0; green, 0; blue, 0 }  ,fill opacity=1 ] (220,115) .. controls (220,112.24) and (222.24,110) .. (225,110) .. controls (227.76,110) and (230,112.24) .. (230,115) .. controls (230,117.76) and (227.76,120) .. (225,120) .. controls (222.24,120) and (220,117.76) .. (220,115) -- cycle ;
\draw  [fill={rgb, 255:red, 0; green, 0; blue, 0 }  ,fill opacity=1 ] (270,115) .. controls (270,112.24) and (272.24,110) .. (275,110) .. controls (277.76,110) and (280,112.24) .. (280,115) .. controls (280,117.76) and (277.76,120) .. (275,120) .. controls (272.24,120) and (270,117.76) .. (270,115) -- cycle ;
\draw  [draw opacity=0] (225,115) .. controls (225,115) and (225,115) .. (225,115) .. controls (225,101.19) and (236.19,90) .. (250,90) .. controls (263.81,90) and (275,101.19) .. (275,115) -- (250,115) -- cycle ; \draw   (225,115) .. controls (225,115) and (225,115) .. (225,115) .. controls (225,101.19) and (236.19,90) .. (250,90) .. controls (263.81,90) and (275,101.19) .. (275,115) ;  
\draw  [fill={rgb, 255:red, 0; green, 0; blue, 0 }  ,fill opacity=1 ] (300,115) .. controls (300,112.24) and (302.24,110) .. (305,110) .. controls (307.76,110) and (310,112.24) .. (310,115) .. controls (310,117.76) and (307.76,120) .. (305,120) .. controls (302.24,120) and (300,117.76) .. (300,115) -- cycle ;
\draw  [fill={rgb, 255:red, 0; green, 0; blue, 0 }  ,fill opacity=1 ] (350,115) .. controls (350,112.24) and (352.24,110) .. (355,110) .. controls (357.76,110) and (360,112.24) .. (360,115) .. controls (360,117.76) and (357.76,120) .. (355,120) .. controls (352.24,120) and (350,117.76) .. (350,115) -- cycle ;
\draw  [draw opacity=0] (305,115) .. controls (305,115) and (305,115) .. (305,115) .. controls (305,101.19) and (316.19,90) .. (330,90) .. controls (343.81,90) and (355,101.19) .. (355,115) -- (330,115) -- cycle ; \draw   (305,115) .. controls (305,115) and (305,115) .. (305,115) .. controls (305,101.19) and (316.19,90) .. (330,90) .. controls (343.81,90) and (355,101.19) .. (355,115) ;  
\draw  [fill={rgb, 255:red, 0; green, 0; blue, 0 }  ,fill opacity=1 ] (380,115) .. controls (380,112.24) and (382.24,110) .. (385,110) .. controls (387.76,110) and (390,112.24) .. (390,115) .. controls (390,117.76) and (387.76,120) .. (385,120) .. controls (382.24,120) and (380,117.76) .. (380,115) -- cycle ;
\draw  [fill={rgb, 255:red, 0; green, 0; blue, 0 }  ,fill opacity=1 ] (405,115) .. controls (405,112.24) and (407.24,110) .. (410,110) .. controls (412.76,110) and (415,112.24) .. (415,115) .. controls (415,117.76) and (412.76,120) .. (410,120) .. controls (407.24,120) and (405,117.76) .. (405,115) -- cycle ;
\draw  [fill={rgb, 255:red, 0; green, 0; blue, 0 }  ,fill opacity=1 ] (460,115) .. controls (460,112.24) and (462.24,110) .. (465,110) .. controls (467.76,110) and (470,112.24) .. (470,115) .. controls (470,117.76) and (467.76,120) .. (465,120) .. controls (462.24,120) and (460,117.76) .. (460,115) -- cycle ;
\draw  [draw opacity=0] (385,115) .. controls (385,100.5) and (402.91,88.75) .. (425,88.75) .. controls (447.09,88.75) and (465,100.5) .. (465,115) -- (425,115) -- cycle ; \draw   (385,115) .. controls (385,100.5) and (402.91,88.75) .. (425,88.75) .. controls (447.09,88.75) and (465,100.5) .. (465,115) ;  
\draw  [draw opacity=0][line width=0.75]  (95,115) .. controls (95,79.1) and (177.83,50) .. (280,50) .. controls (382.17,50) and (465,79.1) .. (465,115) -- (280,115) -- cycle ; \draw  [line width=0.75]  (95,115) .. controls (95,79.1) and (177.83,50) .. (280,50) .. controls (382.17,50) and (465,79.1) .. (465,115) ;  
\draw  [fill={rgb, 255:red, 0; green, 0; blue, 0 }  ,fill opacity=1 ] (115,115) .. controls (115,112.24) and (117.24,110) .. (120,110) .. controls (122.76,110) and (125,112.24) .. (125,115) .. controls (125,117.76) and (122.76,120) .. (120,120) .. controls (117.24,120) and (115,117.76) .. (115,115) -- cycle ;
\draw  [fill={rgb, 255:red, 0; green, 0; blue, 0 }  ,fill opacity=1 ] (165,115) .. controls (165,112.24) and (167.24,110) .. (170,110) .. controls (172.76,110) and (175,112.24) .. (175,115) .. controls (175,117.76) and (172.76,120) .. (170,120) .. controls (167.24,120) and (165,117.76) .. (165,115) -- cycle ;
\draw  [fill={rgb, 255:red, 0; green, 0; blue, 0 }  ,fill opacity=1 ] (245,115) .. controls (245,112.24) and (247.24,110) .. (250,110) .. controls (252.76,110) and (255,112.24) .. (255,115) .. controls (255,117.76) and (252.76,120) .. (250,120) .. controls (247.24,120) and (245,117.76) .. (245,115) -- cycle ;
\draw  [fill={rgb, 255:red, 0; green, 0; blue, 0 }  ,fill opacity=1 ] (325,115) .. controls (325,112.24) and (327.24,110) .. (330,110) .. controls (332.76,110) and (335,112.24) .. (335,115) .. controls (335,117.76) and (332.76,120) .. (330,120) .. controls (327.24,120) and (325,117.76) .. (325,115) -- cycle ;
\draw  [fill={rgb, 255:red, 0; green, 0; blue, 0 }  ,fill opacity=1 ] (435,115) .. controls (435,112.24) and (437.24,110) .. (440,110) .. controls (442.76,110) and (445,112.24) .. (445,115) .. controls (445,117.76) and (442.76,120) .. (440,120) .. controls (437.24,120) and (435,117.76) .. (435,115) -- cycle ;
\draw    (120,115) -- (170,115) ;
\draw    (250,115) -- (275,115) ;
\draw    (305,115) -- (330,115) ;
\draw    (410,115) -- (440,115) ;

\draw (91,122.4) node [anchor=north west][inner sep=0.75pt]    {$i$};
\draw (461,122.4) node [anchor=north west][inner sep=0.75pt]    {$j$};
\draw (191,122.4) node [anchor=north west][inner sep=0.75pt]    {$k_{1}$};
\draw (221,122.4) node [anchor=north west][inner sep=0.75pt]    {$k_{2}$};
\draw (351,122.4) node [anchor=north west][inner sep=0.75pt]    {$k_{3}$};
\draw (381,122.4) node [anchor=north west][inner sep=0.75pt]    {$k_{4}$};
\draw (128,80.4) node [anchor=north west][inner sep=0.75pt]    {$e_{1}$};
\draw (163,78.4) node [anchor=north west][inner sep=0.75pt]    {$e_{2}$};
\draw (243,78.4) node [anchor=north west][inner sep=0.75pt]    {$e_{3}$};
\draw (323,78.4) node [anchor=north west][inner sep=0.75pt]    {$e_{4}$};
\draw (404,78.4) node [anchor=north west][inner sep=0.75pt]    {$e_{5}$};

\end{tikzpicture}
    \caption{Classification of edges nested inside $(i,j)$.
    Apart from edge $(i,j)$, edges represented by an arc are directly nested edges, and edges represented by a straight line are deeply nested edges. Directly nested edges are classified depending on whether there is a path from them to $i$, or $j$, or neither.
    Hence, edges $e_1$ and $e_2$ are left-nested (because if we avoid edge $(i,j)$, they are connected to $i$ and not to $j$), edges $e_3$ and $e_4$ are centered-nested because they are not connected to $(i,j)$, and edge $e_5$ is right-connected.
    Concerning edges deeply nested inside $(i,j)$, i.e. edges represented by a straight line, they have the same classification as the directly nested edge they are nested inside. Hence,
    every edge nested inside $e_1$ or $e_2$ is left-nested inside $(i,j)$,
    every edge nested inside $e_3$ or $e_4$ is centered-nested inside $(i,j)$,
    and every edge nested inside $e_5$ is right-nested inside $(i,j)$.
    All in all, we can summarize our example by saying that
    every edge with both endpoints between vertices $i$ and $k_1$ is left-nested,
    every edge with both endpoints between vertices $k_2$ and $k_3$ is centered-nested, and
    every edge with both endpoints between vertices $k_4$ and $j$ is right-nested.
    Such a partitioning using vertices $k_1$ to $k_4$ is always possible.}
    \label{fig:tree-decomp}
\end{figure}
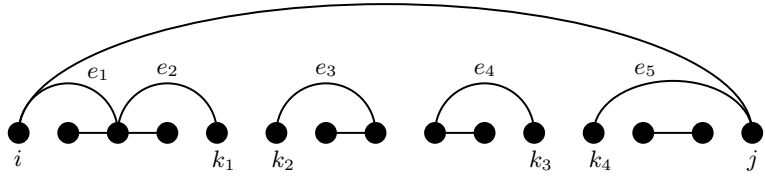


When we have a planar forest pattern $P$, we can also introduce some special edges and sub-pattern.

\renewcommand{\ll}{\texttt{LOME}}
\newcommand{\rr}{\texttt{ROME}}

\begin{definition}
    Let $P=(V,M,F,U)$ be a positive planar forest pattern. 
    \begin{itemize}
        \item Let $\ll$ be the left-outermost edge, i.e. the edge $(i,j) \in M$ such that $i$ is minimal, and in case of multiple such edges, such that $j$ is maximal.
        \item Let $\rr$ be the right-outermost edge, i.e. the decided edge $(i,j) \in M$ such that $j$ is maximal
        , and in case of multiple such edges, such that $i$ is minimal.
    \end{itemize}
\end{definition}
    Edges $\ll$ and $\rr$ have the particularity to divide $P$ into 3 particular sub-patterns each. See Figure~\ref{fig:tree-lome}.
    
    Let first discuss $\ll$.
    Let $\ll$ be the edge $(i_\ll,j_\ll)$ with $i_\ll < j_\ll$.
    By definition, if $(i,j) \in E$, then either $(i,j)$ is nested inside $(i_\ll, j_\ll)$ or is to the right of $(i_\ll , j_\ll)$.
    If $(i,j)$ is nested inside $(i_\ll, j_\ll)$, since $P$ is a forest,  $(i,j)$ cannot be double-nested.
    Hence,  $(i,j)$ is either left-, right-, or centered-nested.
    We define the three following sub-pattern:
    \begin{itemize}
        \item The pattern $P_L$ made of left-nested edges and centered-nested edges inside $(i_\ll, j_\ll)$.
        \item The pattern $P_R$ made of right-nested edges inside $(i_\ll, j_\ll)$.
        \item The pattern $P_O$ made of edges not nested inside $(i_\ll, j_\ll)$, i.e. edges to the right of $(i_\ll, j_\ll)$.
    \end{itemize}
    Note that every decided edge of $P$ is in exactly one sub-pattern $P_L$, $P_R$ or $P_O$, except for $\ll$ which is not in any pattern.

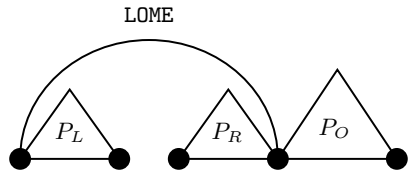
\begin{figure}[!h]
    \centering
    \tikzset{every picture/.style={line width=0.75pt}} 

\begin{tikzpicture}[x=0.75pt,y=0.75pt,yscale=-1,xscale=1]

\draw  [fill={rgb, 255:red, 0; green, 0; blue, 0 }  ,fill opacity=1 ] (70,115) .. controls (70,112.24) and (72.24,110) .. (75,110) .. controls (77.76,110) and (80,112.24) .. (80,115) .. controls (80,117.76) and (77.76,120) .. (75,120) .. controls (72.24,120) and (70,117.76) .. (70,115) -- cycle ;
\draw  [fill={rgb, 255:red, 0; green, 0; blue, 0 }  ,fill opacity=1 ] (120,115) .. controls (120,112.24) and (122.24,110) .. (125,110) .. controls (127.76,110) and (130,112.24) .. (130,115) .. controls (130,117.76) and (127.76,120) .. (125,120) .. controls (122.24,120) and (120,117.76) .. (120,115) -- cycle ;
\draw  [fill={rgb, 255:red, 0; green, 0; blue, 0 }  ,fill opacity=1 ] (150,115) .. controls (150,112.24) and (152.24,110) .. (155,110) .. controls (157.76,110) and (160,112.24) .. (160,115) .. controls (160,117.76) and (157.76,120) .. (155,120) .. controls (152.24,120) and (150,117.76) .. (150,115) -- cycle ;
\draw  [fill={rgb, 255:red, 0; green, 0; blue, 0 }  ,fill opacity=1 ] (200,115) .. controls (200,112.24) and (202.24,110) .. (205,110) .. controls (207.76,110) and (210,112.24) .. (210,115) .. controls (210,117.76) and (207.76,120) .. (205,120) .. controls (202.24,120) and (200,117.76) .. (200,115) -- cycle ;
\draw  [fill={rgb, 255:red, 0; green, 0; blue, 0 }  ,fill opacity=1 ] (260,115) .. controls (260,112.24) and (262.24,110) .. (265,110) .. controls (267.76,110) and (270,112.24) .. (270,115) .. controls (270,117.76) and (267.76,120) .. (265,120) .. controls (262.24,120) and (260,117.76) .. (260,115) -- cycle ;
\draw  [draw opacity=0] (75,115) .. controls (75,81.86) and (104.1,55) .. (140,55) .. controls (175.9,55) and (205,81.86) .. (205,115) -- (140,115) -- cycle ; \draw   (75,115) .. controls (75,81.86) and (104.1,55) .. (140,55) .. controls (175.9,55) and (205,81.86) .. (205,115) ;  
\draw   (235,70) -- (265,115) -- (205,115) -- cycle ;
\draw   (180,80) -- (205,115) -- (155,115) -- cycle ;
\draw   (100,80) -- (125,115) -- (75,115) -- cycle ;

\draw (140,42.4) node [anchor=center][inner sep=0.75pt]    {$\ll$};
\draw (91,95.4) node [anchor=north west][inner sep=0.75pt]    {$P_{L}$};
\draw (170,95.4) node [anchor=north west][inner sep=0.75pt]    {$P_{R}$};
\draw (224,93.4) node [anchor=north west][inner sep=0.75pt]    {$P_{O}$};

\end{tikzpicture}
    \caption{Outline of sub-patterns induced by edge $\ll$. 
    A triangle represents a sub-pattern.
    Note that any of these patterns can be empty. Furthermore, patterns $P_L$ and $P_O$ can be disconnected from $\ll$. This happens because $P_L$ can contain only centered-nested edges, and because  $P_O$ is not required to be connected to $\ll$. On the other hand, $P_R$ is always connected to $\ll$ by definition.}
    \label{fig:tree-lome}
\end{figure}

    We have the same 3 kinds of sub-patterns for $\rr$:
    \begin{itemize}
        \item The pattern $P_L$ made of left-nested edges and centered-nested edges inside $(i_\rr, j_\rr)$.
        \item The pattern $P_R$ made of right-nested edges inside $(i_\rr, j_\rr)$.
        \item The pattern $P_O$ made of edges not nested inside $(i_\rr, j_\rr)$, i.e. edges to the left of $(i_\rr, j_\rr)$.
    \end{itemize}

Given a pattern $P=(V(P), F(P), M(P), U(P))$, a graph $G=(V(G), E(G))$ and a subgraph $H$ of $G$ that realises $P$, we can consider $u=\min_\tau V(H)$ and $v=\max_\tau V(H)$. This means that according to order $\tau$, every vertex in $H$ is in the interval $[u..v]$.
An interesting set is the set $S$ of all $[u..v]$ for every subgraph $H$ that realizes $P$.
However, such a set can be quadratic in the size of $G$.
If we want to certify patterns in linear time, we do not have access to this set.
Hence, we consider two arrays ($m^+_P$ and $m^-_P$) that are enough to run our algorithm, as well as two other arrays ($M^+_P$ and $M^-_P$) that can  be deduced from the other two in linear time:

\begin{definition}
    Let $m^+_P(u)$ be the smallest $v$ such that there exists a subgraph $H$ of $G$ that realize $P$, with $u=\min_\tau V(H)$ and $v=\max_\tau V(H)$. If there is no such $H$, we set $m^+_P(u) = +\infty$.
    Let $M^+_P(u) = m^+_P(>u) = \min_{u' > u} m^+_P(u')$. This means that $M^+_P(u)$ is the smallest $v$ among all subgraphs $H$ of $G$ that realize $P$, with $u<\min_\tau V(H)$ and $v=\max_\tau V(H)$. 
    
    In a similar way, let $m^-_P(v)$ be the biggest $u$ among all subgraphs $H$ of $G$ that realize $P$, with  $u=\min_\tau V(H)$ and $v=\max_\tau V(H)$. If there is no such $H$, we set $m^-_P(v) = -\infty$.
    Furthermore, $M^-_P(v) = m^+_P(<v) = \max_{v' < v} M^-_P(v')$.
\end{definition}

We have the following lemma:

\begin{lemma}
    Let $P=(V(P), M(P), F(P), U(P))$ be a positive planar tree with $k$ edges, let $G=(V(G), E(G))$ be a graph with $n$ vertices and $m$ edges, and let $f$ be either the function $m^+_P(\cdot)$ or $m^-_P(\cdot)$.
    Then we can compute all values of $f$ in time $ck(n+m)$ for some constant $c$.
\end{lemma}

\begin{proof}
    We prove this by induction on $k$.
    The base case is $k=1$, which consists of a pattern with only one edge.
    The values of $f$ can be computed in $O(n+m)$ by scanning the adjacency list of every vertex of $G$.
    
    Now we suppose that the lemma is true for all patterns with $k$ edges, and we prove it for a pattern $P$ with $k+1$ edges.
    First, let us do the case $f=m^+_P$.
    We consider the $\ll$ decomposition. Let $P_L, P_R$ and $P_O$ be the corresponding 3 sub-patterns.
    By induction, we can compute
    all values of $m^+_{P_L}$ in time $c|E(P_L)|(n+m)$,
    all values of $m^-_{P_R}$ in time $c|E(P_R)|(n+m)$, and
    all values of $m^+_{P_O}$ in time $c|E(P_O)|(n+m)$.
    If a sub-pattern is empty, its corresponding function ($m^+_{P_L}, m^-_{P_R}$ or $ m^+_{P_O}$) is the identity function.
    So far, we have spend time $c(|E(P_L)| + |E(P_R)| + |E(P_O)|)(n+m) = ck(n+m)$.
    We can also compute the values of $M^+_{P_L}$ and $M^+_{P_O}$ in time $2n$ using values from $m^+_{P_L}$ and $m^+_{P_O}$ and scanning through them.
    
    Now, for each vertex $u$ of $G$, let's compute $f(u)$.
    This is done by going through all edges $(u,v)$ of $G$ such that
    $u \leq m^+_{P_L}(u) < m^-_{P_R}(v) \leq v$,
    and returning the minimum value of $m^+_{P_O}(v)$.
    However, if $P_L$ is not connected to \ll, we replace $m^+_{P_L}$ by $M^+_{P_L}$ in the above formula, and if $P_O$ is not connected to \ll, we replace $m^+_{P_O}$ by $M^+_{P_O}$ in the above minimum
    All of this is done in time $3m$.

    All in all, we have computed all values of $f(u)$ in time $ck(n+m) + 2n + 3m \leq c(k+1)(n+m)$ providing that $c \geq 5$.

    We do a similar thing for $f=m^-_P$, by using $\rr$ instead of \ll. \qed
\end{proof}

This is enough to prove Theorem \ref{thm:forest}. To detect a positive outerplanar forest~$P$, we compute the value of $m^+_P(u)$ for each vertex $u \in V(G)$. If the value is $+\infty$ for every vertex, pattern $P$ is not realized in $G$. Otherwise, $P$ is realized in $G$.

\section{Linear detection of every positive $P_4$} 
\label{sec:positive-P4}


In this section, we prove that the positive patterns on four vertices such that the mandatory edges form a path on four vertices ($P_4$) can be detected in linear time. 
Note that this holds for any ordering of $P_4$, not just the natural one. 
Compared to Section~\ref{sec:forests} about of positive outerplanar forests, this is both more general in the sense that there can be edge crossings, and more restricted since we have a short path of mandatory edges, and not a forest.


\begin{theorem}
All positive patterns on four vertices, such that the mandatory edges form a $P_4$ can be detected in linear time.
\end{theorem}

\subsection{Some notations}

Let $P$ be one of the patterns on four vertices whose mandatory edges form a $P_4$.
Pattern~$P$ is made up of three edges: one in the middle and two on the sides.
As an example, in the following pattern, the edge $e$ is the middle edge of $P_4$.

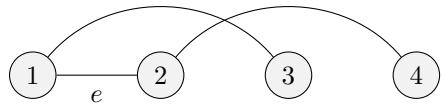
\begin{center}
    \begin{tikzpicture}
  [scale=1.7,
  auto=left,every node/.style
  ={circle,draw,fill=black!5}]
  \node (1) at (0,0) {1};
  \node (2) at (1,0) {2};
  \node (3) at (2,0) {3};
  \node (4) at (3,0) {4};
  \draw (1) to node[draw=none, fill=none, label=below:$e$] {} (2);
  \draw[bend left=50] (2) to (4) ;
  \draw[bend left=50] (1) to (3) ;
\end{tikzpicture}
\end{center}

Let $e=(i,j)$ be an edge on a graph ordered by $\tau$. We define:
\begin{itemize}[noitemsep]
  \item $e^+ = (i,j')$, where $j'$ is the smallest vertex greater than $j$ such that $(i,j')$ is an edge.
  \item $e^- = (i,j')$, where $j'$ is the greatest vertex smaller than $j$ such that $(i,j')$ is an edge.
  \item $^+e = (i',j)$, where $i'$ is the smallest vertex greater than $i$ such that $(i',j)$ is an edge.
  \item $^-e = (i',j)$, where $i'$ is the greatest vertex smaller than $i$ such that $(i',j)$ is an edge.
\end{itemize}

If $e=(i,j)$ is an edge,
let $\min e$ be the smallest vertex among $i$ and $j$, and
let $\max e$ be the greatest vertex among $i$ and $j$.
This notation is useful when we consider an edge $e$ by its name rather than by its endpoints $(i,j)$.

In linear time $O(n+m)$ we can build a lookup table that contains the value of $e^+$, $e^-$, $^+e$ and $^-e$ for each edge $e$.
This is done by sorting all adjacency lists in $O(n+m)$ and then scanning through them.

Further more, we can build in linear time $O(n+m)$ a lookup table that contains the value of $\min N^-(i)$, $\max N^-(i)$, $\min N^+(i)$ and $\max N^+(i)$ for all vertex $i$.

As a consequence, the min and the max of $N^-(i)$, $N^+(i)$, $e^+$, $e^-$, $^+e$ and $^-e$ can be found in constant time for every vertex $i$ and edge $e$.

\subsection{Focus on one case}

We will explain the algorithm on one specific example, the pattern above (and below), and then explain how to adapt the proof.

The main idea is to scan through all edges $e$, and check if a pattern $P$ can be found with $e$ as the middle edge.
The other key idea is to look for only one potential pattern for each edge $e$, but ensuring that if a pattern is not found, then there is no pattern $P$ with $e$ as the middle edge.

In the example below, for each edge $e$, we will only try to detect a pattern that is maximal in the sense that, in the realization:
\begin{itemize}[noitemsep, nolistsep]
  \item vertex 3 is max $e^+$
  \item vertex 4 is max $N^+(j)$.
\end{itemize}
This will be enough for the detection problem.

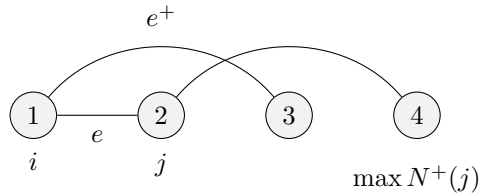
\begin{center}
\begin{tikzpicture}
  [scale=1.7,
  auto=left,every node/.style
  ={circle,draw,fill=black!5}]
  \node[label=below:$i$] (1) at (0,0) {1};
  \node[label=below:$j$] (2) at (1,0) {2};
  \node (3) at (2,0) {3};
  \node (4) at (3,0) {4};
  \node[draw=none, fill=none] at (3,-0.5) {$ \max N^+(j)$};
  \draw (1) to node[draw=none, fill=none, label=below:$e$] {} (2);
  \draw[bend left=50] (2) to (4) ;
  \draw[bend left=50] (1) to node[draw=none, fill=none] {$e^+$} (3) ;
\end{tikzpicture}
\end{center}

There are two cases:
If $\max e^+ < \max N^+(j) $ , then we have found a pattern $P$.
If $\max e^+ \geq \max N^+(j) $ , then there is no pattern $P$ with $e$ as the middle edge.
Let us prove this second case.

\begin{proposition}
  If $\max e^+ \geq \max N^+(j) $, then there is no pattern $P$ with $e$ as the middle edge.
  \label{claimPhiMiddle}
\end{proposition}

\begin{proof}
  Suppose for a contradiction that there exists a pattern $P$ with $e$ as the middle edge and that $\max e^+ \geq \max N^+(j) $.
  Since there is a pattern $P$ with $e$ as the middle edge, it exists two vertices $i' < j'$ such that $(i,i')$ and $(j,j')$ are both an edge.
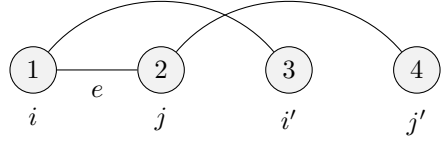
\begin{center}

  \begin{tikzpicture}
    [scale=1.7,
    auto=left,every node/.style
    ={circle,draw,fill=black!5}]
    \node[label=below:$i$] (1) at (0,0) {1};
    \node[label=below:$j$] (2) at (1,0) {2};
    \node[label=below:$i'$] (3) at (2,0) {3};
    \node[label=below:$j'$] (4) at (3,0) {4};
    \draw (1) to node[draw=none, fill=none, label=below:$e$] {} (2);
    \draw[bend left=50] (2) to (4) ;
    \draw[bend left=50] (1) to (3) ;
  \end{tikzpicture}
  \end{center}

  Since $j'\in N^+(j)$, we have $j' \leq\max N^+(j)$.
  Since $i'$ is a vertex greater than $j$ such that $(i,i')$ is an edge, and since $\max e^+$ is the smallest vertex greater than $j$ such that $(i, \max e^+)$ is an edge, we have that $\max e^+ \leq i'$.
  All in all, we have that $ j' \leq\max N^+(j) \leq \max e^+ \leq i'$, meaning that $j' \leq i'$, which contradicts the fact that $i' < j'$.\qed
\end{proof}

Thus, the pattern $P$ can be detected by scanning through every edge $e$ of the graph $G$, and check the condition $\max e^+ < \max N^+(j)$.
Since this condition can be checked in constant time, providing a linear pre-computation, $P$ can be detected in linear time.
In order to generalize this idea, let $\phi$ be a Boolean function that takes $e$ as input and outputs the Boolean ``$\max e^+ < \max N^+(j)$".
The detection algorithm can now be written like the algorithm \ref{certifP4meta}

\begin{algorithm}
  \DontPrintSemicolon
  \SetAlgoLined
  \SetKwInOut{Input}{Input}
  \SetKwInOut{Output}{Output}

  \Input{$G$ a graph given by its adjacency lists and $\tau$ a total ordering of the vertices}
  \Output{Does $(G,\tau)$ contain the pattern $P$}

  \BlankLine
  Wlog we suppose $\tau=1, 2, \dots,n$\;
  \For {$e$ in $E$}
  {\If{ $\phi(e)$}
  {STOP  Output ``YES, $P$ found"}}
  Output ``NO, $P$ not found"

  \caption{Detecting the pattern $P$}\label{certifP4meta}
\end{algorithm}

\subsection{The other cases}
In order to prove that the twelve version of a positive $P_4$ can be detected in linear time, it suffices to find a Boolean function $\phi$ that can be evaluated in constant time (eventually, providing a linear pre-computation),
together with a proof that ``If $\phi(e)$ is false, then there is no pattern $P$ with $e$ as the middle edge."
Here is the list of eight patterns of $P_4$, including the one shown as an example, but excluding the symmetries obtained by reversing the order $\tau$.
Each pattern is presented along with its function $\phi(e)$.
Proofs of ``If $\phi(e)$ is false, then there is no pattern $P$ with $e$ as the middle edge." are not given, as there are all very similar to Proposition \ref{claimPhiMiddle}.\\

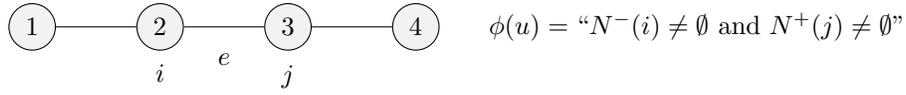
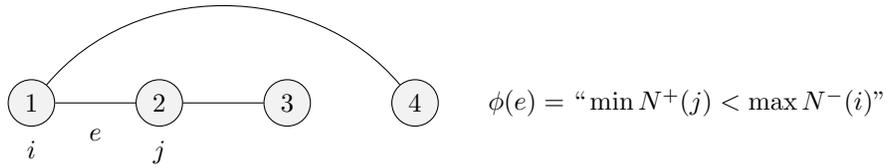
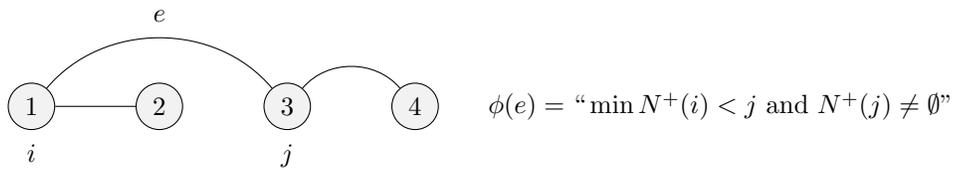
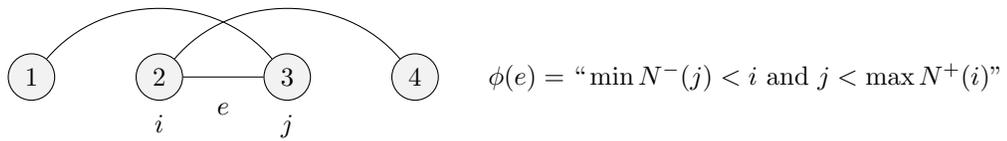
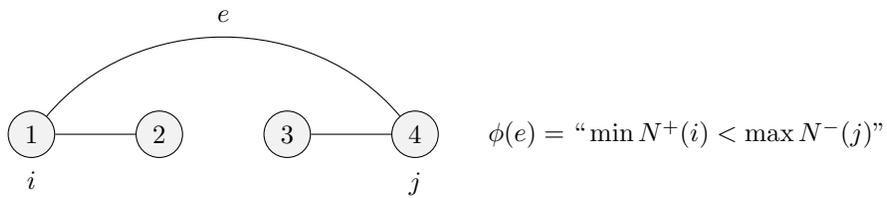
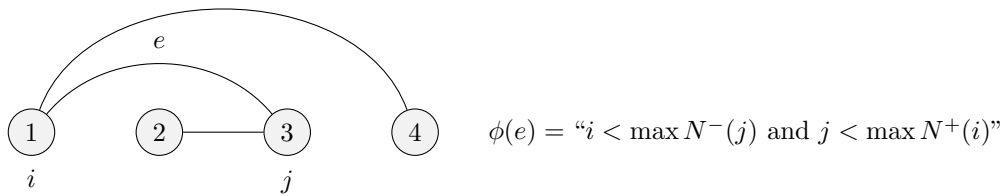
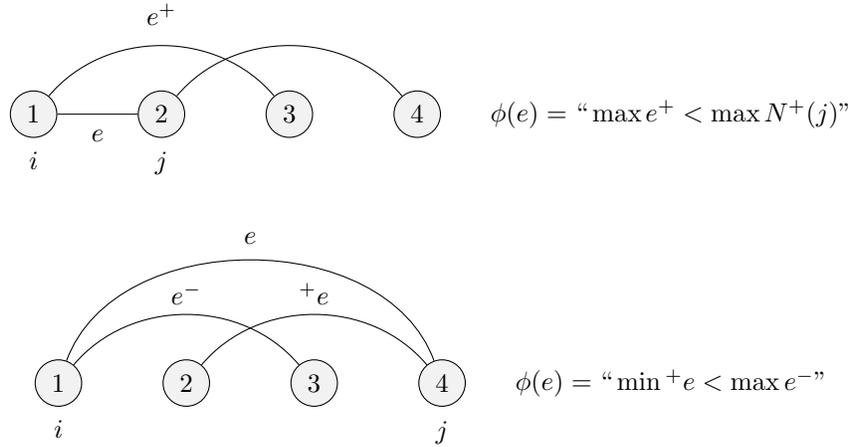
\begin{center}
\begin{tikzpicture}
  [scale=1.7,
  auto=left,every node/.style
  ={circle,draw,fill=black!5}]
  \node (1) at (0,0) {1};
  \node[label=below:$i$] (2) at (1,0) {2};
  \node[label=below:$j$] (3) at (2,0) {3};
  \node (4) at (3,0) {4};
  \node[draw=none, fill=none] at (1.5,-0.25) {$e$};
  \node[draw=none, fill=none, rectangle, anchor=west] at (3.5,0) {$\phi(u) = ``N^-(i) \neq \emptyset$ and $N^+(j) \neq \emptyset"$};
  \draw[left=50] (1) to (2);
  \draw[left=50] (2) to (3) ;
  \draw[left=50] (3) to (4) ;
\end{tikzpicture}
\end{center}

\begin{center}

\begin{tikzpicture}
  [scale=1.7,
  auto=left,every node/.style
  ={circle,draw,fill=black!5}]
  \node[label=below:$i$] (1) at (0,0) {1};
  \node[label=below:$j$] (2) at (1,0) {2};
  \node (3) at (2,0) {3};
  \node (4) at (3,0) {4};
  \node[draw=none, fill=none] at (0.5,-0.25) {$e$};
  \node[draw=none, fill=none, rectangle, anchor=west] at (3.5,0) {$\phi(e) = ``\min N^+(j) < \max N^-(i)"$};
  \draw[left=50] (1) to (2);
  \draw[left=50] (2) to (3) ;
  \draw[bend left=50] (1) to (4) ;
\end{tikzpicture}
\end{center}

\begin{center}
\begin{tikzpicture}
  [scale=1.7,
  auto=left,every node/.style
  ={circle,draw,fill=black!5}]
  \node[label=below:$i$] (1) at (0,0) {1};
  \node (2) at (1,0) {2};
  \node[label=below:$j$] (3) at (2,0) {3};
  \node (4) at (3,0) {4};
  \node[draw=none, fill=none, rectangle, anchor=west] at (3.5,0) {$\phi(e) = ``\min N^+(i) < j$ and $N^+(j) \neq \emptyset"$};
  \draw[left=50] (1) to (2);
  \draw[bend left=50] (3) to (4) ;
  \draw[bend left=50] (1) to node[draw=none, fill=none] {$e$} (3) ;
\end{tikzpicture}
\end{center}

\begin{center}
\begin{tikzpicture}
  [scale=1.7,
  auto=left,every node/.style
  ={circle,draw,fill=black!5}]
  \node (1) at (0,0) {1};
  \node[label=below:$i$] (2) at (1,0) {2};
  \node[label=below:$j$] (3) at (2,0) {3};
  \node (4) at (3,0) {4};
  \node[draw=none, fill=none] at (1.5,-0.25) {$e$};
  \node[draw=none, fill=none, rectangle, anchor=west] at (3.5,0) {$\phi(e) = ``\min N^-(j) < i$ and $j < \max N^+(i)"$};
  \draw[left=50] (2) to (3);
  \draw[bend left=50] (2) to (4) ;
  \draw[bend left=50] (1) to (3) ;
\end{tikzpicture}
\end{center}

\begin{center}
    \begin{tikzpicture}
  [scale=1.7,
  auto=left,every node/.style
  ={circle,draw,fill=black!5}]
  \node[label=below:$i$] (1) at (0,0) {1};
  \node (2) at (1,0) {2};
  \node (3) at (2,0) {3};
  \node[label=below:$j$] (4) at (3,0) {4};
  \node[draw=none, fill=none, rectangle, anchor=west] at (3.5,0) {$\phi(e) = ``\min N^+(i) < \max N^-(j)"$};
  \draw[left=50] (1) to (2);
  \draw[left=50] (3) to (4) ;
  \draw[bend left=50] (1) to node[draw=none, fill=none] {$e$} (4) ;
\end{tikzpicture}
\end{center}

\begin{center}
\begin{tikzpicture}
  [scale=1.7,
  auto=left,every node/.style
  ={circle,draw,fill=black!5}]
  \node[label=below:$i$] (1) at (0,0) {1};
  \node (2) at (1,0) {2};
  \node[label=below:$j$] (3) at (2,0) {3};
  \node (4) at (3,0) {4};
  \node[draw=none, fill=none, rectangle, anchor=west] at (3.5,0) {$\phi(e) = ``i < \max N^-(j) $ and $ j < \max N^+(i)"$};
  \draw[bend left=70] (1) to (4) ;
  \draw[bend left=50] (1) to node[draw=none, fill=none] {$e$} (3);
  \draw[left=50] (2) to (3) ;
\end{tikzpicture}
\end{center}

\begin{center}

\begin{tikzpicture}
  [scale=1.7,
  auto=left,every node/.style
  ={circle,draw,fill=black!5}]
  \node[label=below:$i$] (1) at (0,0) {1};
  \node[label=below:$j$] (2) at (1,0) {2};
  \node (3) at (2,0) {3};
  \node (4) at (3,0) {4};
  \draw (1) to node[draw=none, fill=none, label=below:$e$] {} (2);
  \node[draw=none, fill=none, rectangle, anchor=west] at (3.5,0) {$\phi(e) = ``\max e^+ < \max N^+(j) "$};
  \draw[bend left=50] (2) to (4) ;
  \draw[bend left=50] (1) to node[draw=none, fill=none] {$e^+$} (3) ;
\end{tikzpicture}
\end{center}

\begin{center}

\begin{tikzpicture}
  [scale=1.7,
  auto=left,every node/.style
  ={circle,draw,fill=black!5}]
  \node[label=below:$i$] (1) at (0,0) {1};
  \node (2) at (1,0) {2};
  \node (3) at (2,0) {3};
  \node[label=below:$j$] (4) at (3,0) {4};
  \node[draw=none, fill=none] at (1, 0.7) {$e^-$};
  \node[draw=none, fill=none] at (2, 0.7) {$^+e$};
  \node[draw=none, fill=none, rectangle, anchor=west] at (3.5,0) {$\phi(e) = ``\min {}^+e < \max e^- "$};
  \draw[bend left=70] (1) to node[draw=none, fill=none] {$e$} (4);
  \draw[bend left=50] (1) to (3) ;
  \draw[bend left=50] (2) to (4) ;
\end{tikzpicture}
\end{center}

\section{Patterns arising from geometry}
\label{sec:geometry}

Most of the patterns studied in the literature have three vertices. 
A set of patterns on four vertices that has been studied recently is the one of~\cite{FH22}, that appears naturally when studying some intersection graphs with an underlying ordering. In this section we will prove detection results for some of these patterns.

\begin{definition}
\label{def:Pabcd}
Consider patterns on four nodes $\{1,2,3,4\}$. We use the following names for some pairs of nodes: $a=(1,2)$, $b=(2,3)$, $c=(3,4)$, $d=(1,4)$ (see Figure~\ref{fig:abcd-patterns}).
Let $F \subseteq \{a,b,c,d\}$, the pattern $P_{F}$ is the pattern on four vertices, with mandatory edges $(1,3)$ and $(2,4)$, and forbidden set $F$. The class $\CC_{F}$ is the class where $P_{F}$ is forbidden.
\end{definition}

\begin{figure}[!h]
\centering
\scalebox{1}{
\begin{tikzpicture}
	[scale=1.7, 
	auto=left,every node/.style
	={circle,draw,fill=black!5}]
	\node (1) at (0,0) {1};
	\node (2) at (1,0) {2};
	\node (3) at (2,0) {3};
	\node (4) at (3,0) {4};
	
	\draw[dashed] (1) to node[draw=none, fill=none] {$a$} (2);
	\draw[bend left=50] (1) to (3);
	\draw[bend left=60, dashed] %
	(1) to node[draw=none, fill=none] {$d$} (4);
	\draw[dashed] (2) to node[draw=none, fill=none] {$b$} (3) ;
	\draw[bend left=50] (2) to (4) ;
	\draw[dashed] (3) to node[draw=none, fill=none] {$c$} (4);
\end{tikzpicture}}
\caption{Illustration of $P_{abcd}$ (Definition~\ref{def:Pabcd}).}
\label{fig:abcd-patterns}
\end{figure}
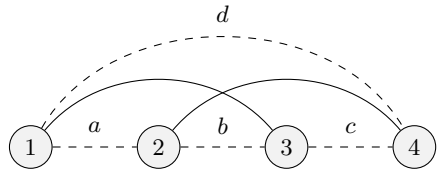

Through this section, we prove the following theorems:

\begin{theorem}
Patterns $P_{\emptyset}$,  $P_a$, $P_b$, $P_c$, $P_{ab}$, $P_{bc}$ can be detected in linear time.
\end{theorem}

\begin{theorem}
Pattern $P_{ab}$ is as hard as the co-comparability pattern on three vertices.
\end{theorem}

Note that according to Proposition \ref{lem:lower-bound-2k2}, assuming Hypothesis~\ref{hyp:k-clique} for $k=4$, we also know that detecting pattern $P_{abcd}$ requires at least $n^{\omega(2,1,1)-o(1)}$ time, where $n^{\omega(2,1,1)}$ is the time needed to multiply a $n^2$ by $n$ matrix and a $n$ by $n$ matrix.
We do not address remaining patterns (i.e. patterns with $|F|=3$ and patterns with $d \in F$), we leave this for further work.

\medskip

We start with pattern $P_\emptyset$.

\begin{lemma}
The pattern $P_{\emptyset}$ can be detected in linear time.
\end{lemma}

\begin{proof}
 To test whether a given ordering $\tau$  of the vertices avoids the pattern $P_{\emptyset}$, we first sort the adjacency list of each vertex, such that the neighbors of a vertex appear in the same order in the list and in $\tau$. Then every edge $uv$ with $u <_{\tau}v$ can be viewed as a pair of parenthesis $(_u$ and $)_v$. Then the test is just the verification that $\tau$ yields a good word of parenthesis. It is well-known that it can be done in one scan of the ordering $\tau$ using a simple stack.\qed
 \end{proof}


Let us consider now the  pattern $P_a$. The main idea is to check during the unique scan if the current vertex $i$ is  the third vertex of a pattern $P_a$.

\begin{proposition}
The pattern $P_a$ can be detected in linear time.
\end{proposition}

\begin{proof}

\begin{algorithm}[ht] 
	\DontPrintSemicolon
    \SetAlgoLined
		\SetKwInOut{Input}{Input}
	\SetKwInOut{Output}{Output}
	\SetKw{Continue}{continue}
	\SetKw{Break}{break}

	\Input{$G$ a graph given by its adjacency lists and $\tau$ a total ordering of the vertices}
	\Output{Does $(G,\tau)$ contain the pattern $P_a$}
	
	\BlankLine
	Adjacency lists are supposed to be ordered with $\tau$ increasing \;
	Wlog we suppose $\tau=1, 2, \dots,n$\;
	$ACTIVE$ is a  list of vertices initialized to $\emptyset$ \;
	
	For every vertex i, partition  its adjacency into 2 lists $N^-(i)=\{x \in N(i)$ s.t. $x<i \}$ and $N^+(i)=\{x \in N(i)$ s.t. $i<x \}$ \;
	
	\For{$i=1,2$}
	{If {$i$ admits a neighbour $x$  with $3\leq_{\tau}x$}\;
	{$append(i, ACTIVE)$}}

	\For {$ i=3$ to $n-1$}
	{\If{ $N^-(i) \neq \emptyset$ and $ACTIVE \neq \emptyset$ and $j \leftarrow  max_{\tau} \{ p  \in ACTIVE$  that has a neighbour $q$ with $i <_{\tau} q\}$ exists and $j \neq First(ACTIVE)$}{
        // $j$ is the largest true active vertex at step $i$ \;
	\If {$NOT (\{x \in N^-(i)$ with $x <_{\tau} j \} \subseteq  N^-(j))$}
	{STOP  Output ``YES, $P_a$ found"}
	\Else{For every  $u \in N^-(i)$: delete $i$ from $N^+(u)$ \;
	\If{$N^+(u) = \emptyset$}{delete u from $ACTIVE$}
	\If{$N^+(i) \neq \emptyset$}{ $append(i, ACTIVE)$}
	}}}
	Output ``NO, $P_a$ not found"	
	
        \vspace{0.5cm}
	\caption{Detecting the pattern $P_a$}\label{certifPa}
\end{algorithm}

We prove that Algorithm \ref{certifPa} detects in linear time if an ordering  contains the pattern $P_a$ or not.
Let us denote by $ACTIVE_i$ the value of the list ACTIVE at the beginning of the $i^{th}$ step of the main \textbf{for} loop.

If the answer is YES it means that for some $i$ we found a $j$ such that $NOT (\{x \in N^-(i)$ with $x <_{\tau} j \} \subseteq  N^-(j))$ and therefore there exists $x <_{\tau} i$ with $xi \in E(G)$ and $xj \notin E(G)$. Therefore, $x \in ACTIVE_i$, and since $j$ is also in $ACTIVE_i$ and not the first of the list and has at least one neighbour bigger than $i$, it yields the $P_a$ pattern.

Now let us consider the main invariant of the NO cases.

\textbf{Main Invariant:} at the end of the $i^{th}$ iteration of the \textbf{for} loop, the ordering $\tau$ does not contain any pattern $P_a$ on vertices $\alpha <_{\tau} \beta <_{\tau}\gamma <_{\tau} \delta$  with $\gamma \leq_{\tau} i$.

Let us prove it by induction.  With $i=3$ the algorithm checks the existence of a pattern $P_a$ on the vertices $1, 2, 3, x$ with $i<_{\tau}x$. So the invariant is correct for $i=3$.

Suppose now that at the end of the $i^{th}$  ($i \geq 3)$ iteration of the \textbf{for} loop, the ordering $\tau$ does not contain any pattern $P_a$ on vertices $\alpha <_{\tau} \beta <_{\tau}i <_{\tau} \delta$.

If $\beta$ is the largest true active vertex as compute in the algorithm ($\beta=j$), then the pattern has been checked, a contradiction.

Else $\beta <_{\tau} j$. Since $\alpha i \in E(G)$, we also know that $\alpha j \in E(G)$.  So we have a pattern $P_a$ with  $\alpha <_{\tau} \beta <_{\tau}j <_{\tau} \delta$ and $j<_{\tau} i$, contradicting the induction hypothesis.

\textbf{Complexity analysis} 
The preprocessing requires $O(n+m)$. The algorithm itself is just one scan of the $\tau$ ordering. and in the \textbf{for} loop the neighbourhood of a vertex is only considered once. So the whole algorithm can be done in $O(n+m)$.

It should be noticed that when the list ACTIVE is scanned to find the first true active vertex if exists, the vertices that do not satisfy the neighbourhood condition will be deleted from the ACTIVE list at the end of the \textbf{for} loop, so they will not be considered again.\qed
\end{proof}

We address $P_c$ by symmetry: to check for $P_c$ in a graph $G$, we mirror the graph in linear time and check for $P_a$ in the mirrored graph.
We now focus on pattern $P_b$.

\begin{proposition}
The pattern $P_b$ can be detected in linear time.
\end{proposition}
\begin{proof}

To prove this Proposition, we  present and analyze a dedicated detection algorithm \ref{CertifPb}

\begin{algorithm}[ht]
	\DontPrintSemicolon
    \SetAlgoLined
		\SetKwInOut{Input}{Input}
	\SetKwInOut{Output}{Output}
	\SetKw{Continue}{continue}
	\SetKw{Break}{break}

	\Input{$G$ a graph given by its adjacency lists and $\tau$ a total ordering of the vertices}
	\Output{Does $(G,\tau)$ contain the pattern $P_b$}
	
	\BlankLine
	Adjacency lists are supposed to be ordered with $\tau$ increasing \;
	Wlog we suppose $\tau=1, 2, \dots,n$\;
	$ACTIVE$ is a double linked  list of vertices initialized to $\emptyset$ \;
	
	For every vertex i, partition  its adjacency into 2 lists $N^-(i)=\{x \in N(i)$ s.t. $x<i \}$ and $N^+(i)=\{x \in N(i)$ s.t. $i<x \}$ \;
	
	\For{$i=1,2$}
	{\If {$i$ admits a neighbour $x$  with $3\leq_{\tau}x$}{$append(i, ACTIVE)$}}
	\For {$ i=3$ to $n-1$}
	{For every  $u \in N^-(i)$: delete $i$ from $N^+(u)$ \;
	\If{$N^+(u) = \emptyset$}{delete u from $ACTIVE$}
	\If{ $|N^-(i)| \geq 2$  and $ACTIVE \neq \emptyset$}
        { $j \leftarrow First(N^-(i)) $ \;
	delete $j$ from $N^-(i)$ \;
	\If {$ACTIVE(]j, i[) \neq N^-(i)$}{STOP  Output ``YES, $P_b$ found"}
	\If{$N^+(i) \neq \emptyset$}{ $append(i, ACTIVE)$}
	}}
	Output ``NO, $P_b$ not found"

        \vspace{0.5cm}
	\caption{Detecting $P_b$}\label{CertifPb}
\end{algorithm}

\textbf{Nota Bene:}  $ACTIVE$ is not only a double linked list, but we also need that for every vertex $u$  which belongs to  $ACTIVE$, we maintain a pointer to its place in $ACTIVE$. Therefore $ACTIVE(]j, i[)$ is the sublist $ACTIVE$ made up with the vertices $k$ in $ACTIVE$  such that  $j < k < i $.

\textbf{Proof of the algorithm}

When the algorithm outputs YES, we have found a vertex $i$ corresponding to the vertex 3 in Figure \ref{fig:abcd-patterns}:
Since $N^-(i) \neq \emptyset$,  $i$ admits  a neighbour  $u <i$, and since $ACTIVE(]j, i[)$ and $N^-(i)$ are different, there exists a vertex $v$
not adjacent to $i$ with $v <i$. Furthermore  since $v \in ACTIVE$  and $v$ is not adjacent to $i$, it admits a neighbour $w>_{\tau} i$.  So $\{u, v, i, w\}$ yields the pattern $P_b$.

At the end of the $i^{th}$ iteration of the main  \textbf{for} loop, if there exists a non detected $P_b$  pattern, $\alpha<_{\tau} \beta <_{\tau} i<_{\tau}\gamma$  with $\beta i \notin E(G)$.  Then $j <_{\tau} \beta <_{\tau} i<_{\tau}\gamma$ is also such a pattern.

As $\beta \in ACTIVE(]j,i[)$ the edge $\beta i$ has been checked in the comparison of the two lists $ACTIVE(]j,i[)$ and $N^-(i)$, a contradiction.

\textbf{Complexity:}
The preprocessing requires $O(n+m)$. The comparison of the two lists ordered with $\tau$,  $ACTIVE(]j,i[)$ and $N^-(i)$  ends at the first difference and can be charged to $N^-(i)$.
So in the \textbf{for} loop the neighbourhood of a vertex is only considered once. So the whole algorithm can be done in $O(n+m)$.\qed
\end{proof}

\begin{lemma}\label{Pab}
The pattern  $P_{ab}$ can be checked in $O(n \cdot m)$.
\end{lemma}

\begin{proof}
For this pattern, we also process a unique scan of the ordering $\tau$. The current vertex is analyzed as a potential vertex \textbf{3} in Figure \ref{fig:abcd-patterns}. We also maintain a list $ACTIVE$ of vertices that have a neighbour after or to the current vertex in the ordering. 

At each step $i$ we need to check that in $ACTIVE$, the neighbours and the non-neighbours of $i$ form a complete bipartite.
This test can  be done in linear time, which yields the complexity $O(n \cdot m)$.\qed
\end{proof}

For the pattern $P_{ab}$, it is not clear that a linear time algorithm exists, using the following reduction.

\begin{proposition}
The pattern $P_{ab}$ is as hard to recognize as the co-comparability pattern on three vertices.
\end{proposition}

\begin{proof}
Let G be a graph and $\tau$ an ordering on $V(G)$.
We associate a new graph $G'$ from $G$ by adding a universal vertex $u$.
We note that if $\tau'=\tau + u$ does not contain any $P_{ab}$, then $\tau$ does not contain the co-comparability pattern.\qed
\end{proof}

As noted in \cite{FeuilloleyH21}, there is no known linear time algorithm to detect the co-comparability pattern.


\end{document}